%% file: main.tex
\documentclass{article}

\usepackage{graphicx,amsmath,amsfonts,amscd,amssymb,bm,bbm,url,color,latexsym,fge,mathdots,mathtools,mathabx}
\usepackage[T1]{fontenc}
\usepackage{tgpagella}
\usepackage{fullpage}
\usepackage[small,bf]{caption}
\usepackage{mwe}
\usepackage{subcaption}
\usepackage{microtype}
\usepackage{algorithm}
\usepackage{algorithmicx}
\usepackage{algpseudocode}
\usepackage{hyperref}
\usepackage[nameinlink]{cleveref}
\usepackage{framed}
\usepackage{comment}
\usepackage{wrapfig}
\usepackage{tikz}
\usepackage{pgfplots}
\usepackage{authblk}
\usepackage{cite}
\usepackage{multirow}

\hypersetup{
    colorlinks=true,%
    citecolor=blue,%
    filecolor=blue,%
    linkcolor=blue,%
    urlcolor=blue
}

\usepackage{authblk}

\usepackage{tablefootnote}

\usepackage[toc, page]{appendix}
\usepackage{tcolorbox}

\usepackage{enumitem}
\def \endprf{\hfill {\vrule height6pt width6pt depth0pt}\medskip}
\newenvironment{proof}{\noindent {\bf Proof} }{\endprf\par}

\newtheorem{theorem}{Theorem}[section]
\newtheorem{lemma}[theorem]{Lemma}
\newtheorem{corollary}[theorem]{Corollary}
\newtheorem{proposition}[theorem]{Proposition}
\newtheorem{definition}[theorem]{Definition}

\renewcommand{\mathbf}{\boldsymbol}

\newcommand{\conv}{\circledast}
\newcommand{\cconv}{  \boxasterisk }
\newcommand{\mb}{\mathbf}
\newcommand{\mc}{\mathcal}

\newcommand{\bb}{\mathbb}

\newcommand{\reals}{\bb R}

\newcommand{\eps}{\varepsilon}

\newcommand{\R}{\reals}
\newcommand{\Cp}{\bb C}

\newcommand{\indicator}[1]{\mathbbm 1_{#1}}

\newcommand{ \Brac }[1]{\left\lbrace #1 \right\rbrace}
\newcommand{ \brac }[1]{\left[ #1 \right]}
\newcommand{ \paren }[1]{ \left( #1 \right) }

\DeclareMathOperator{\dist}{dist}

\DeclareMathOperator{\diag}{diag}
\DeclareMathOperator{\sign}{sign}
\DeclareMathOperator{\grad}{grad}

\newcommand{\wh}{\widehat}
\newcommand{\wc}{\widecheck}
\newcommand{\wt}{\widetilde}
\newcommand{\ol}{\overline}

\newcommand{\norm}[2]{\left\| #1 \right\|_{#2}}

\newcommand{\abs}[1]{\left| #1 \right|}
\newcommand{\innerprod}[2]{\left\langle #1,  #2 \right\rangle}
\newcommand{\prob}[1]{\bb P\left[ #1 \right]}
\newcommand{\expect}[1]{\bb E\left[ #1 \right]}

\newcommand{\parans}[1]{\left(#1\right)}

\newcommand{\shift}[2]{{s_{#2}}\left[#1\right]}

\newcommand{\revised}[1]{{\color{black}{#1}}}
\newcommand{\edited}[1]{{\color{black}{#1}}}

\usepackage{pifont}
\newcommand{\xmark}{\ding{55}}%
\title{A Nonconvex Approach for Exact and Efficient Multichannel Sparse Blind Deconvolution}

\pgfplotsset{compat=1.14}

\author[$\sharp$]{Qing Qu}
\author[$\dagger$]{Xiao Li}
\author[$\Diamond$]{Zhihui Zhu}

\affil[$\sharp$]{Center for Data Science, New York University}
\affil[$\dagger$]{Department of Electronic Engineering, the Chinese University of Hong Kong}
\affil[$\Diamond$]{Mathematical Institute for Data Science, Johns Hopkins University}

\begin{document}

\maketitle

\begin{abstract}
    We study the multi-channel sparse blind deconvolution (MCS-BD) problem, whose task is to simultaneously recover a kernel $\mathbf a$ and multiple sparse inputs $\{\mathbf x_i\}_{i=1}^p$ from their circulant convolution $\mathbf y_i = \mathbf a \circledast \mathbf x_i $ ($i=1,\cdots,p$). We formulate the task as a nonconvex optimization problem over the sphere. Under mild statistical assumptions of the data, we prove that the vanilla Riemannian gradient descent (RGD) method, with random initializations, provably recovers both the kernel $\mathbf a$ and the signals $\{\mathbf x_i\}_{i=1}^p$ up to a signed shift ambiguity. In comparison with state-of-the-art results, our work shows significant improvements in terms of sample complexity and computational efficiency. Our theoretical results are corroborated by numerical experiments, which demonstrate superior performance of the proposed approach over the previous methods on both synthetic and real datasets. 
\end{abstract}


\section{Introduction}\label{sec:intro}
\input{sec/intro}

\section{Problem Formulation}\label{sec:problem}
\input{sec/problem}

\section{Main Results and Analysis}\label{sec:main}
\input{sec/analysis}


\section{Experiment}\label{sec:exp}
\input{sec/experiment}

\section{Discussion and Conclusion}\label{sec:discuss}

\input{sec/discussion}

\section*{Acknowledgement}
Part of this work is done when QQ, XL and ZZ were attending "Computational Imaging" workshop at ICERM Brown in Spring 2019. We would like to thank the National Science Foundation under Grant No. DMS-1439786 for the generous support of participating in this workshop. We would like to thank Shuyang Ling (NYU Shanghai), Carlos Fernandez-Granda (NYU Courant), Yuxin Chen (Princeton), Yuejie Chi (CMU), and Pengcheng Zhou (Columbia U.) for fruitful discussions. QQ also would like to acknowledge the support of Microsoft PhD fellowship, and Moore-Sloan foundation fellowship. XL would like to acknowledge the support by Grant CUHK14210617 from the Hong Kong Research Grants Council. ZZ was partly supported by NSF Grant 1704458.

\newpage

{\small
\bibliographystyle{alpha}
\bibliography{cdl,ncvx,deconv,nips}
}

\newpage
\appendices

The appendices are organized as follows. In Appendix \ref{app:notation} we introduce the basic notations and problem reductions that are used throughout the main draft and the appendix. We list the basic technical tools and results in Appendix \ref{app:basic}. In Appendix \ref{app:geometry-main} we describe and prove the main geometric properties of the optimization landscape for Huber loss. In Appendix \ref{app:convergence}, we provide global convergence analysis for the propose Riemannian gradient descent methods for optimizing the Huber loss, and the subgradient methods for solving LP rounding.  All the technical geometric analysis are postponed to Appendix \ref{app:regularity-population}, Appendix \ref{app:curv-population}, Appendix \ref{app:gradient-concentration}, and Appendix \ref{app:precond}. Finally, in Appendix \ref{app:algorithm} we describe the proposed optimization algorithms in full details for all $\ell^1$, Huber, and $\ell^4$ losses.

\section{Basic Notations and Problem Reductions}\label{app:notation}

\input{sec/app_notation}

\section{Basic Tools}\label{app:basic}

\input{sec/app_basic}

\section{Geometry: Main Results}\label{app:geometry-main}

\input{sec/app_main_geometry}

\section{Convergence Analysis}\label{app:convergence}

\input{sec/app_convergence}

\section{Regularity Condition in Population}\label{app:regularity-population}

\input{sec/app_geometry}

\section{Implicit Regularization in Population}\label{app:curv-population}\label{app:implicit-population}

\input{sec/app_stay}

\section{Gradient Concentration}\label{app:gradient-concentration}

\input{sec/app_concentration}

\section{Preconditioning}\label{app:precond}

\input{sec/app_preconditioning}

\section{Algorithms and Implementation Details}\label{app:algorithm}

\input{sec/app_algorithm}




\end{document}

%% file: sec/intro.tex
We study the blind deconvolution problem with multiple inputs: given \emph{circulant} convolutions
\begin{align}\label{eqn:bd-measurement}
   \mb y_i \;=\; \mb a \conv \mb x_i \; \in \bb R^n ,\qquad i \in [p]:= \{1,\ldots,p\},	
\end{align}
we aim to recover both the kernel $\mb a \in \bb R^n$ and the signals $\Brac{\mb x_i}_{i=1}^p \in \bb R^n $ using efficient methods. Blind deconvolution is an \emph{ill-posed} problem in its most general form. Nonetheless, problems in practice often exhibits \emph{intrinsic} low-dimensional structures, showing promises for efficient optimization. One such useful structure is the \emph{sparsity} of the signals $\Brac{\mb x_i}_{i=1}^p$. The multichannel sparse blind deconvolution (MCS-BD) broadly appears in the context of communications \cite{amari1997multichannel,tian2017multichannel}, computational imaging \cite{betzig2006imaging,she2015image}, seismic imaging \cite{kaaresen1998multichannel,nose2015fast,repetti2015euclid}, neuroscience \cite{gitelman2003modeling,ekanadham2011blind,wu2013blind,friedrich2017fast,pnevmatikakis2016simultaneous}, computer vision \cite{levin2011understanding,zhang2013multi,sroubek2012robust}, and more.
\begin{itemize}[leftmargin=*]
\item \textbf{Neuroscience.} Detections of neuronal spiking activity is a prerequisite for understanding the mechanism of brain function. Calcium imaging \cite{friedrich2017fast,pnevmatikakis2016simultaneous} and functional MRI \cite{gitelman2003modeling,wu2013blind} are two widely used techniques, which record the convolution of unknown neuronal transient response and sparse spike trains. The spike detection problem can be naturally cast as a MCS-BD problem.
\item \textbf{Computational (microscopy) imaging.} Super-resolution fluorescent microscopy imaging \cite{betzig2006imaging,hess2006ultra,rust2006sub} conquers the resolution limit by solving sparse deconvolution problems. Its basic principle is using photoswitchable fluorophores that stochastically activate fluorescent molecular, creating a video sequence of sparse superpositions of point spread function (PSF). In many scenarios (especially in 3D imaging), as it is often difficult to obtain the PSF due to defocus and unknown aberrations \cite{sarder2006deconvolution}, it is preferred to estimate the point-sources and PSF jointly by solving MCS-BD.
\item \textbf{Image deblurring.} Sparse blind deconvolution problems also arise in natural image processing: when a blurry image is taken due to the resolution limit or malfunction of imaging procedure, it can be modeled as a blur pattern convolved with visually plausible sharp images (whose gradient are sparse) \cite{zhang2013multi,sroubek2012robust}.
\end{itemize}

\begin{center}
\setlength{\arrayrulewidth}{0.4mm}
\setlength{\tabcolsep}{12pt}
\renewcommand{\arraystretch}{1.3}
 \begin{table*}  
 \caption{Comparison with existing methods for solving MCS-BD\tablefootnote{}}\label{tab:comparison}
 \resizebox{\textwidth}{!}{
 \begin{tabular}{c||c|c|c}
 \hline
 Methods & Wang et al.\cite{wang2016blind} &  Li et al. \cite{li2018global} & \textbf{Ours}\\ 
 \hline
 \multirow{ 2}{*}{Assumptions} &  $\mb a$ spiky $\&$ invertible, & $\mb a$ invertible,  & $\mb a$ invertible,  \\ 
  &$\mb x_i \sim_{i.i.d.} \mc {BG}(\theta)$ & $\mb x_i \sim_{i.i.d.} \mc {BR}(\theta)$ & $\mb x_i \sim_{i.i.d.} \mc {BG}(\theta)$ \\ 
\hline
 Formulation & $\min_{ q_1 =1 } \norm{ \mb C_{\mb q} \mb Y}{1}$ & $\max_{\mb q\in \bb S^{n-1}} \norm{ \mb C_{\mb q} \mb P\mb Y}{4}^4$ & $\min_{\mb q\in \bb S^{n-1}} H_\mu\paren{ \mb C_{\mb q} \mb P\mb Y}$\\
   \hline 
  Algorithm & interior point & \emph{noisy} RGD & \emph{vanilla} RGD \\
 \hline
 \multirow{ 2}{*}{ Recovery  }
 & $\theta \in \mc O(1/\sqrt{n})$, & $\theta \in \mc O(1)$,  & $\theta \in \mc O(1)$,  \\ 
 Condition& $p \geq \wt{\Omega}(n)$ & $p \geq \wt{\Omega}(\max\Brac{n,\kappa^8}\frac{n^8}{\eps^8})$ &$p \geq \wt{\Omega}(\max\Brac{n,\frac{\kappa^8}{\mu^2}}n^4)$ \\
  \hline 
   Time Complexity  & $\wt{\mc O}( p^4n^5 \log(1/\eps) )$ & $\wt{\mc O}(pn^{13}/\eps^8  )$  & $\wt{\mc O}(pn^5+ pn\log \paren{1/\eps}  )$ \\
 \hline
\end{tabular}
}
\end{table*}
\end{center}
\footnotetext{Here, (i) $\mc {BG}(\theta)$ and $\mc {BR}(\theta)$ denote Bernoulli-Gaussian and Bernoulli-Rademacher distribution, respectively; (ii) $\theta\in [0,1] $ is the Bernoulli parameter controlling the sparsity level of $\mb x_i$; (iii) $\eps$ denotes the recovery precision of global solution $\mb a_\star$, i.e., $\min_{\ell} \norm{\mb a - \shift{\mb a_\star}{\ell} }{}\leq \eps$; (iv) $\wt{\mc O}$ and $\wt{\Omega}$ hides $\log(n)$, $\theta$ and other factors. For \cite{wang2016blind}, we may get rid of the spiky assumption by solving a preconditioned problem $\min_{q_1 = 1 } \norm{ \mb C_{\mb q} \mb P \mb Y}{1}$,  where $\mb P$ is a preconditioning matrix defined in \eqref{eqn:precond-mtx}. }

\paragraph{Prior arts on MCS-BD.} Recently, there have been a few attempts to solve MCS-BD with guaranteed performance. Wang et al. \cite{wang2016blind} formulated the task as finding the sparsest vector in a subspace problem \cite{qu2014finding}. They considered a convex objective, showing that the problem can be solved to exact solutions when $p\geq \Omega(n \log n)$ and the sparsity level $\theta \in \mc O(1/\sqrt{n})$. \edited{A similar approach has also been investigated by \cite{cosse2017note}}. Li et al. \cite{li2018global} consider a nonconvex $\ell^4$-maximization problem over the sphere\footnote{Recently, similar loss has been considered for short and sparse deconvolution \cite{zhang2018structured} and complete dictionary learning \cite{zhai2019complete}.} , revealing benign global geometric structures of the problem. Correspondingly, they introduced a \emph{noisy} Riemannian gradient descent (RGD) that solves the problem to approximate solutions in polynomial time.

These results are very inspiring but still suffer from quite a few limitations. The theory and method in \cite{wang2016blind} \emph{only} applies to cases when $\mb a$ is approximately a delta function (which excludes most problems of interest) and $\Brac{\mb x_i}_{i=1}^p$ are \emph{very} sparse. Li et al. \cite{li2018global} suggests that more generic kernels $\mb a$ can be handled via preconditioning of the data. However, due to the \emph{heavy-tailed} behavior of $\ell^4$-loss, the sample complexity provided in \cite{li2018global} is quite \emph{pessimistic}\footnote{As the tail of $\mc {BG}(\theta)$ distribution is heavier than that of $\mc {BR}(\theta)$, their sample complexity would be even worse if $\mc {BG}(\theta)$ model was considered.}. Moreover, noisy RGD is proved to converge with huge amounts of iterations \cite{li2018global}, and it requires additional efforts to tune the noise parameters which is often unrealistic in practice. As mentioned in \cite{li2018global}, one may use vanilla RGD which almost surely converges to a global minimum, but without guarantee on the number of iterations. On the other hand, \edited{Li et al.} \cite{li2018global} only considered the Bernoulli-Rademacher model\footnote{We say $\mb x$ obeys a Bernoulli-Rademacher distribution when $\mb x \;=\; \mb b \odot \mb r$ where $\odot$ denotes point-wise product, $\mb b$ follows i.i.d. Bernoulli distribution and $\mb r$ follows i.i.d. Rademacher distribution.} which is restrictive for many problems.

\paragraph{Contributions of this paper.} In this work, we introduce an efficient optimization method for solving MCS-BD. We consider a natural nonconvex formulation based on a smooth relaxation of $\ell^1$-loss. Under mild assumptions of the data, we prove the following result. 
\begin{framed}
With \emph{random} initializations, a \emph{vanilla} RGD efficiently finds an approximate solution, which can be refined by a subgradient method that converges exactly to the target solution in a \emph{linear} rate.
\end{framed}
We summarize our main result in \Cref{tab:comparison}. By comparison\footnote{We do not find a direct comparison with \cite{wang2016blind} meaningful, mainly due to its limitations of the kernel assumption and sparsity level $\theta \in \mc O(1/\sqrt{n})$ discussed above.} with \cite{li2018global}, our approach demonstrates \emph{substantial} improvements for solving MCS-BD in terms of both sample and time complexity. Moreover, our experimental results imply that our analysis is  still far from tight -- the phase transitions suggest that $p \geq \Omega( \mathrm{poly}\log(n) )$ samples might be sufficient for exact recovery, which is favorable for applications (as real data in form of images can have millions of pixels, resulting in huge dimension $n$). Our analysis is inspired by recent results on orthogonal dictionary learning \cite{gilboa2018efficient,bai2018subgradient}, but much of our theoretical analysis is tailored for MCS-BD with a few extra new ingredients. Our work is the first result provably showing that \emph{vanilla} gradient descent type methods \edited{with random initialization} solve MCS-BD efficiently. Moreover, our ideas could potentially lead to new algorithmic guarantees for other nonconvex problems such as blind gain and phase calibration \cite{li2017identifiability,ling2018self} and convolutional dictionary learning \cite{bristow2013fast,garcia2018convolutional}. 

\edited{\paragraph{Organizations, notations, and reproducible research.} We organize the rest of the paper as follows. In \Cref{sec:problem}, we introduce the basic assumptions and nonconvex problem formulation. \Cref{sec:main} presents the main results and sketch of analysis. In \Cref{sec:exp}, we demonstrate the proposed approach by experiments on both synthetic and real datasets. We conclude the paper in \Cref{sec:discuss}. The basic notations are introduced in Appendix \ref{app:notation}, and all the detailed analysis are deferred to the appendices. For reproducing the experimental results in this work, we refer readers to 
\begin{center}
\color{blue}{\url{https://github.com/qingqu06/MCS-BD}}.
\end{center}

}

%% file: sec/problem.tex

\subsection{Assumptions and Intrinsic Properties}
\paragraph{Assumptions} To begin, we list our assumptions on the kernel $\mb a \in \bb R^n $ and sparse inputs $\Brac{\mb x_i}_{i=1}^p \in \bb R^n$:
\begin{enumerate}[leftmargin=*]
\item \emph{Invertible kernel.} We assume the kernel $\mb a$ to be \emph{invertible} \edited{in the sense that its spectrum $\wh {\mb a} = \mb F\mb a$ does not have zero entries, where $\wh {\mb a} = \mb F\mb a$ is the discrete Fourier transform (DFT) of $\mb a$ with $\mb F\in\Cp^{n\times n}$ being the DFT matrix. \revised{Let $\mb C_{\mb a}\in \bb R^{n \times n}$ be an $n\times n$ circulant matrix whose first column is $\mb a$; see \eqref{eqn:circulant matrx constrcut} for the formal definition. Since this circulant matrix} $\mb C_{\mb a} $ can be decomposed as $\mb C_{\mb a} = \mb F^* \diag\paren{ \wh{\mb a} } \mb F$ \cite{gray2006toeplitz}, it is also invertible and we define its condition number}
  \begin{align*}
  	\kappa(\mb C_{\mb a}) \;:=\; \max_{i}\abs{\wh a_i}/ \min_{i}\abs{\wh a_i}.
  \end{align*}
\item \emph{Random sparse signal.} We assume the input signals $\Brac{\mb x_i}_{i=1}^p$ follow i.i.d. Bernoulli-Gaussian ($\mc {BG}(\theta)$) distribution:
\begin{align*}
    \mb x_i \;=\; \mb b_i \odot \mb g_i, \qquad \mb b_i \sim_{i.i.d.} \mc B(\theta), \quad \mb g_i \sim_{i.i.d.} \mc N(\mb 0,\mb I),
\end{align*}
where $\theta\in [0,1]$ is the Bernoulli-parameter which controls the sparsity level of each $\mb x_i$.
\end{enumerate}
As aforementioned, this assumption generalizes those used in \cite{wang2016blind,li2018global}. In particular, the first assumption on kernel $\mb a$ is much more practical than that of \cite{wang2016blind}, in which $\mb a$ is assumed to be approximately a delta function. The second assumption is a generalization of the Bernoulli-Rademacher model adopted in \cite{li2018global}.

\paragraph{Intrinsic symmetry.} Note that the MCS-BD problem exhibits intrinsic \emph{signed scaling-shift} symmetry, i.e.,  for any $\alpha \not = 0$, 
\edited{\begin{align}\label{eqn:symmetry-ambiguity}
	\mb y_i \;=\; \mb a \; \conv \;\mb x_i \;=\; \shift{ \pm \alpha \mb a }{-\ell} \;\conv\; \shift{ \pm \alpha^{-1} \mb x_i }{\ell}, \qquad i \in \Brac{0,1,\cdots,p-1},
\end{align} }
where $\shift{ \cdot }{\ell}$ denotes a cyclic shift operator of length $\ell$. Thus, we only hope to recover $\mb a$ and $\Brac{\mb x_i}_{i=1}^p$ up to a \emph{signed shift ambiguity}. Without loss of generality, for the rest of the paper we assume that the kernel $\mb a$ is normalized with $\norm{\mb a}{} = 1$. 


\subsection{A Nonconvex Formulation}
Let $ \mb Y \;=\; \begin{bmatrix}\mb y_1 & \mb y_2 & \cdots & \mb y_p \end{bmatrix} $ and $	\mb X \;=\;\begin{bmatrix} \mb x_{1} & \mb x_{2} & \cdots & \mb x_{p} \end{bmatrix}$. We can rewrite the measurement \eqref{eqn:bd-measurement} in a matrix-vector form via circulant matrices,
\begin{align*}
   \mb y_i \;=\; \mb a \conv \mb x_i	 \;=\; \mb C_{\mb a} \mb x_i, \;\; i \in [p] \quad \;\Longrightarrow \; \quad \mb Y \;=\; \mb C_{\mb a} \mb X,
\end{align*}
Since $\mb C_{\mb a}$ is assumed to be invertible, we can define its corresponding \emph{inverse kernel} $\mb h \in \bb R^n$ by $\mb h := \mb F^{-1} \wh {\mb a}^{\odot-1}$ whose corresponding circulant matrix satisfies
\begin{align*}
   \mb C_{\mb h}\;:=\; \mb F^* \diag\paren{  \wh{\mb a}^{\odot -1}  }  \mb F \;=\;  \mb C_{\mb a}^{-1},
\end{align*}
where $(\cdot)^{\odot-1}$ \edited{denotes} entrywise inversion. Observing 
\edited{\begin{align*}
	\mb C_{\mb h} \cdot \mb Y \;=\;  \underbrace{\mb C_{\mb h} \cdot  \mb C_{\mb a} }_{=\;\mb I} \cdot  \mb X\;=\; \underbrace{\mb X}_{\text{sparse}},
\end{align*}}
it leads us to consider the following objective
\begin{align}\label{eqn:formulation-bd-l0}
   \min_{\mb q}\; \frac{1}{np} \norm{ \mb C_{\mb q} \mb Y }{0} = \frac{1}{np} \sum_{i=1}^p \norm{ \mb C_{\mb y_i} \mb q }{0} , \qquad \text{s.t.} \quad \mb q \neq \mb 0.
\end{align}
Obviously, when the solution of \eqref{eqn:formulation-bd-l0} is unique, the \emph{only} minimizer is the inverse kernel $\mb h$ up to signed scaling-shift (i.e., $\mb q_\star = \pm \alpha \shift{\mb h }{\ell}$), producing $\mb C_{\mb h} \mb Y=\mb X$ with the highest sparsity. The nonzero constraint $\mb q \neq \mb 0$ is enforced simply to prevent the trivial solution $\mb q= \mb 0$. \edited{Ideally, if we could solve \eqref{eqn:formulation-bd-l0} to obtain one of the target solutions $\mb q_\star = \shift{\mb h }{\ell}$ up to a signed scaling, the kernel $\mb a$ and sparse signals $\Brac{\mb x_i}_{i=1}^p$ can be exactly recovered up to signed shift via 
\begin{align*}
	\mb a_\star \;=\; \mb F^{-1}  \brac{ \paren{\mb F  \mb q_\star }^{\odot -1} },\qquad \mb x_i^\star \;=\; \mb C_{\mb y_i} \mb q_\star,\;(1\leq i \leq p).
\end{align*}
}
However, it has been known for decades that optimizing the basic $\ell_0$-formulation \eqref{eqn:formulation-bd-l0} is an NP-hard problem \cite{coleman1986null,natarajan1995sparse}. Instead, we consider the following \emph{nonconvex}\footnote{It is nonconvex because of the spherical constraint $\mb q \in \bb S^{n-1}$.} relaxation of the original problem \eqref{eqn:formulation-bd-l0}:
\begin{align}\label{eqn:problem}
  \boxed{ \min_{\mb q}\; \varphi_h(\mb q) \;:=\; \frac{1}{np} \sum_{i=1}^p H_\mu \paren{ \mb C_{\mb y_i} \mb P \mb q }, \qquad \text{s.t.} \quad \mb q \in \bb S^{n-1},}
\end{align}
where $H_\mu(\cdot)$ is the Huber loss \cite{huber1992robust} and $\mb P$ is a preconditioning matrix, both of which will be defined and discussed in detail as follows.

\paragraph{Smooth sparsity surrogate.} It is well-known that $\ell^1$-norm serves as a natural sparsity surrogate for $\ell^0$-norm, but its nonsmoothness often makes it difficult for analysis\footnote{\edited{The subgradient of $\ell^1$-loss is \emph{non-Lipschitz}, which introduces tremendous difficulty in controlling suprema of random process and perturbation analysis for preconditioning}}. Here, we consider the Huber loss\footnote{Actually, $h_\mu(\cdot)$ is a scaled and elevated version of \edited{the} standard Huber function $h_{\mu}^s\left(z\right)$, with $h_{\mu}\left(z\right) = \frac{1}{\mu} h_{\mu}^s\left(z\right) + \frac{\mu}{2}$. Hence in our framework minimizing with $h_{\mu}\left(z\right)$ is equivalent to minimizing with $h_{\mu}^s\left(z\right)$. } $H_\mu \paren{\cdot }$ \edited{which is widely used in robust optimization \cite{huber1992robust}. It} acts as a \emph{smooth} sparsity surrogate of $\ell^1$ penalty and is defined as:
\begin{equation}\label{eqn:huber-loss}
	H_\mu (\mb Z) \;:=\; \sum_{i=1}^n \sum_{j=1}^p h_\mu (Z_{ij}),\qquad h_{\mu}\paren{z} \;:=\; 
\begin{cases}
\abs{z} & \abs{z} \geq \mu \\
\frac{z^2}{2\mu} + \frac{\mu}{2}  & \abs{z} < \mu
\end{cases},
\end{equation}
where $\mu>0$ is a smoothing parameter. Our choice $h_{\mu}\left(z\right)$ is first-order smooth, and behaves exactly same as the $\ell^1$-norm for all $\abs{z} \geq \mu$. \edited{In contrast, although the $\ell^4$ objective in \cite{li2018global} is smooth, it only promotes sparsity in special cases. Moreover,} it results in a heavy-tailed process, producing flat landscape around target solutions, and requiring substantially more samples for measure concentration. \Cref{fig:landscape-original} shows a comparison of optimization landscapes of all losses in low dimension: \revised{the Huber-loss produces an almost identical landscape as the $\ell^1$-loss}, while optimizing the $\ell^4$-loss could result in large approximation error.

\begin{figure}
	\captionsetup[sub]{font=small,labelfont={bf,sf}}
	\centering
	\begin{minipage}[c]{0.3\textwidth}
		\subcaption{$\ell^1$-loss, \xmark}
		\centering
		\includegraphics[width = 1.25in]{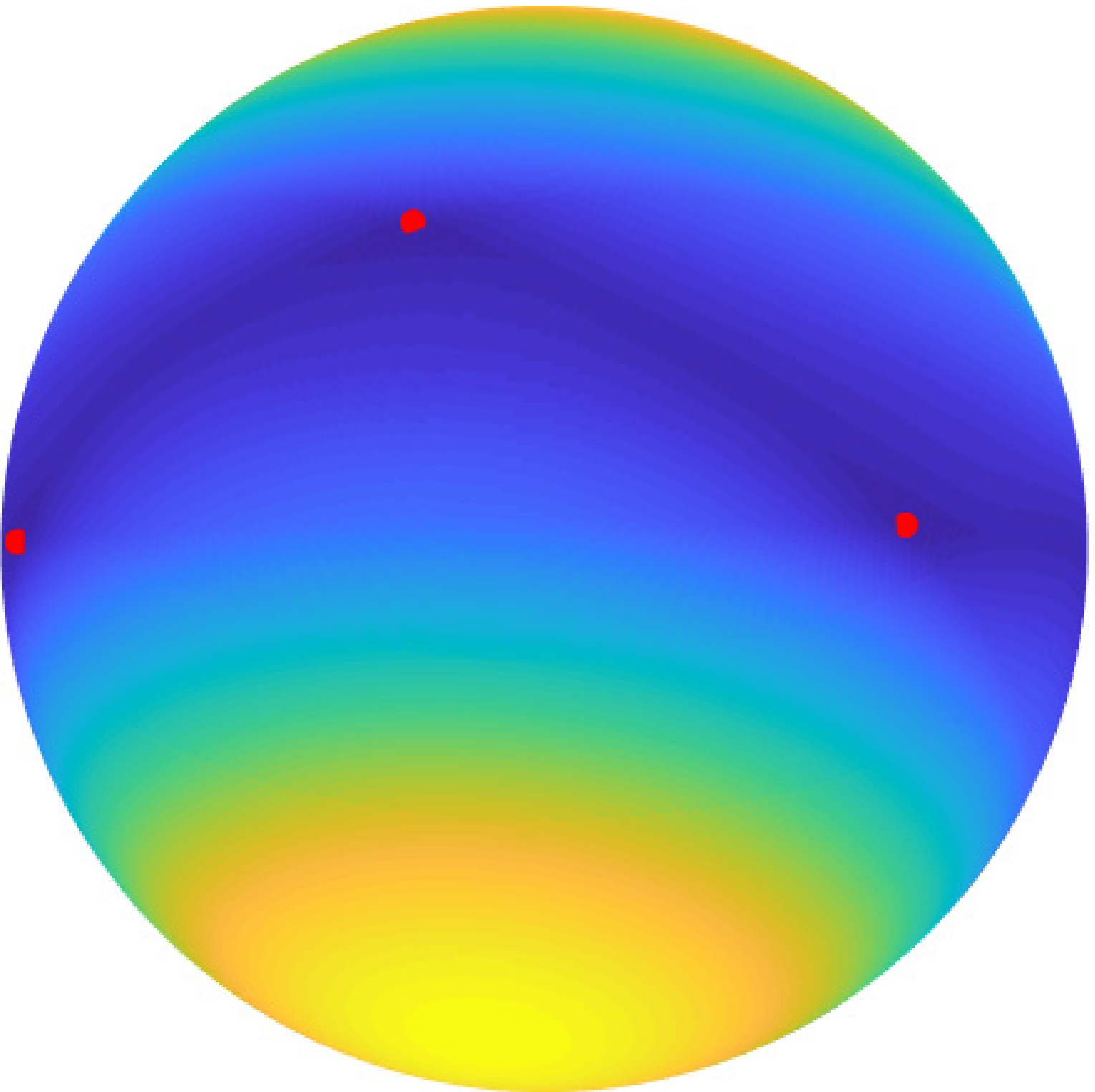}
	\end{minipage}
	\begin{minipage}[c]{0.3\textwidth}
		\subcaption{Huber-loss, \xmark}
		\centering
		\includegraphics[width = 1.25in]{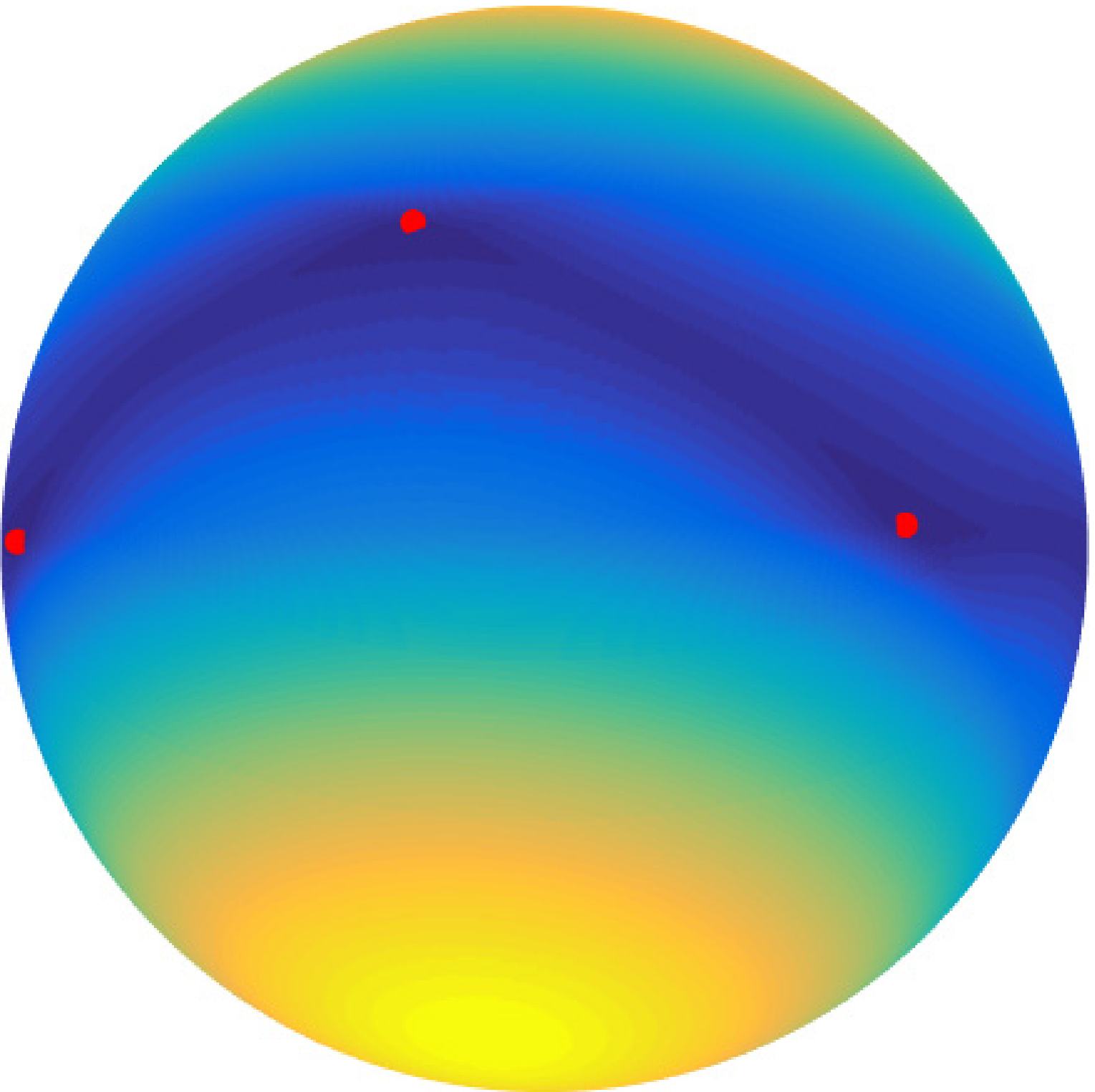}
	\end{minipage}
	\begin{minipage}[c]{0.3\textwidth}
		\subcaption{$\ell^4$-loss, \xmark}
		\centering
		\includegraphics[width = 1.25in]{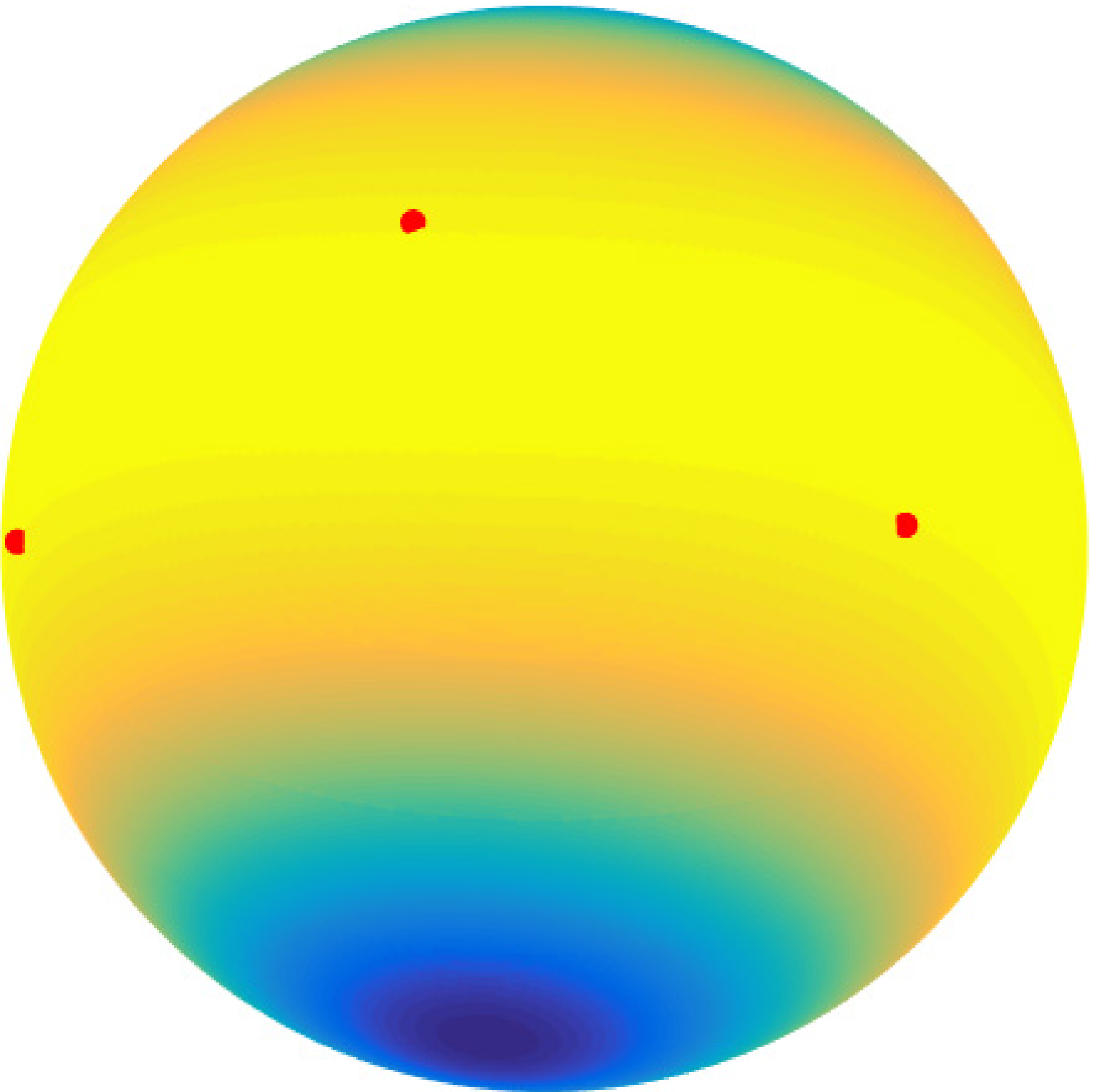}
	\end{minipage}
	
	\begin{minipage}[c]{0.3\textwidth}
		\subcaption{$\ell^1$-loss, $\checkmark$}
		\centering
		\includegraphics[width = 1.25in]{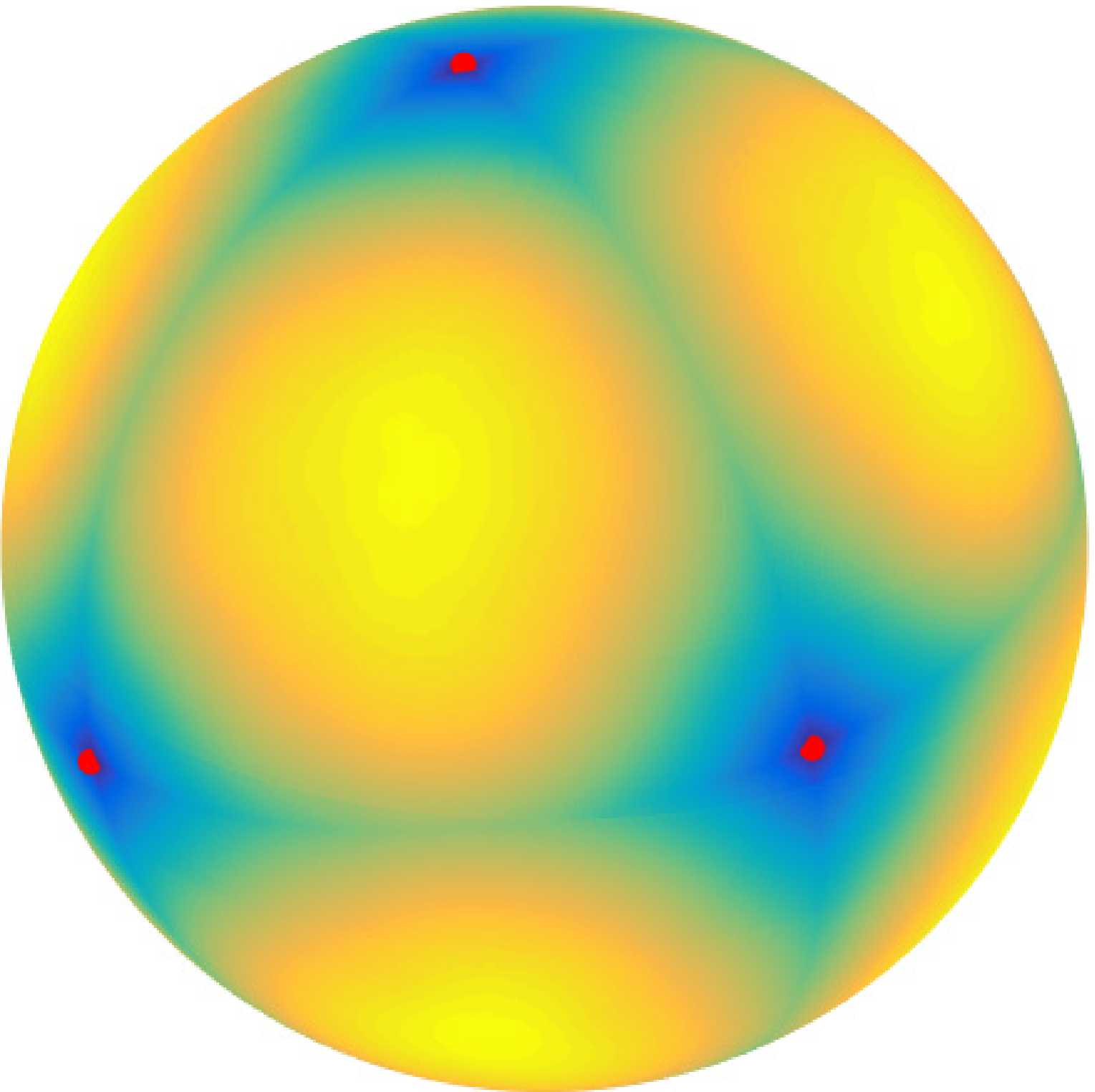}
	\end{minipage}
	\begin{minipage}[c]{0.3\textwidth}
		\subcaption{Huber-loss, $\checkmark$}
		\centering
		\includegraphics[width = 1.25in]{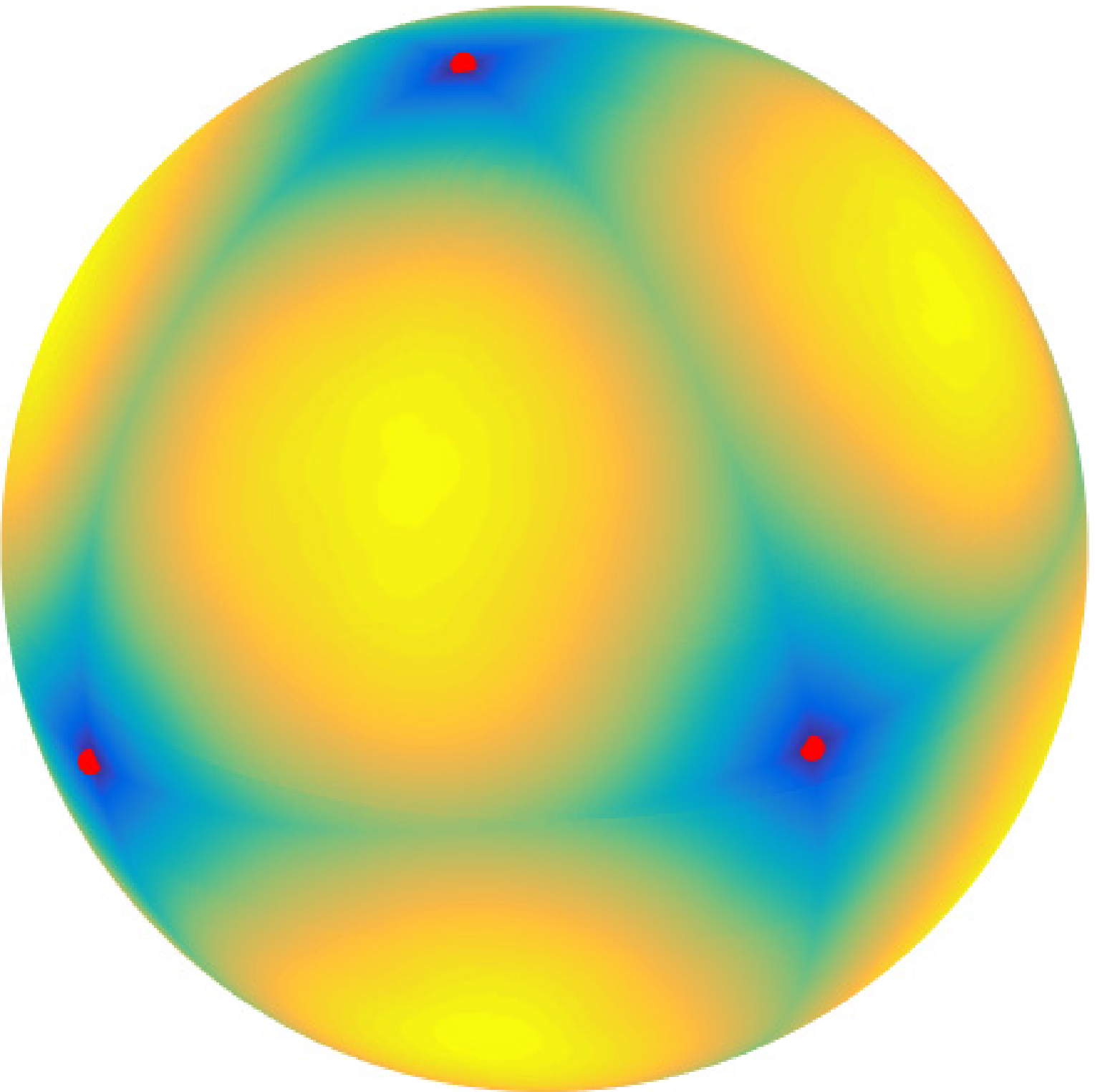}
	\end{minipage}
	\begin{minipage}[c]{0.3\textwidth}
		\subcaption{$\ell^4$-loss, $\checkmark$}
		\centering
		\includegraphics[width = 1.25in]{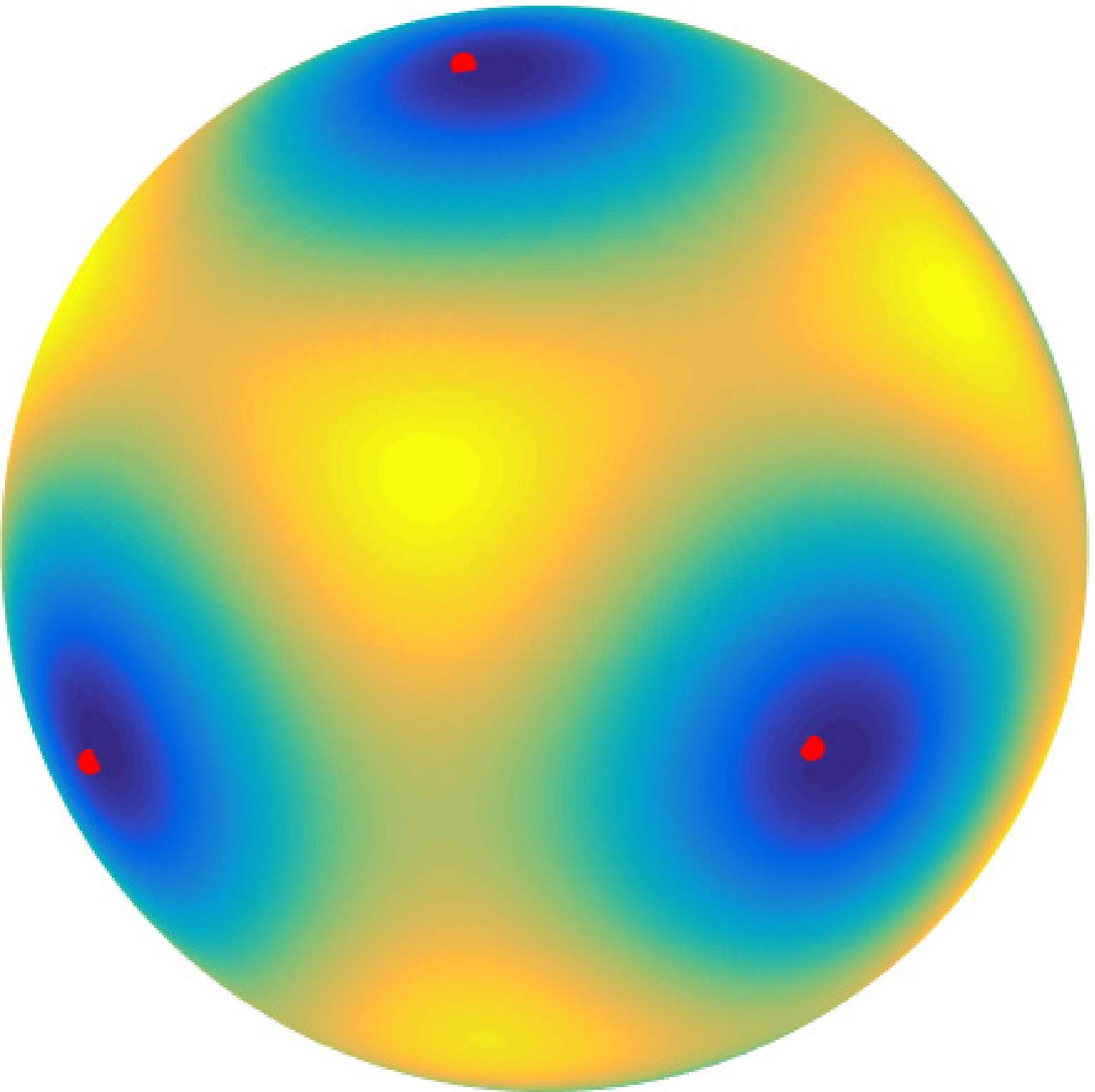}
	\end{minipage}
\caption{\textbf{Comparison of optimization landscapes for different loss functions.} Here \xmark~  and $\checkmark$ mean without and with the preconditioning matrix $\mb P$, respectively. Each figure plots the function values of the loss over $\bb S^2$, where the function values are all normalized between $0$ and $1$ (darker color means smaller value, and vice versa). The small red dots on the landscapes denote shifts of the ground truths.
}
\label{fig:landscape-original}
\end{figure}
\paragraph{Preconditioning.} \edited{An ill-conditioned kernel $\mb a$ can result in poor optimization landscapes (see \Cref{fig:landscape-original} for an illustration). To alleviate this effect, we introduce a preconditioning matrix $\mb P\in\R^{n\times n}$~\cite{sun2016complete,zhang2018structured,li2018global}, defined as follows\footnote{Here, the sparsity $\theta$ serves as a normalization purpose. It is often not known ahead of time, but the scaling here does not change the optimization landscape.}}
\begin{align}\label{eqn:precond-mtx}
   \mb P \;=\;  \paren{ \frac{1}{\theta np  } \sum_{i=1}^p \mb C_{\mb y_i}^\top \mb C_{\mb y_i} }^{-1/2},
\end{align}
which refines the optimization landscapes by orthogonalizing the circulant matrix $\mb C_{\mb a}$ as
\begin{align}
  \underbrace{ \mb C_{\mb a}\mb P}_{\mb R} \;\approx\; \underbrace{ \mb C_{\mb a} \paren{\mb C_{\mb a}^\top \mb C_{\mb a}}^{-1/2} }_{\mb Q \text{ orthogonal} }.
\label{eq:why precond}
\end{align}
Since \edited{$\mb P \approx \paren{\mb C_{\mb a}^\top \mb C_{\mb a}}^{-1/2}$}, $\mb R$ can be proved to be very close to the orthogonal matrix $\mb Q$. Thus, $\mb R$ is much more well-conditioned than $\mb C_{\mb a}$. As illustrated in \Cref{fig:landscape-original}, a comparison of optimization landscapes without and with preconditioning shows that preconditioning symmetrifies the optimization landscapes and eliminates \emph{spurious} local minimizers. Therefore, it makes the problem more amendable \edited{to optimization algorithms}.
\paragraph{Constrain over the sphere $\bb S^{n-1}$.} \edited{We relax the nonconvex constraint $\mb q\not =\mb 0$ in \eqref{eqn:formulation-bd-l0} by a unit norm constraint on $\mb q$.} The norm constraint removes the scaling ambiguity as well as prevents the trivial solution $\mb q = \mb 0$. Note that the choice of \edited{the norm} has strong implication for computation. When $\mb q$ is constrained over $\ell^\infty$-norm, the $\ell^1/\ell^\infty$ optimization problem breaks beyond sparsity level $\theta \geq \Omega(1/\sqrt{n})$ \cite{wang2016blind}. In contrast, \edited{the sphere $\bb S^{n-1}$} is a smooth homogeneous Riemannian manifold and it has been shown recently that optimizing over the sphere leads to optimal sparsity $\theta \in  \mc O(1)$ for several sparse learning problems \cite{qu2014finding,sun2016complete,sun2017complete,li2018global}. Therefore, we choose to work with a nonconvex spherical constraint $\mb q\in \bb S^{n-1}$ \revised{and we will also show similar results for MCS-BD.}

Next, we develop efficient first-order methods and provide guarantees for exact recovery.

%% file: sec/analysis.tex
In this section, we show that the underlying benign \emph{first-order geometry} of the optimization landscapes of \Cref{eqn:problem} \emph{enables} efficient and exact recovery using \emph{vanilla} gradient descent methods, even with \emph{random} initialization. Our main result can be captured by the following theorem, \revised{with details described in the following subsections.}
\begin{theorem}\label{thm:Main-all}
We assume that the kernel $\mb a$ is invertible with condition number $\kappa$, and $\Brac{\mb x_i}_{i=1}^p \sim \mc {BG}(\theta)$. Suppose $\theta \in \paren{\frac{1}{n} , \frac{1}{3} }$ and $\mu \leq c\min\Brac{ \theta, \frac{1}{\sqrt{n}} } $. Whenever
\begin{align}\label{eqn:sample-complexity-p}
  p \geq C \max\Brac{n , \frac{\kappa^8 }{ \theta \mu^2 \sigma_{\min}^2 } \log^4 n } \theta^{-2} n^4 \log^3(n) \log \paren{ \frac{ \theta n }{ \mu} },
\end{align}
w.h.p. the function \eqref{eqn:problem} satisfies certain regularity conditions \revised{(see \Cref{prop:regularity-main}), allowing us to design an efficient vanilla first-order method. In particular, with probability at least $\frac{1}{2}$, by using a random initialization, the algorithms provably recover} the target solution up to a signed shift with $\eps$-precision in a linear rate
\vspace{-0.05in}
   \begin{align*}
   	  \# Iter \;\leq \; C'\paren{ \theta^{-1}n^4 \log \paren{ \frac{1}{\mu} } + \log(np) \log\paren{ \frac{1}{\eps}} }.
   \end{align*}
\end{theorem}

\paragraph{Remark 1.} The detailed proofs are detained to Appendix \ref{app:geometry-main} and Appendix \ref{app:convergence}. In the following, we explain our results in several aspects.
\begin{itemize}[leftmargin=*]
	\item \emph{Conditions and Assumptions.} Here, as the MCS-BD problem becomes trivial\footnote{\edited{The problem becomes trivial when $\theta \leq 1/n$ because $\theta n= 1$ so that each $\mb x_i$ tends to be an one sparse $\delta$-function.}} when $\theta \leq {1}/{n}$, we only focus on the regime $\theta > {1}/{n}$. Similar to \cite{li2018global}, our result only requires the kernel $\mb a$ \edited{to be invertible} and sparsity level $\theta$ to be constant. In contrast, the method in \cite{wang2016blind} only works when the kernel $\mb a$ is spiky and $\Brac{\mb x_i}_{i=1}^p$ are very sparse $\theta \in \mc O(1/\sqrt{n})$, excluding most problems of interest.
	\item \emph{Sample Complexity.}  As shown in \Cref{tab:comparison}, our sample complexity $p \geq \wt{\Omega}(\max\Brac{n,{\kappa^8}/{\mu^2}}n^4)$ in \Cref{eqn:sample-complexity-p} improves upon the result $p \geq \wt{\Omega}(\max\Brac{n,\kappa^8}{n^8}/{\eps^8})$ in \cite{li2018global}. As aforementioned, this improvement partly owes to the similarity of the Huber-loss to $\ell^1$-loss, so that the Huber-loss is much less heavy-tailed than the $\ell^4$-loss studied in \cite{li2018global}, requiring fewer samples for measure concentration. Still, our result leaves much room for improvement --  we believe the sample dependency on $\theta^{-1}$ is an artifact of our analysis\footnote{The same $\theta^{-1}$ dependency also appears in \cite{sun2016complete,li2018global,bai2018subgradient,zhang2018structured,gilboa2018efficient}.}, and the phase transition in \Cref{fig:phase pn} suggests that $p \geq \Omega( \mathrm{poly}\log(n) )$ samples might be sufficient for exact recovery. 
	\item \emph{Algorithmic Convergence.} Finally, it should be noted that the number of iteration $\wt{O}\paren{ n^4 + \log\paren{{1}/{\eps} } } $ for our algorithm substantially improves upon that $\wt{\mc O}(n^{12}/\eps^2)$ of the noisy RGD in \cite[Theorem IV.2]{li2018global}. This has been achieved via a two-stage approach: (i) we first run $\mc O(n^4)$ iterations of vanilla RGD to obtain an approximate solution; (ii) then perform a subgradient method with linear convergence to the ground-truth. Moreover, without any noise parameters to tune, vanilla RGD is more practical than the noisy RGD in \cite{li2018global}.
\end{itemize}




\subsection{A glimpse of high dimensional geometry}\label{subsec:geometry}
\input{sec/geometry}

\subsection{From geometry to efficient optimization}\label{subsec:algorithm}
\input{sec/algorithm}

%% file: sec/geometry.tex
\edited{To study the optimization landscape of the MCS-BD problem \eqref{eqn:problem}}, we simplify the problem by a change of variable $\ol{\mb q} = \mb Q\mb q$, which rotates the space by the orthogonal matrix $\mb Q$ in \eqref{eq:why precond}. Since the rotation $\mb Q$ does not change the optimization landscape, by an abuse of notation of $\mb q$ and $\ol{\mb q}$, we can rewrite the problem \eqref{eqn:problem} as 
\begin{align}\label{eqn:problem-rotate}
  \min_{ \mb q }\; f(\mb q) \;:=\; \frac{1}{np} \sum_{i=1}^p H_\mu \paren{ \mb C_{\mb x_i} \mb R \mb Q^{-1} \mb q }, \qquad \text{s.t.} \quad \norm{\mb q}{} \;=\; 1,
\end{align}
where we also used the fact that $\mb C_{\mb y_i} \mb P = \mb C_{\mb x_i} \mb R $ in \eqref{eq:why precond}. Moreover, since $\mb R \approx \mb Q$ is \emph{near orthogonal}, by assuming $\mb R\mb Q^{-1} = \mb I$, \edited{for \emph{pure} analysis purposes} we can further reduce \eqref{eqn:problem-rotate} to
\begin{align}\label{eqn:problem-simple}
    \min_{\mb q} \wt{f}({\mb q}) = \frac{1}{np} \sum_{i=1}^p H_\mu \paren{  \mb C_{\mb x_i} {\mb q} },\quad \text{s.t.}\quad \norm{ {\mb q}}{} = 1.
\end{align}

\begin{figure}[t]
		\centering
		\includegraphics[width = 0.4\textwidth]{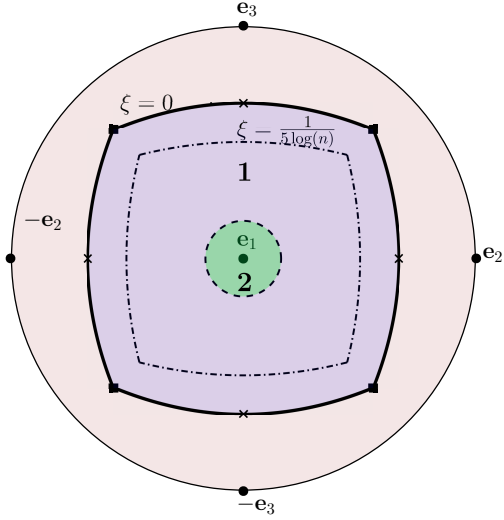}
		\caption{\edited{\textbf{Illustration of the set $\mc S_\xi^{1+}$ in 3-dimension.} \revised{Here we project the 3D unit sphere onto the plan spanned by $\mb e_2$ and $\mb e_3$.} Region 1 (purple region) denotes the interior of $\mc S_\xi^{1+}$ when $\xi =0$, where it includes one unique target solution. In this case, $\bigcup_{i=1}^3 \mc S_\xi^{\pm i}$ forms a full partition of the sphere, and the saddle points (denoted by $\times$) and local maximizers (denoted by $\blacksquare$) are on the boundary of the set. When $\xi>0$, the boundary of the set $\mc S_\xi^{1+}$ shrinks so that saddle points and local maximizers are excluded. We show the regularity condition \eqref{eqn:RC-condition} within $\mc S_\xi^{1+}$, excluding a green region of order $\mc O(\mu)$ (i.e., Region 2) due to the smoothing effect of the Huber. To obtain the exact solution within Region 2, rounding is required.} }
		\label{fig:set-demo}
\end{figure}

The reduction in \eqref{eqn:problem-simple} is simpler and much easier for \edited{parsing}. By a similar analysis as \cite{sun2016complete,gilboa2018efficient}, it can be shown that asymptotically the landscape is highly symmetric and the standard basis vectors $\Brac{ \pm \mb e_i }_{i=1}^n$ are \edited{approximately}\footnote{\edited{The standard basis $\Brac{ \pm \mb e_i }_{i=1}^n$ are exact global solutions for $\ell^1$-loss. The Huber loss we considered here introduces small approximation errors due to its smoothing effects.}} the only global minimizers. Hence, as $\mb R\mb Q^{-1} \approx \mb I$, we can study the optimization landscape of $f({\mb q})$ via studying the landscape of $\wt{f}({\mb q})$ followed by a perturbation analysis. \edited{As illustrated in \Cref{fig:set-demo},} based on the target solutions of $\wt{f}({\mb q})$, \edited{we partition} the sphere into $2n$ \edited{symmetric regions, and consider $2n$ (disjoint) subsets of each region}\footnote{\edited{Here, we define $\norm{ {\mb q}_{-i} }{\infty }^{-1} = + \infty$ when $\norm{ {\mb q}_{-i} }{\infty }=0$, so that the set $\mc S_\xi^{i+}$ is compact and $\mb e_i$ is also contained in the set.}}~\cite{gilboa2018efficient,bai2018subgradient} 
\begin{align*}
   \mc S_\xi^{i\pm} \; := \; \Brac{ { \mb q } \in \bb S^{n-1} \; \mid\; \frac{\abs{{q}_i}}{ \norm{ {\mb q}_{-i} }{\infty } }\ge \sqrt{1 + \xi}, \; q_i \gtrless 0  }, \quad \xi\in[0,\infty),
\end{align*}
where $\mb q_{-i}$ is a subvector of $\mb q$ with $i$-th entry removed. For every $i \in [n]$, $\mc S_\xi^{i+}$ (or $\mc S_\xi^{i-}$) contains exactly one of the target solution $\mb e_i$ (or $-\mb e_i$), and all points in this set have one unique largest entry with index $i$, so that they are closer to $\mb e_i$ (or $-\mb e_i$) in $\ell^\infty$ distance than all the other standard basis vectors. \edited{As shown in \Cref{fig:set-demo}, the union of these sets form a full partition of the sphere only when $\xi =0$. While for small $\xi>0$, each disjoint set excludes all the saddle points and maximizers, but their union covers most \edited{measure} of the sphere: when $\xi = \paren{5\log n}^{-1}$, \revised{their union covers at least half of the sphere, and hence} a random initialization falls into one of the regions $\mc S_\xi^{i\pm}$ with probability at least $1/2$ \cite{bai2018subgradient}}. \edited{Therefore, we can only consider the optimization landscapes on the sets $\mc S_\xi^{i\pm}$, where we show the Riemannian gradient of $f(\mb q)$}
\begin{align*}
   \grad f(\mb q) \;:=\;  \mc P_{\mb q^\perp} \nabla f(\mb q) \;=\;\paren{\mb I - \mb q \mb q^\top} \nabla f\paren{\mb q}
\end{align*}
satisfies the following properties. For convenience, we will simply present the results in terms of $\mc S_\xi^{i+}\;(1\leq i \leq n)$, but they also hold for $\mc S_\xi^{i-}$.
\begin{proposition}[Regularity Condition]\label{prop:regularity-main}
	Suppose $\theta \in \paren{\frac{1}{n} , \frac{1}{3} }$ and $\mu \leq c\min\Brac{ \theta, \frac{1}{\sqrt{n}} } $. When $p$ satisfies \eqref{eqn:sample-complexity-p}, w.h.p. over the randomness of $\Brac{\mb x_i}_{i=1}^p$, the Riemannian gradient of $f(\mb q)$ satisfies
\begin{align}
  	\innerprod{ \grad f(\mb q)  }{q_i\mb q - \mb e_i } \; &\geq \;  \alpha(\mb q)  \norm{ \mb q - \mb e_i }{}, \label{eqn:RC-condition}
\end{align}
for any $\mb q \in \mc S_\xi^{i+}$ with $ \sqrt{1 - q_i^2} \geq \mu $, where the regularity parameter is
\begin{align*}
	\alpha(\mb q) = \begin{cases}
 	c'\theta (1- \theta) q_i & \sqrt{1 - q_i^2} \in \brac{ \mu, \gamma } \\
 	c'\theta (1- \theta) n^{-1} q_i & \sqrt{1 - q_i^2} \geq \gamma
 \end{cases}
\end{align*}
 which increases as $\mb q$ gets closer to $\mb e_i$. Here $\gamma\in [\mu,1)$ is some constant.
\end{proposition}
\paragraph{Remark 2.} \edited{We defer detailed proofs to Appendix \ref{app:geometry-main}.} Here, our result is stated with respect to $\mb e_i$ for the sake of simplicity. It should be noted that asymptotically the global minimizer of \eqref{eqn:problem-rotate} is $\beta(\mb R\mb Q^{-1})^{-1} \mb e_i$ rather than $\mb e_i$, where $\beta$ is a \edited{normalization factor}. Nonetheless, as $\mb R\mb Q^{-1}\approx \mb I $, the global optimizer $ \beta (\mb R\mb Q^{-1})^{-1} \mb e_i $ of \eqref{eqn:problem-rotate} is very close to $\mb e_i$, so that we can state a similar result with respect to $ \beta(\mb R\mb Q^{-1})^{-1} \mb e_i $. The regularity condition \eqref{eqn:RC-condition} shows that any $\mb q \in \mc S_\xi^{i+}$ with $ \sqrt{1 - q_i^2} \geq \mu $ is not a stationary point. Similar regularity condition has been proved for phase retrieval \cite{candes2015phase}, dictionary learning \cite{bai2018subgradient}, etc. Such condition implies that the negative gradient direction coincides with the direction to the target solution. Even when it is close to the target, the lower bound on Riemannian gradient ensures that the gradient is large enough so that the iterate still makes rapid progress to the target solution. \edited{Finally, it should be noted that the regularity condition holds within all $\mc S_\xi^{i-}$ excluding a ball around $\mb e_i$ of radius $\mc O(\mu)$ (see \Cref{fig:set-demo}). This is due to the smoothing effect of the Huber. In the subsequent section, we will show how to obtain the exact solution within the ball via a rounding procedure.}

To ensure convergence of RGD, \edited{we also need to show the following property, so that once initialized in $\mc S_\xi^{i+}$ the iterates of the RGD method \emph{implicitly} regularize themselves staying in the set $\mc S_\xi^{i+}$. This ensures that the regularity condition \eqref{eqn:RC-condition} holds through the solution path of the RGD method.}
\begin{proposition}[Implicit Regularization]\label{prop:nega-curv-grad}
	Under the same condition of Proposition \ref{prop:regularity-main}, w.h.p. over the randomness of $\Brac{\mb x_i}_{i=1}^p$, the Riemannian gradient of $f(\mb q)$ satisfies
	\begin{align}
	\innerprod{ \grad f(\mb q) }{ \frac{1}{q_j} \mb e_j - \frac{1}{q_i} \mb e_i }	\;\geq \; c_4 \frac{\theta(1-\theta)}{n} \frac{\xi}{1+\xi},
\label{eqn:orthogonal-manifold}
\end{align}
for all $\mb q \in \mc S_\xi^{i+}$ and any $q_j$ such that $j \neq i$ and $q_j^2\geq \frac{1}{3}q_i^2 $.
\end{proposition}
\paragraph{Remark 3.} \edited{We defer detailed proofs to Appendix \ref{app:geometry-main}.} In a nutshell, \eqref{eqn:orthogonal-manifold} guarantees that the negative gradient direction points towards $\mb e_i$ component-wisely for relatively large components (i.e., $q_j^2\geq \frac{1}{3}q_i^2, \ \forall j\neq i$). With this, we can prove that those components will not increase after gradient update, ensuring the iterates stay within the region $\mc S_\xi^{i+}$. \edited{This type of implicit regularizations for the gradient has also been discovered for many nonconvex optimization problems, such as low-rank matrix factorizations \cite{gunasekar2017implicit,ma2017implicit,chi2018nonconvex,chen2018harnessing}, phase retrieval \cite{chen2019gradient}, and neural network training \cite{neyshabur2017geometry}.}

%% file: sec/algorithm.tex
\edited{Based on the geometric properties of the function we characterized in the previous section, we show how they lead to efficient optimization via a two-stage optimization method. All the detailed proofs of convergence are postponed to Appendix \ref{app:convergence}, and the implementation details of our methods can be found in Appendix \ref{app:algorithm}.}
\subsubsection*{Phase 1: Finding an approximate solution via RGD.} 
Starting from a \emph{random} initialization $\mb q^{(0)}$ uniformly drawn from $\bb S^{n-1}$, we solve the problem \eqref{eqn:problem} via \emph{vanilla} RGD
\begin{align}\label{eqn:grad-descent}
   \mb q^{(k+1)} = \mc P_{\bb S^{n-1}} \paren{ \mb q^{(k)} - \tau \cdot \grad f(\mb q^{(k)} ) },
\end{align}
where $\tau>0$ is the stepsize, and $\mc P_{\bb S^{n-1}} \paren{ \cdot } $ is a projection operator onto the sphere $\bb S^{n-1}$. 
\begin{proposition}[Linear convergence of gradient descent]\label{thm:grad-convergence}
Suppose Proposition \ref{prop:regularity-main} and Proposition \ref{prop:nega-curv-grad} hold. With probability at least $1/2$, the random initialization $\mb q^{(0)}$ falls into one of the regions $\mc S_\xi^{i\pm }$ \edited{for some $i\in [n]$}. Choosing a fixed step size $\tau \leq \frac{c}{n} \min \Brac{ \mu,n^{-3/2}  } $ in \eqref{eqn:grad-descent}, we have 
\begin{align*}
    \norm{ \mb q^{(k)} - \mb e_i }{} \leq 2\mu, \ \forall k\ge N:= \frac{C}{\theta} n^4 \log \paren{ \frac{1}{\mu} }.
\end{align*}
\end{proposition}
Because of the preconditioning and smoothing via Huber loss in \eqref{eqn:huber-loss}, the geometry structure in Proposition \ref{prop:regularity-main} implies that the gradient descent method can only produce \edited{an approximate solution $\mb q_s$} up to a precision $\mc O(\mu)$. Moreover, as we can show that $\| \mb e_i - \beta (\mb R\mb Q^{-1})^{-1} \mb e_i\|\leq \mu/2$, it does not make much difference \edited{of stating the result} in terms of either $\mb e_i$ or $\beta (\mb R\mb Q^{-1})^{-1} \mb e_i$. Next, we show that, \edited{by using $\mb q_s$} as a \emph{warm start}, an extra linear program (LP) rounding procedure produces an exact solution $(\mb R\mb Q^{-1})^{-1} \mb e_i$ up to a scaling factor in a few iterations.

\subsubsection*{Phase 2: Exact solutions via projected subgradient method for LP rounding.} 
Given the solution $\mb r = \mb q_s$ of running the RGD, we recover the exact solution by solving the following LP problem\footnote{\edited{Here, we state this problem in the same rotated space as \eqref{eqn:problem-rotate}. Since our geometric analysis is conducted in the rotated space, this is for convenience of stating our result. We will state the original problem subsequently.}}
\edited{\begin{align}\label{eqn:LP-rounding}
   \min_{\mb q}\;  \zeta(\mb q):=\frac{1}{np} \sum_{i=1}^p \norm{ \mb C_{\mb x_i} \mb R\mb Q^{-1} \mb q }{1} \quad\text{s.t.} \quad \innerprod{\mb r}{\mb q}\;=\; 1.	
\end{align}}
Since the feasible set $\innerprod{\mb r}{\mb q}= 1$ is essentially the tangent space of the sphere $\bb S^{n-1}$ at $\mb r$, and \edited{$\mb r= \mb q_s$} is pretty close to the target solution, one should expect that the optimizer $\mb q_r$ of \eqref{eqn:LP-rounding} exactly recovers the inverse kernel $\mb h$ up to a scaled-shift. The problem \eqref{eqn:LP-rounding} is \emph{convex} and can be directly solved using standard tools such as CVX \cite{grant2008cvx}, but it will be time consuming for large dataset. Instead, we introduce an efficient projected subgradient method for solving \eqref{eqn:LP-rounding}, 
\begin{align}\label{eqn:subgradient-LP}
	\mb q^{(k+1)} \; =\; \mb q^{(k)} - \tau^{(k)} \mc P_{\mb r^\perp} \partial \zeta(\mb q^{(k)}),
\end{align}
where $\partial \zeta(\mb q) $ is the subgradient of $\zeta(\cdot)$ at $\mb q$. For convenience, let $\wt{\mb r} \;:=\; \paren{ \mb R\mb Q^{-1} }^{-\top}\mb r$, and define the distance $d(\mb q)$ between $\mb q$ and the truth
\begin{align*}
   \mathrm{dist}(\mb q) \;:=\; \norm{ \mb d(\mb q) }{},\quad \mb d(\mb q) \;:=\; \mb q - \paren{ \mb R\mb Q^{-1} }^{-1} \frac{\mb e_i}{ \wt{r}_i }.
\end{align*}
\begin{proposition}Suppose $\mu \le \frac{1}{25}$ and let \edited{$\mb r = \mb q_s$} which satisfies $\norm{ \mb r - \mb e_i }{} \leq 2\mu$. Choose $\tau^{(k)}= \eta^k \tau^{(0)}$ with $\tau^{(0)} = c_1\log^{-2}(np) $ and $\eta \in [ \paren{ 1- c_2\log^{-2}(np) }^{1/2},1)$. Under the same condition of \Cref{thm:Main-all}, w.h.p. the sequence $\{\mb q^{(k)}\}$ produced by \eqref{eqn:subgradient-LP} with $\mb q^{(0)} = \mb r$ converges to the target solution in a linear rate, i.e.,
\begin{align*}
    \mathrm{dist}(\mb q^{(k)})\;\leq \; C\eta^k,\qquad \ \forall \ k\;=\; 0,1,2,\cdots .
\end{align*}
\end{proposition}

\paragraph{Remark 4.} Unlike smooth problems, \edited{in general, subgradient methods} for nonsmooth problems have to use \emph{geometrically} diminishing stepsize \edited{to achieve} linear convergence\footnote{Typical choices such as $\tau^{(k)} = \mc O(1/k)$ and $\tau^{(k)} = \mc O(1/\sqrt{k})$ lead to sublinear convergence~\cite{boyd2003subgradient}.} \cite{goffin1977convergence,li2018nonconvex,davis2018subgradient,li2019incremental}. The underlying geometry that supports the use of geometric diminishing step size and linear convergence \revised{in the above proposition} is the so-called \emph{sharpness} property \cite{burke1993weak,davis2018subgradient} of the problem \eqref{eqn:LP-rounding}. In particular, 
\edited{we prove that} w.h.p. $\zeta(\mb q)$ is sharp \edited{in the sense} that
\begin{align*}
   \zeta(\mb q) - \zeta\paren{  \paren{ \mb R\mb Q^{-1} }^{-1}  \mb e_i / \wt{r}_i}	\; \geq\; \revised{\frac{1}{50}\sqrt{\frac{2}{\pi}} \theta}  \cdot \mathrm{dist}(\mb q),\quad \forall \; \innerprod{\mb r}{\mb q}= 1.
\end{align*}
\revised{In a nutshell, the above sharpness implies that $(i)$ a scaled version of $\mb e_i$ is the unique global minimum of \eqref{eqn:LP-rounding}, and $(ii)$ the
 objective function $\zeta(\mb q)$ increases at least proportional to the distance that $\mb q$ moves away from the global minimum. This sharpness along with the convexity of \eqref{eqn:LP-rounding} enables us to develop efficient projected subgradient method that converges in a linear rate with geometrically diminishing step size. }
\paragraph{Remark 5.}
\edited{It should be noted that the LP rounding problem \eqref{eqn:LP-rounding} is stated in the same rotated space as \eqref{eqn:problem-rotate}, which is only for analysis purposes. By plugging $\mb q = \mb Q\mb q'$ into \eqref{eqn:problem-rotate} and abusing notations of $\mb q$ and $\mb q'$, we get back the \emph{actual} rounding problem in the same space as the problem \eqref{eqn:problem},
\begin{align*}
	\min_{\mb q} \; \frac{1}{np} \sum_{i=1}^p \norm{ \mb C_{\mb y_i} \mb P \mb q }{1}, \quad \text{s.t.} \quad \innerprod{ \mb r' }{ \mb q } \;=\;1,
\end{align*}
where $\mb r' = \mb Q \mb r = \mb Q\mb q_s$ is the \emph{actual} solution produced by running the RGD. }  

Finally, we end this section by noting that although we use \edited{the matrix-vector} form of convolutions in \eqref{eqn:grad-descent} and \eqref{eqn:subgradient-LP}, all the matrix-vector multiplications can be efficiently implemented by FFT, including the preconditioning matrix in \eqref{eqn:precond-mtx} which is also a circulant matrix. With FFT, the complexities of implementing one gradient update in \eqref{eqn:grad-descent} and subgradient in \eqref{eqn:subgradient-LP} are both $\mc O(pn\log n)$ for 1D problems.

%% file: sec/experiment.tex
\begin{figure}[t]
  \centering
  \begin{minipage}[b]{0.49\textwidth}
  \centering
    \includegraphics[width=\textwidth]{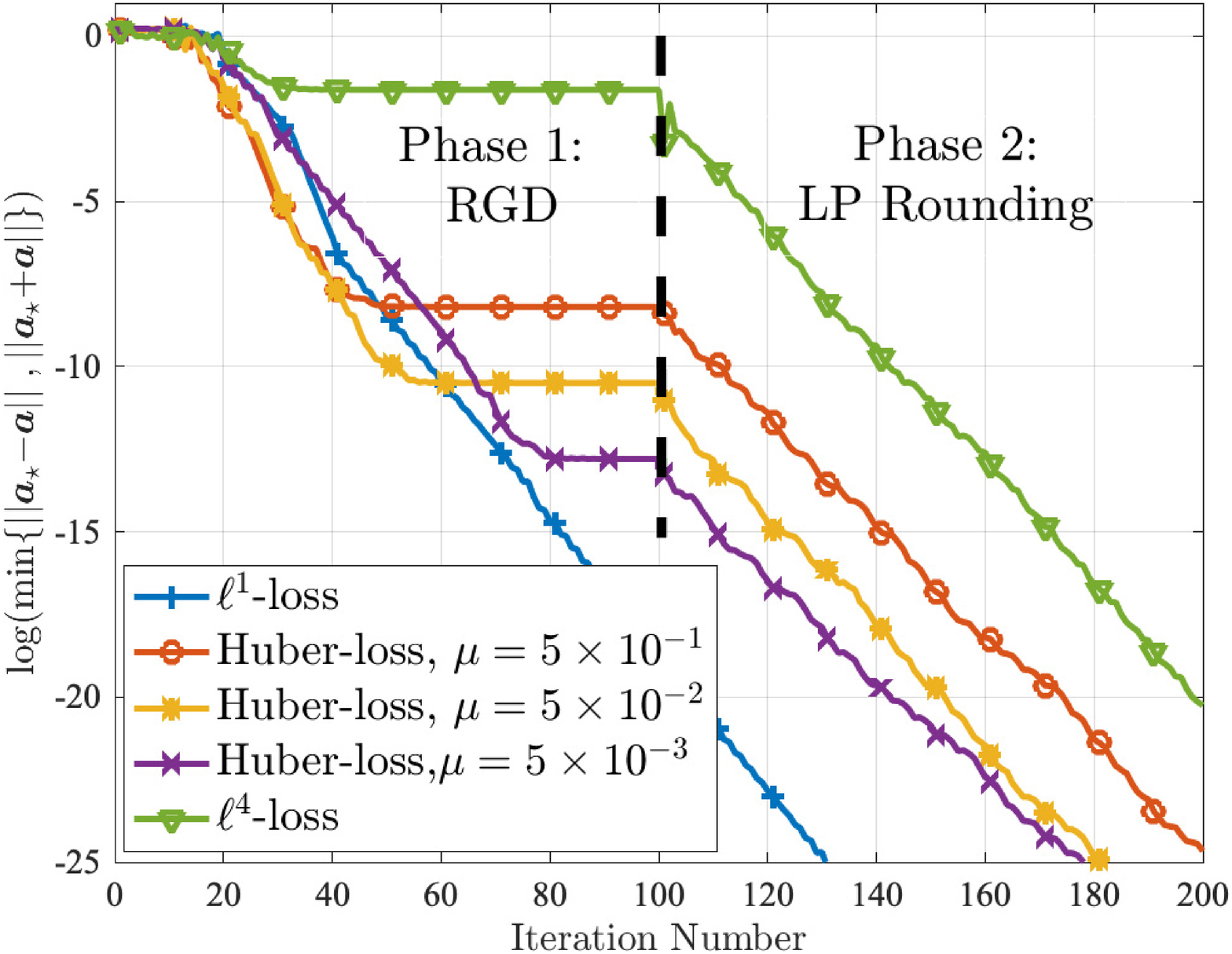}
    \caption{\textbf{Comparison of iterate convergence.} $p=50$, $n=200$, $\theta = 0.25 $. }\label{fig:convergence}
  \end{minipage}
  \hfill
  \begin{minipage}[b]{0.49\textwidth}
  \centering
    \includegraphics[width=\textwidth]{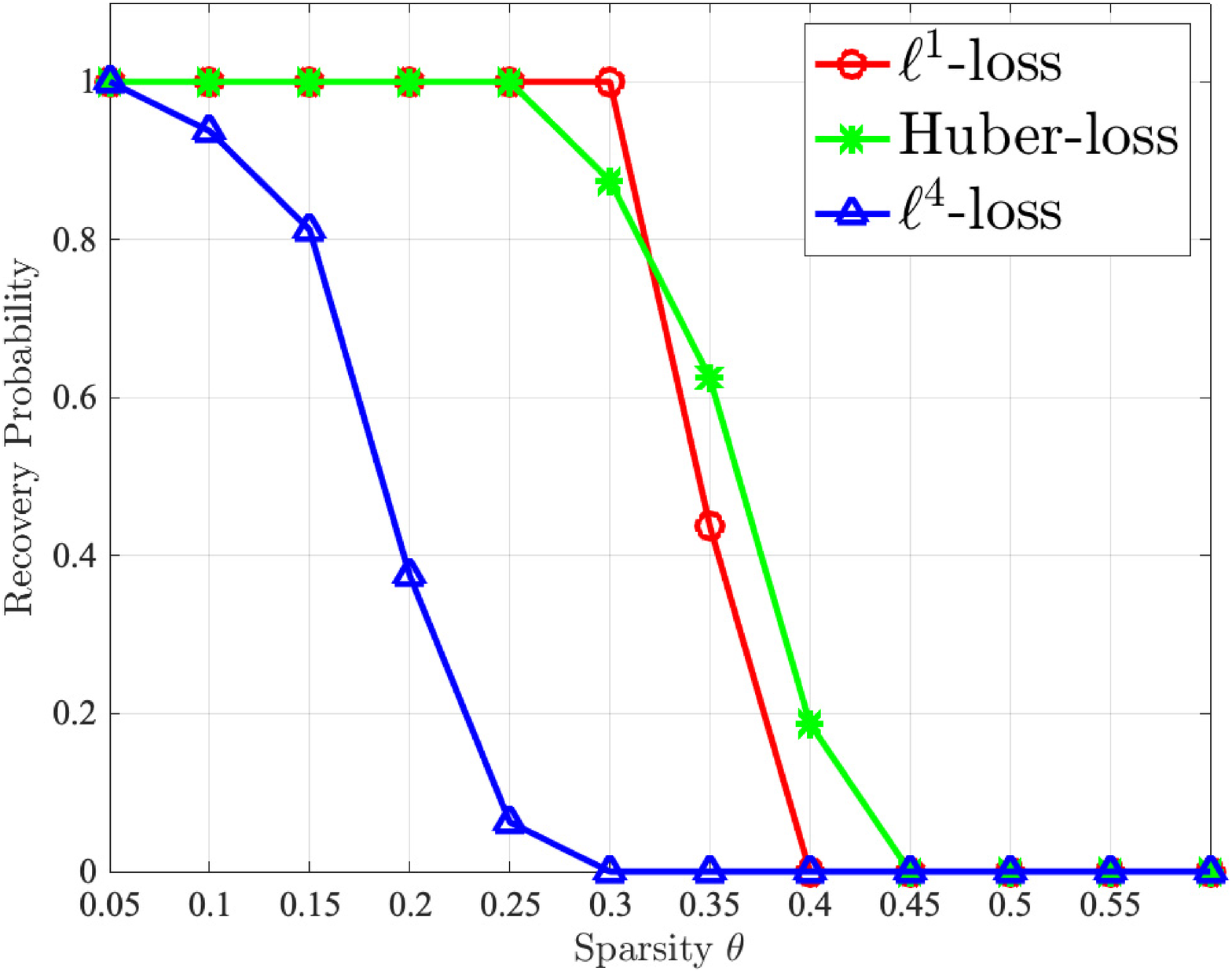}
    \caption{ \textbf{Comparison of recovery probability with varying $\theta$}. $p = 50$, $n = 500$.}\label{fig:vary theta}
  \end{minipage}
\end{figure}

\edited{In this section, we demonstrate the performance of the proposed methods on both synthetic and real dataset. On the synthetic dataset, we compare the iterate convergence and phase transition for optimizing Huber, $\ell^1$, and $\ell^4$ losses; for the real dataset, we demonstrate the effectiveness of our methods on sparse deconvolution for super-resolution microscopy imaging.}

\subsection{Experiments on 1D synthetic dataset} 
First, we conduct a series of experiments on synthetic dataset to demonstrate the superior performance of the vanilla RGD method \eqref{eqn:grad-descent}. For all synthetic experiments, we generate the measurements $\mb y_i = \mb a \;\conv\; \mb x_i$ ($1\leq i\leq p$), where the ground truth kernel $\mb a\in \R^n $ is drawn uniformly random from the sphere $\bb S^{n-1}$ (i.e., $\mb a \sim \mc U(\bb S^{n-1})$), and sparse signals  $\mb  x_i \in \R^n, i = [p]$ are drawn from i.i.d. Bernoulli-Gaussian distribution $\mb x_i \sim_{i.i.d.} \mc {BG}(\theta)$. 

We compare the performances of RGD\footnote{For $\ell^1$-loss, we use Riemannian subgradient method.} with random initialization on $\ell^1$-loss, Huber-loss, and the $\ell^4$-loss considered in \cite{li2018global}. We use line-search for adaptively choosing stepsize. \edited{For more implementation details, we refers the readers to Appendix \ref{app:algorithm}.}  For a fair comparison of optimizing all losses, we refine solutions with the LP rounding procedure \eqref{eqn:LP-rounding} optimized by projected subgradient descent \eqref{eqn:subgradient-LP}, and use the same random initialization uniformly drawn from the sphere. 

For judging the success of recovery, let $\mb q_\star$ be a solution produced by \edited{the two-stage} algorithm and we define
\begin{align*}
	\rho_{acc}(\mb q_\star) \;:=\;  \norm{ \mb C_{\mb a} \mb P \mb q_\star }{\infty} / \norm{ \mb C_{\mb a} \mb P \mb q_\star }{} \;\in\; [0,1].
\end{align*}
If $ \mb  q_\star$ achieves the target solution, it should satisfy $\mb P \mb  q_\star  = \mathrm{s}_\ell \brac{ \mb h}$, with $\mathrm{s}_\ell \brac{ \mb h}$ being some circulant shift of the inverse kernel of $\mb a$ and thus $\rho_{acc}(\mb q_\star) = 1$. Therefore, we should expect $\rho_{acc}(\mb q_\star) \approx 1$ when an algorithm produces a correct solution. For the following simulations, we assume successful recovery whenever $\rho_{acc}(\mb q_\star)\geq 0.95$.

\paragraph{Comparison of iterate convergence.} We first compare the convergence of our two-stage approach in terms of the distance from the iterate to the target solution (up to a shift ambiguity) for all losses using RGD. \edited{For Huber and $\ell^4$ losses, we run RGD for 100 iterations in Phase 1 and use the solution as warm start for solving LP rounding in Phase 2. For $\ell^1$-loss, we run Riemannian subgradient descent without rounding.} As shown in \Cref{fig:convergence}, in Phase 1, optimizing $\ell^4$-loss can only produce an approximate solution up to precision $10^{-2}$. In contrast, optimizing Huber-loss converges with much faster linear rate before iterates stagnate, and \edited{produces} much more accurate solutions as $\mu$ decreases, even without LP rounding. \edited{In Phase 2, for both losses, projected subgradient descent converges linearly to the target solution. For $\ell^1$ loss, the experiments tend to suggest that Riemannian subgradient exactly recovers the target solution in a linear rate even without LP rounding. We leave analyzing $\ell^1$-loss for future research.}
 
\paragraph{Recovery with varying sparsity.} \edited{Fixing $n =500$ and $p=50$}, we compare the recovery probability with varying sparsity level $\theta \in (0,0.6]$. \edited{For Huber loss, we use $\mu = 10^{-2}$.} For each value of $\theta$ and each loss, we run our two-stage optimization method and repeat the simulation \edited{$15$ times}. As illustrated in \Cref{fig:vary theta}, optimizing Huber-loss enables successful recovery for much larger $\theta $ in comparison with that of $\ell^4$-loss. The performances of optimizing $\ell^1$-loss and Huber-loss are quite similar, \edited{which achieves constant sparsity level $\theta \approx 1/3$ as suggested by our theory}.
 
\paragraph{Phase transition on $(p,n)$.} Finally, we fix $\theta = 0.25$, and test the dependency of sample number $p$ on the dimension $n$ via phase transition plots. \edited{For Huber loss, we use $\mu = 10^{-2}$.} For each individual $(p,n)$, we run our two-stage optimization method and repeat the simulation $15$ times. In \Cref{fig:phase pn}, whiter pixels indicate higher success probability, and vice versa. As illustrated in \Cref{fig:phase pn}, for each individual $n$, optimizing Huber-loss requires much fewer samples $p$ for recovery in comparison with that of $\ell^4$-loss. The performances of optimizing $\ell^1$-loss and Huber-loss are comparable; we conjecture sample dependency for optimizing both losses is $p\geq \Omega(\mathrm{poly}\log (n) )$, \edited{which is much better than our theory predicted}. In contrast, optimizing $\ell^4$-loss might need $p \geq \Omega(n)$ samples. \edited{This is mainly due to the heavy-tailed behavior for high order polynomial of random variables}.

\begin{figure*}[t]
\centering
\captionsetup[sub]{font=normalsize,labelfont={bf,sf}}
\centering
\begin{minipage}[c]{0.33\textwidth}
\centering
\subcaption{$\ell^1$-loss}
	\includegraphics[width = \linewidth]{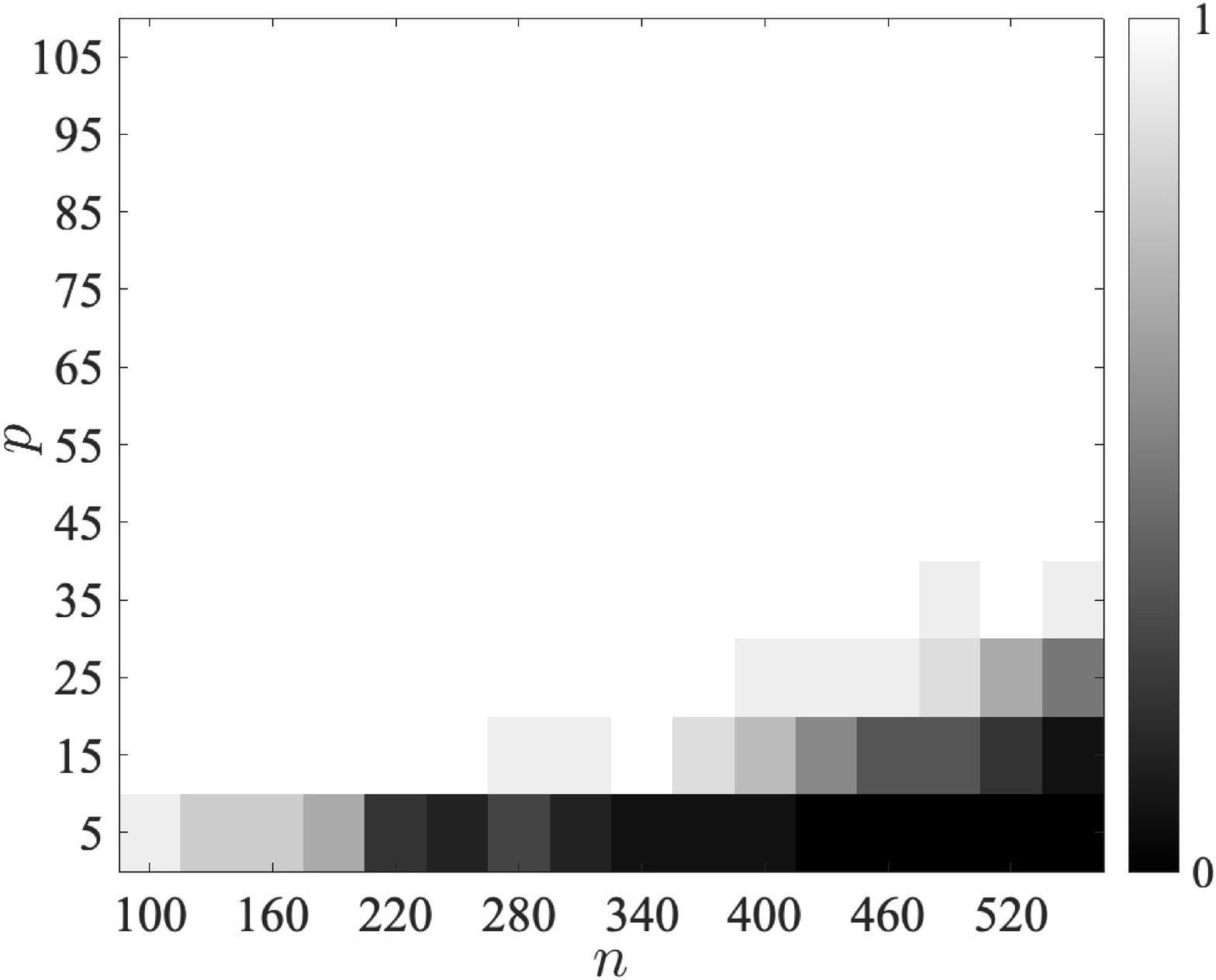}
\end{minipage}
\begin{minipage}[c]{0.33\textwidth}
\centering
\subcaption{Huber-loss}
	\includegraphics[width = \linewidth]{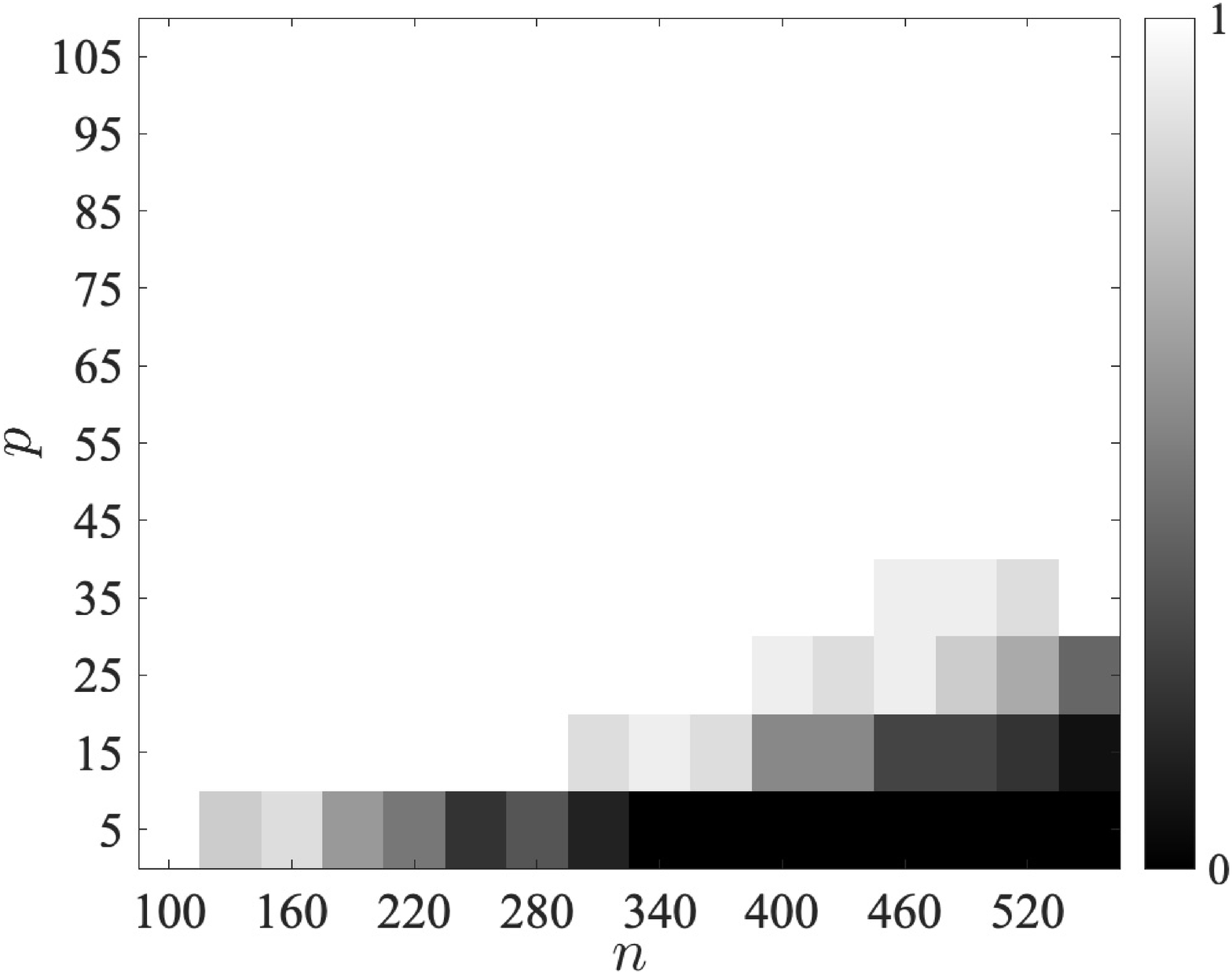}
\end{minipage}
\begin{minipage}[c]{0.33\textwidth}
\centering
\subcaption{$\ell^4$-loss}
	\includegraphics[width = \linewidth]{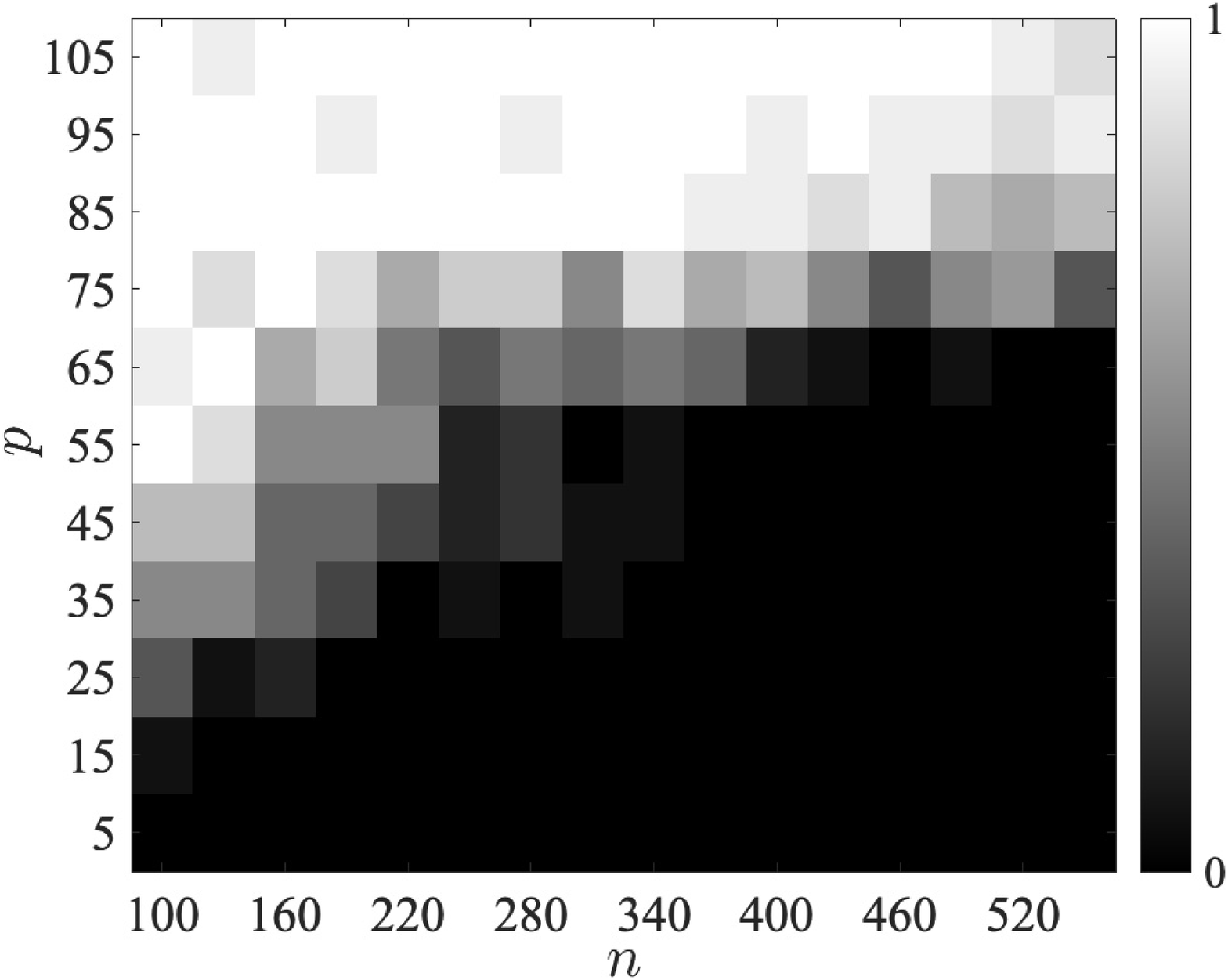}
\end{minipage}
\caption{\textbf{Comparison of phase transition on $(p,n)$ with fixed $\theta = 0.25$.} \revised{Here white denotes successful recovery while black indicates failure.}}
\label{fig:phase pn}
\end{figure*}

\begin{figure*}[hbt!]
	\centering
	\captionsetup[sub]{font=normalsize,labelfont={bf,sf}}
	\centering
    \begin{minipage}[c]{0.24\textwidth}
    	\centering
    	\subcaption{Observation}
    	\includegraphics[width = \linewidth]{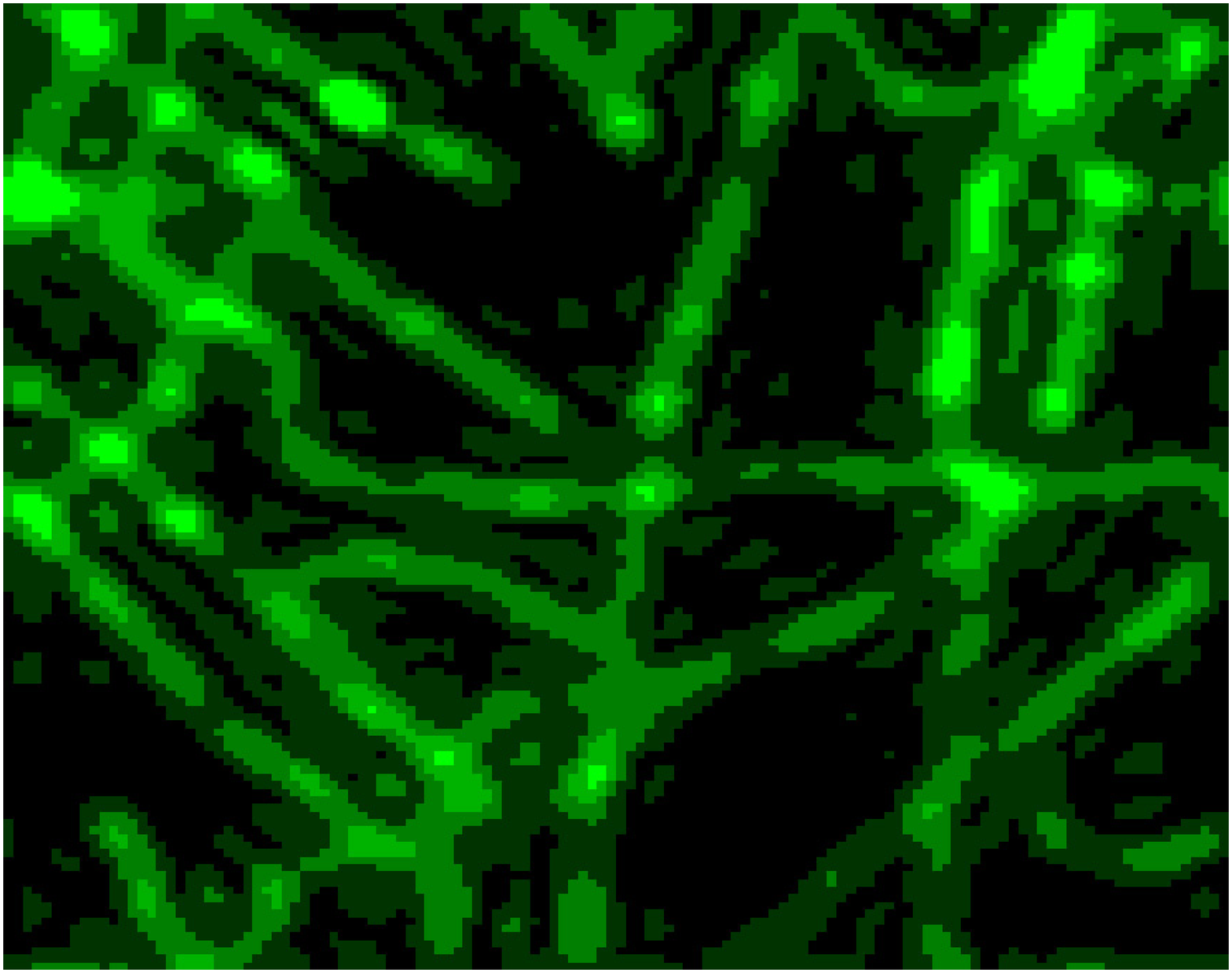}
    \end{minipage}
    \begin{minipage}[c]{0.24\textwidth}
    	\centering
    	\subcaption{Ground truth}
    	\includegraphics[width = \linewidth]{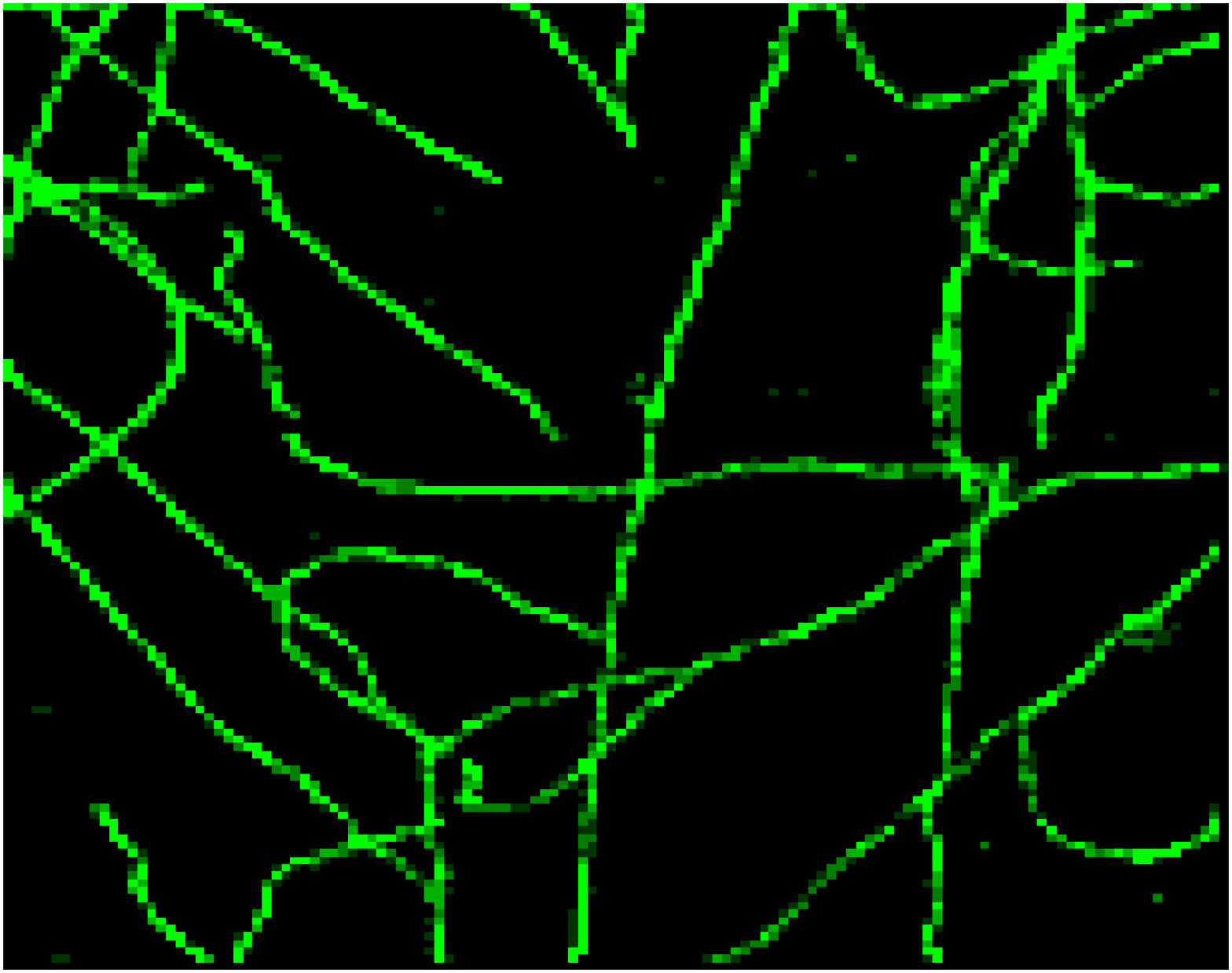}
        \end{minipage}
    \begin{minipage}[c]{0.24\textwidth}
    	\centering
    	\subcaption{Huber-loss}
    	\includegraphics[width = \linewidth]{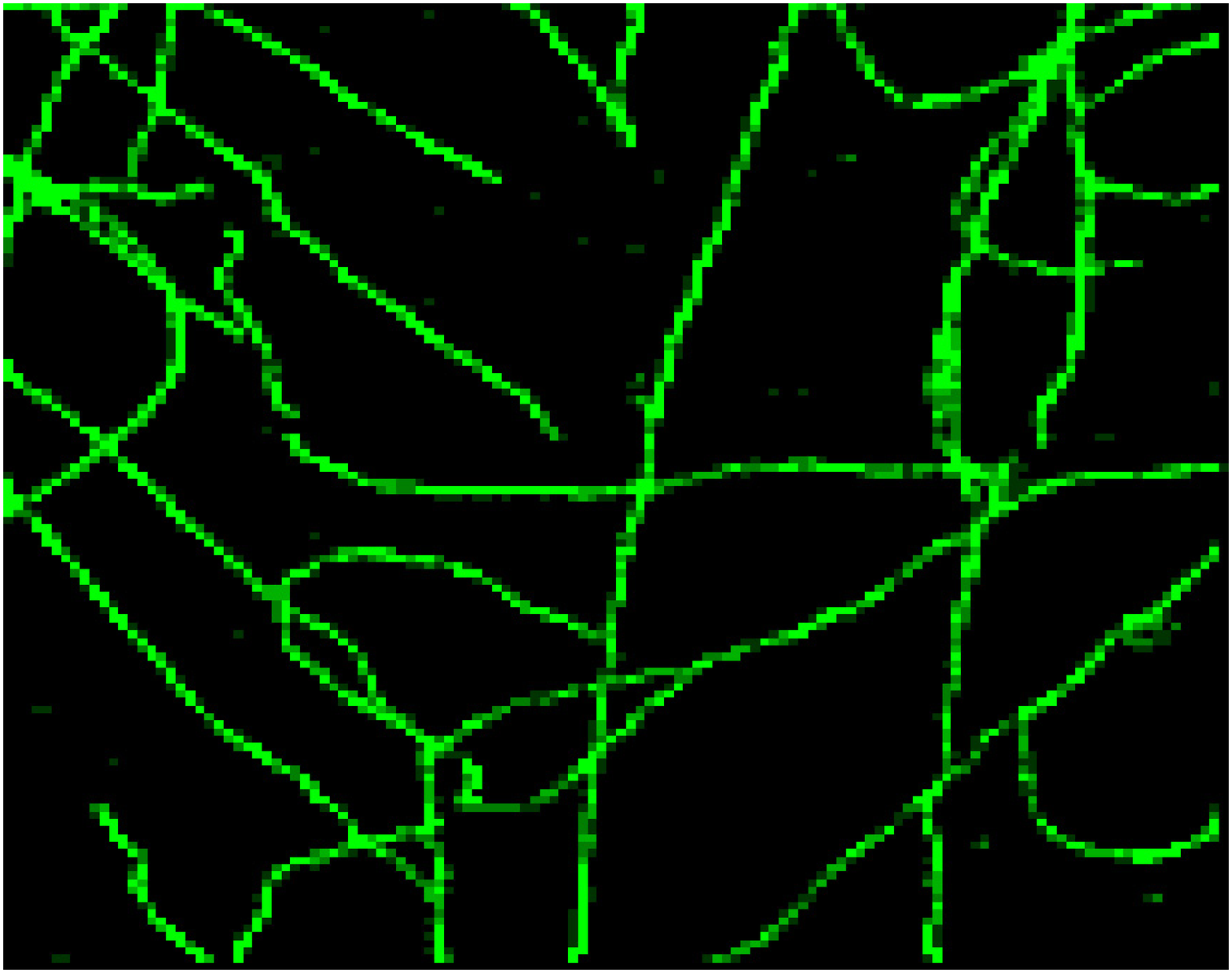}
    \end{minipage}
	\begin{minipage}[c]{0.24\textwidth}
		\centering
		\subcaption{$\ell^4$-loss}
		\includegraphics[width = \linewidth]{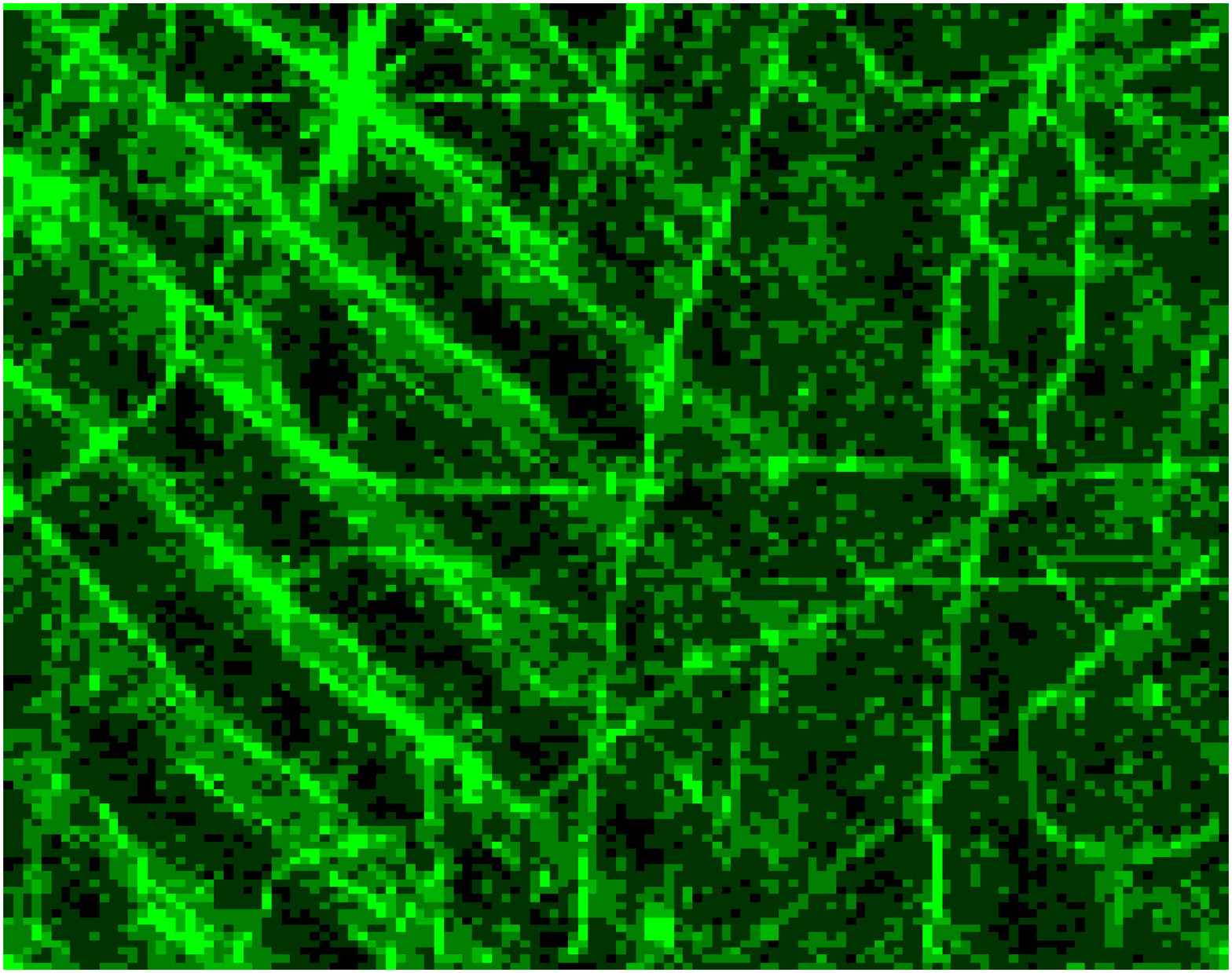}
	\end{minipage} \\ 
	\vspace{0.15in}
	\begin{minipage}[c]{0.32\textwidth}
		\centering
		\subcaption{Ground truth}
		\includegraphics[width = 0.8\linewidth]{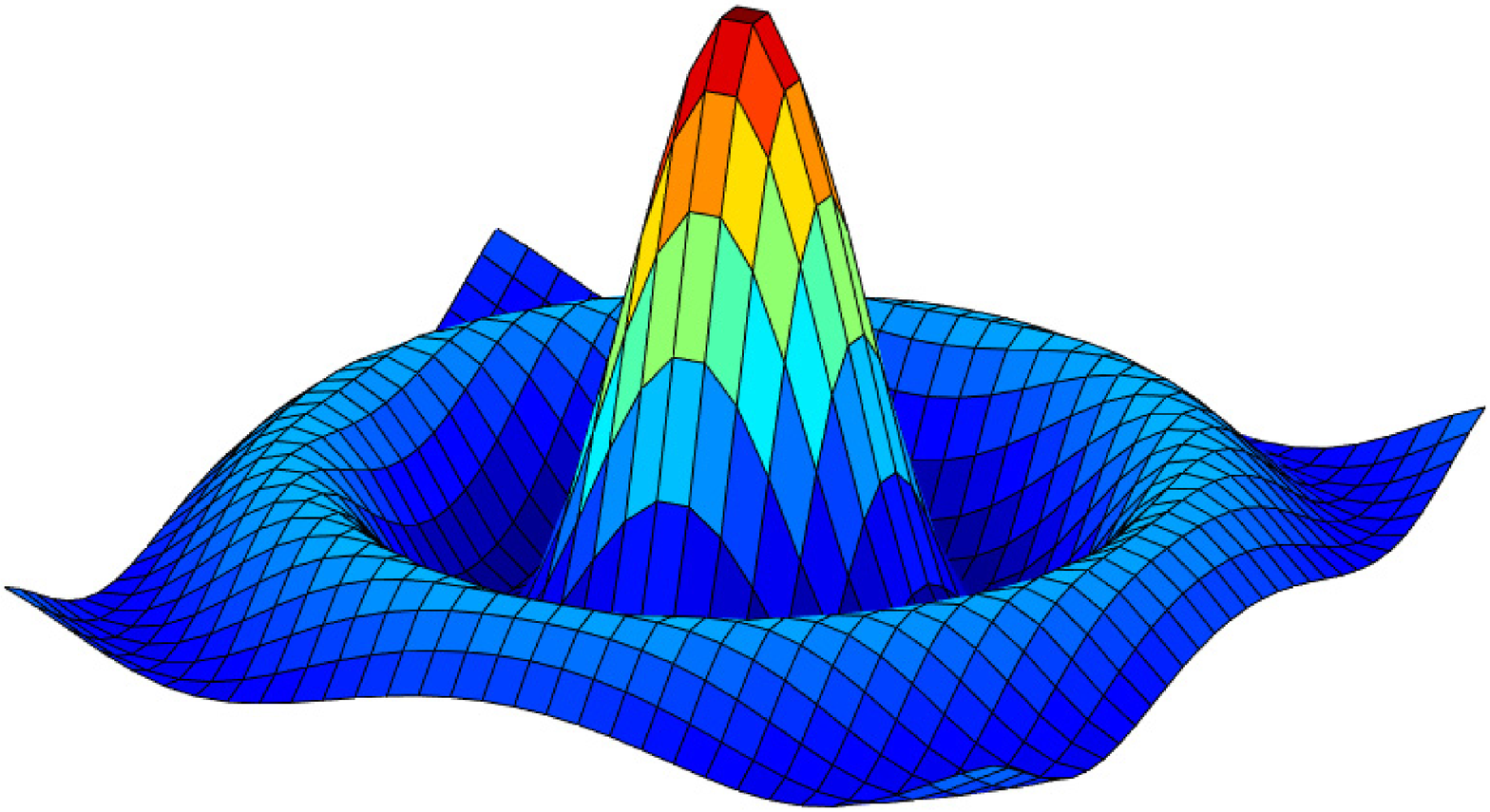}
	\end{minipage}
	\begin{minipage}[c]{0.32\textwidth}
		\centering
		\subcaption{Huber-loss}
		\includegraphics[width = 0.8\linewidth]{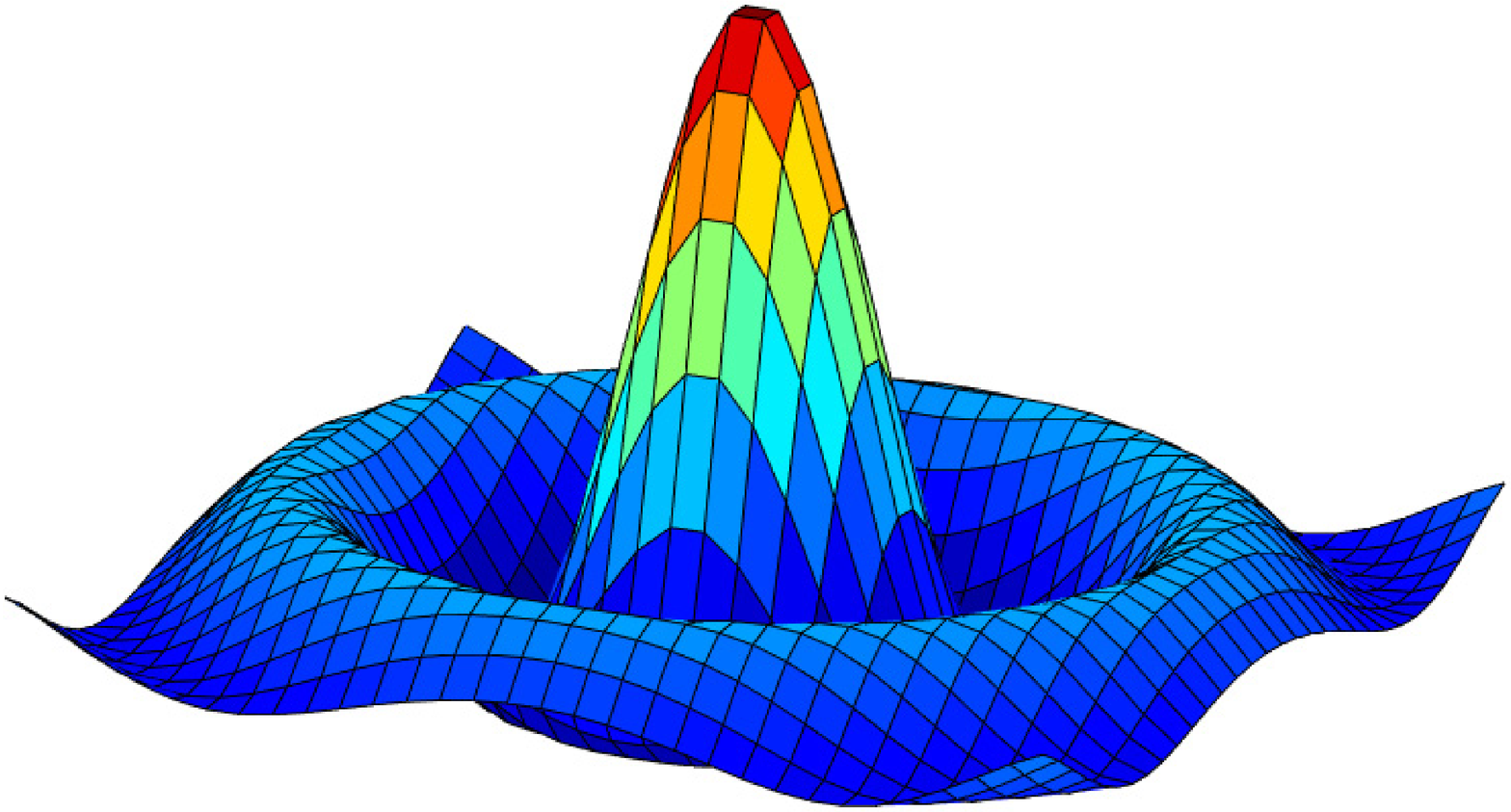}
	\end{minipage}
	\begin{minipage}[c]{0.34\textwidth}
		\centering
		\subcaption{$\ell^4$-loss}
		\includegraphics[width = 0.8\linewidth]{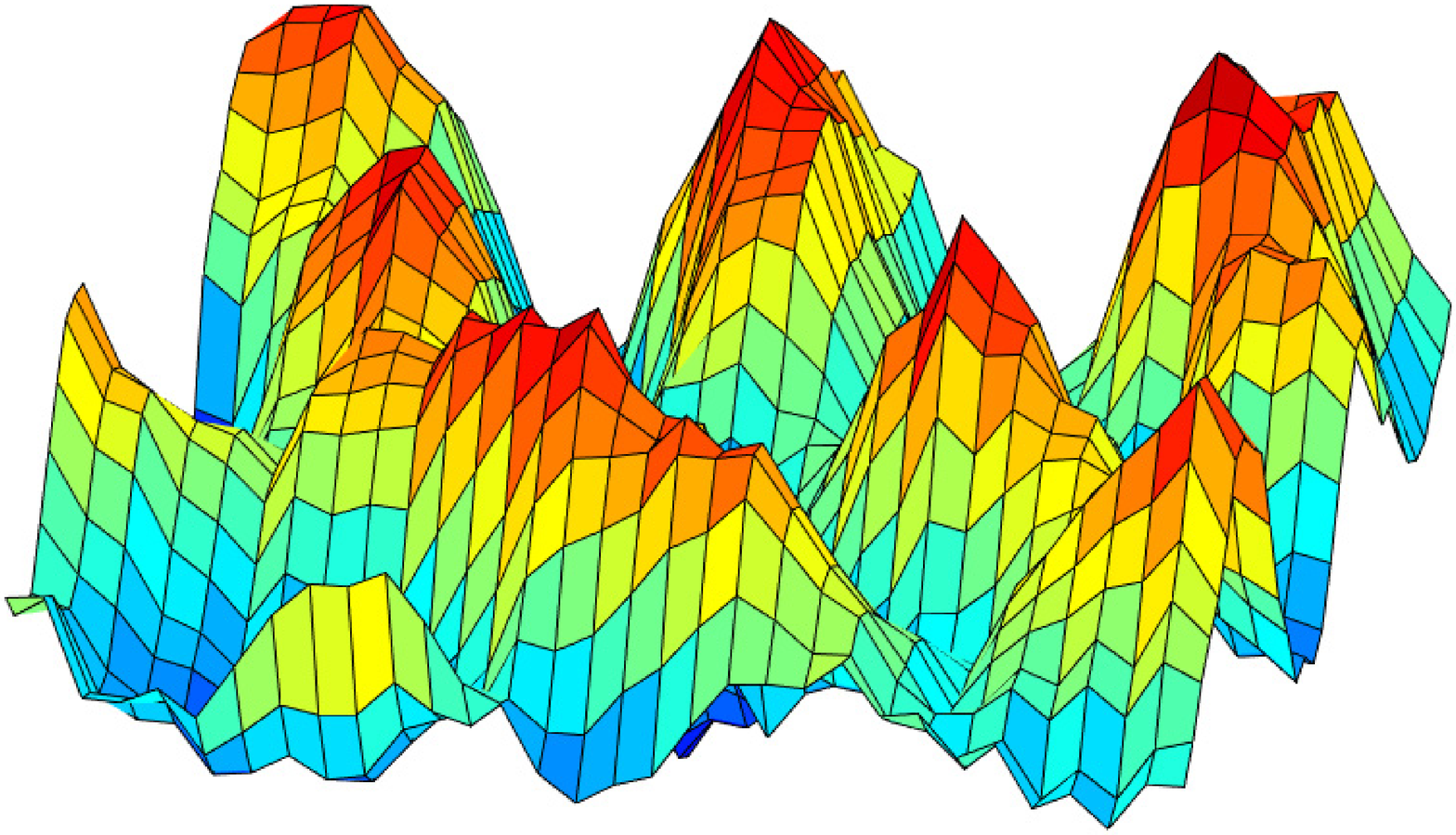}
	\end{minipage}
	\caption{\textbf{STORM imaging via solving MCS-BD.} The first line shows (a) observed image, (b) ground truth, (c) recovered image by optimizing Huber-loss, and (d) by $\ell^4$-loss. The second line, (e) ground truth kernel, (f) recovered by optimizing Huber-loss, and (g) by $\ell^4$-loss.  }
	\label{fig:storm}
\end{figure*}

\subsection{Real experiment on 2D super-resolution microscopy imaging}\label{subsec:real-exp}

As introduced in \Cref{sec:intro}, stochastic optical reconstruction microscopy (STORM) is a new computation based imaging technique which breaks the resolution limits of optical fluorescence microscopy \cite{betzig2006imaging,hess2006ultra,rust2006sub}. The basic principle is using photoswitchable florescent probes to create multiple sparse frames of individual molecules to temporally separates the spatially overlapping low resolution image,
\begin{align}\label{eqn:real-exp}
    \underbrace{\mb Y_i}_{\text{frame}} \quad=\quad \underbrace{\mb A}_{\text{PSF}} \quad \cconv\; \underbrace{\mb X_i}_{\text{sparse point sources}} \;+\quad \underbrace{\mb N_i}_{\text{noise}},
\end{align}
where $\cconv$ denotes 2D circular convolution, $\mb A$ is a 2D point spread function (PSF), $\Brac{\mb X_i}_{i=1}^p$ are sparse point-sources. \edited{The classical approaches solve the problem by fitting the blurred spots with Gaussian PSF using either maximum likelihood estimation or sparse recovery \cite{holden2011daostorm,zhu2012faster,small2014fluorophore}. However, these approaches suffer from limitations: $(i)$ for the case when the cluster of spots overlap, it is often computationally expensive and results in bad estimation; $(ii)$ for 3D imaging, the PSF exhibits aberration across the focus plane \cite{sarder2006deconvolution}, making it almost impossible to directly estimate it due to defocus and unknown aberrations.}

Therefore, given multiple frames $\Brac{\mb Y_i}_{i=1}^p$, in many cases we want to jointly estimate the PSF $\mb A$ and point sources $\Brac{\mb X_i}_{i=1}^p$. Once $\Brac{\mb X_i}_{i=1}^p$ are recovered, we can obtain a high resolution image by aggregating all sparse point sources $\mb X_i$. We test our algorithms on this task, by using $p=1000$ frames from Tubulin Conj-AL647 dataset obtained from SMLM challenge website\footnote{Available at \url{http://bigwww.epfl.ch/smlm/datasets/index.html?p=tubulin-conjal647}.}. The fluorescence wavelength is 690 nanometer (nm) and the imaging frequency is $f=25Hz$. Each frame is of size $128\times 128$ with 100 nm pixel resolution, and we solve the single-molecule localization problem on the same grid\footnote{\edited{Here, we are estimating the point sources $\mb X_i$ on the same pixel grid as the original image. To obtain even higher resolution than the result we obtain here, people are usually estimating the points sources on a finer grid. This results in a simultaneous sparse deconvolution and super-resolution problem, which could be an interesting problem for future research.}}. As demonstrated in \Cref{fig:storm}, optimizing Huber-loss using our two-stage method can near perfectly recover both the underlying Bessel PSF and point-sources $\Brac{\mb X_i}_{i=1}^p$, producing accurate high resolution image. In contrast, optimizing $\ell^4$-loss \cite{li2018global} fails to recover the PSF, resulting in some aliasing effects of the recovered high resolution image.





%% file: sec/discussion.tex
In this section, we first discuss related work on provable nonconvex methods for blind deconvolution and dictionary learning. We then conclude by pointing out several promising directions for future research.

\subsection{Relation to the literature}

Aside from the multichannel sparse model we studied here, many other low-dimensional models for blind deconvolution problems have been considered and investigated in the literature, that we discuss the relationship below.

\paragraph{Blind deconvolution with subspace model} Recently, there is a line of work studied the blind deconvolution problem with a single input $\mb y = \mb a \conv \mb x$, where the unknown $\mb a$ and $\mb x$ either live in known low-dimensional subspaces, or are sparse in some known dictionaries \cite{ahmed2014blind,chi2016guaranteed,ling2015self,li2016identifiability,kech2017optimal,ahmed2018leveraging,li2018bilinear}. These results assumed that the subspaces/dictionaries are chosen at random, such that the problem  \revised{exhibits no} signed shift ambiguity and can be provably solved either by convex relaxation \cite{ahmed2014blind,chi2016guaranteed} or nonconvex optimization \cite{li2018rapid,ma2017implicit}. 
\revised{However, their application to real problem is limited by the assumption of random subspace/dictionary model which is often not satisfied in practice.} 
In contrast, sparsity is a more natural assumption that appears in many signal processing \cite{tian2017multichannel}, imaging \cite{betzig2006imaging,kaaresen1998multichannel,levin2011understanding} and neuroscience \cite{gitelman2003modeling,ekanadham2011blind,wu2013blind,friedrich2017fast,pnevmatikakis2016simultaneous} applications. 

\paragraph{Multichannel deconvolution via cross-correlation based methods} The MCS-BD problem we considered here is also closely related to the multichannel blind deconvolution \revised{with  finite impulse response (FIR) models} \cite{xu1995least,moulines1995subspace,harikumar1998fir,lin2008blind,lee2018spectral,lee2018fast}. These methods utilize the second-order statistics of the observation, resulting in problems of larger size than MCS-BD. They often solve the problem via least squares or spectral methods. In particular, (i) Lin et al. \cite{lin2008blind} proposed an $\ell^1$-regularized least-squares method based on convex relaxation. However, the convex method could suffer similar sparsity limitation as \cite{wang2016blind,cosse2017note}, and it limits to two channels without theoretical guarantees. 
\revised{Lee et al. \cite{lee2018spectral} proposed an eigen approach for subspace model, and thus as discussed above it cannot directly handle our case with random sparse nonzero support.}

\paragraph{Short-and-sparse deconvolution} Another line of research related to this work is sparse blind deconvolution with short-and-sparse (SaS) model \cite{zhang2017global,zhang2018structured,kuo2019geometry,lau2019short}. They assume that there is a single measurement of the form $\mb y = \mb a \conv \mb x$, that $\mb x$ is sparse and the length of the kernel $\mb a$ is much shorter than $\mb y$ and $\mb x$. In particular, Zhang et al. \cite{zhang2018structured} formulated the problem as an $\ell^4$-maximization problem over the sphere similar to \cite{li2018global}, proving on a local region that every local minimizer is near a truncated signed shift of $\mb a$. Kuo et al. \cite{kuo2019geometry} studied a dropped quadratic simplification of bilinear Lasso objective \cite{lau2019short}, which provably obtains exact recovery for an incoherent kernel $\mb a$ and sparse $\mb x$. However, as the kernel and measurements are not the same length in SaS, the SaS deconvolution is much harder than MCS-BD: the problem has spurious local minimizers such as shift-truncations, so that most of the results there can only show benign local geometry structure regardless of the choice of objectives. This is in contrast to the MCS-BD problem we considered here, which has benign global geometric structure: as $\mb y$ and $\mb a$ are of the same length, every local minimizer corresponds to a full shift of $\mb a$ and there is no spurious local minimizer over the sphere \cite{li2018global}. On the other hand, despite the apparent similarity between the SaS model and MCS-BD, these problems are not equivalent: it might seem possible to reduce SaS to MCS-BD by dividing the single observation $\mb y$ into $p$ pieces; this apparent reduction fails due to boundary effects (e.g., shift-truncations on each piece).

\paragraph{Relation to dictionary learning} It should be noted that the MCS-BD problem is closely related to the complete dictionary learning problem studied in \cite{sun2016complete}. Indeed, if one writes 
\begin{align*}
  \underset{\textbf{data } \mb Y}{ \begin{bmatrix}
  	\mb C_{\mb y_1} & \cdots & \mb C_{\mb y_p}
  \end{bmatrix} }	\quad = \quad \mb C_{\mb a} \cdot \underset{\textbf{data } \mb X}{ \begin{bmatrix}
  	\mb C_{\mb x_1} & \cdots & \mb C_{\mb x_p},
  \end{bmatrix} }
\end{align*}
so that it reduces to the dictionary learning model $\mb Y = \mb C_{\mb a}\mb X$ with structured dictionary $\mb C_{\mb a}$. Thus, one may expect to directly recover\footnote{The intuition is that $\bb E \brac{ \norm{ \mb q^\top \mb P\mb Y }{1}  } \propto  \norm{ \mb q^\top \mb P \mb C_{\mb a} }{1}  $. Given $ \mb P\mb C_{\mb a}$ is near orthogonal, one may expect $\mb q^\top \mb P \mb C_{\mb a}$ is one sparse when $\mb q$ equals one preconditioned shift, which is the target solution.} one shift of $\mb a$ by optimizing
\begin{align*}
	\min_{\mb q} \; \norm{ \mb q^\top \mb P \mb Y }{1} = \norm{ \mb q^\top \mb P \mb C_{\mb a} \mb X }{1},\quad \text{s.t.}\quad \norm{\mb q}{} = 1.
\end{align*}
However, our preliminary experiment suggests that this formulation has some stability issues and often requires more samples in comparison to our formulation \eqref{eqn:problem}. We left further investigations for future work.

It should be noted that our proof ideas of convergence of RGD from random initialization are similar to that of Bai et al. and Gilboa et al. \cite{bai2018subgradient,gilboa2018efficient} for dictionary learning. Although dictionary learning and MCS-BD are related, these results do not directly apply to the sparse blind deconvolution problem. First of all, these results only apply to orthogonal dictionaries, while in sparse blind deconvolution the dictionary (in other words, the circulant matrix) $\mb C_{\mb a}$ is not orthogonal for generic unknown $\mb a$. To deal with this issue, preconditioning is needed as shown in our work. Furthermore, as the authors in \cite{bai2018subgradient} considered a nonsmooth $\ell^1$-loss, the non-Lipschitzness of subgradient of $\ell^1$ causes tremendous difficulties in measure concentration and dealing with preconditioning matrix. In this work, we considered the Huber-loss, which can be viewed as a first-order smooth surrogate of $\ell^1$-loss. Thus, we can utilize the Lipschitz continuity of its gradient to ease the analysis but achieving similar performances of using $\ell^1$-loss in terms sample complexity. In comparison with the sample complexity for complete dictionary learning with $p \sim \mc O(n^9) $ ignoring the condition number, our result is much tighter $p \sim \mc O(n^5)$ here.

Moreover, it should also be noted that both results \cite{bai2018subgradient,gilboa2018efficient} only guarantees sublinear convergence of their methods. In this work, we show a stronger result, that the vanilla RGD converges linearly to the target solution. Finally, we noticed that there appeared a result similar to ours \cite{shi2019manifold} after submission of our work, which considered a $\log\cosh$ function with improved sample complexity $p \sim \mc O(n^{4.5})$.

\paragraph{Finding the sparsest vectors in a subspace} As shown in \cite{wang2016blind}, the problem formulation  considered here for MCS-BD can be regarded as a variant of finding the sparsest vectors in a subspace \cite{qu2014finding,qu2020finding}. Prior to our result, similar ideas have led to new provable guarantees and efficient methods for several fundamental problems in signal processing and machine learning, such as complete dictionary learning \cite{sun2016complete,sun2017complete} and robust subspace recovery \cite{tsakiris2015dual,lerman2018overview,zhu2018dual}. We hope the methodology developed here can be applied to other problems falling in this category.

\subsection{Future directions}

Finally, we close this paper by pointing out several interesting directions for future work.

\paragraph{Improving sample complexity} Our result substantially improves upon \cite{li2018global}. However, there is still a large sample complexity \emph{gap} between our theory and practice. From the degree of freedom perspective (e.g., \cite{moulines1995subspace}), a constant $p$ is seemingly enough for solution uniqueness of MCS-BD. However, as the problem is highly nonconvex with unknown nonzero supports of the signals $\{\mb x_i\}$, to have provable efficient methods, we conjecture that paying extra log factors $p \geq \Omega\paren{\mathrm{poly}\log(n)}$ is necessary for optimizing $\ell^1$ and Huber losses, which is empirically confirmed by the phase transitions in \Cref{fig:phase pn} and experiments on 2D super-resolution imaging. This is similar to recent results on provable efficient method for multichannel blind deconvolution, which considers a different FIR model \cite{lee2018fast,lee2018spectral}. On the other hand, we believe our \emph{far from tight} sample complexity $p \geq \Omega\paren{ \mathrm{poly}(n) }$ is due to the looseness in our analysis: (i) tiny gradient near the boundary of the set $\mc S_\xi^{i\pm}$ for measure concentration, and (ii) loose control of summations of dependent random variables. To seek improvement, as the iterates of RGD only visit a small measure of the sphere, it could be better to perform an iterative analysis instead of uniformly characterizing the function landscape over $\mc S_\xi^{i\pm }$. Additionally, for tighter concentration of summation of dependent random variables, one might need to resort to more advanced probability tools such as decoupling \cite{de2012decoupling,qu2017convolutional} and generic chaining \cite{talagrand2014upper,dirksen2015tail}.   

\paragraph{Huber vs. $\ell^1$ loss and smooth vs. nonsmooth optimization} Our choice of Huber loss rather than $\ell^1$ -loss is to simplify theoretical analysis. Undoubtedly, $\ell^1$ -loss is a more natural sparsity promoting function and performs better than Huber as demonstrated by our experiments. When $\ell^1$-loss is utilized, \Cref{fig:convergence} tends to suggest that the underlying kernel and signals can be exactly recovered even without LP rounding\footnote{As the preconditioning matrix $\mb P$ introduces approximation error $\mb R\mb Q^{-1} \approx \mb I$ from \eqref{eqn:problem-rotate} to \eqref{eqn:problem-simple}, this is against our intuition in some sense.}. However, on the theoretic side, the subgradient of $\ell^1$-loss is \emph{non-Lipschitz} which introduces tremendous difficulty in controlling suprema of a random process and in perturbation analysis for preconditioning. Although recent work \cite{bai2018subgradient,ding2019noisy} introduced a novel method of controlling suprema of \emph{non-Lipschitz} random process, the difficulty of dealing with the preconditioning matrix in the subgradient remains very challenging. Similar to the ideas of \cite{li2018nonconvex,charisopoulos2019low}, one possibility might be showing \emph{weak convexity} and \emph{sharpness} of the Lipschitz $\ell^1$-loss function, rather than proving the regularity condition for the non-Lipschitz subgradient. We leave analyzing $\ell^1$-loss as a promising future research direction.

\paragraph{Robustness in the presence of noise} The current analysis focuses on the noiseless case. It is of interest to extend our result to the noisy case with measurements $ \mb y_i = \mb a \conv \mb x_i + \mb n_i,\forall \ i \in [p]$, where $\mb n_i$ denotes the additive Gaussian noise. Note that in the noiseless case (i.e., $\mb n_i = 0$), $\mb C_{\mb y_i} \mb q $ is sparse when $\mb q$ is the inverse of $\bm a$, motivating our approach \eqref{eqn:problem}. Therefore, in the noisy case, we expect $\mb C_{\mb y_i} \mb q $ to be close to a spare vector in the Euclidean space, which may lead to the following approach: $
   \min_{\mb q\in \bb S^{n-1}, \bm v_i}\;  \frac{1}{np} \sum_{i=1}^p \lambda H_\mu \paren{\mb v_i} +  \norm{\mb v_i - \mb C_{\mb y_i} \mb P \mb q }{}^2,$ where $\lambda$ is the balancing factor. On the other hand, the recent work \cite{ding2019noisy} on noisy robust subspace learning demonstrates that directly minimizing the $\ell^1$-loss over the sphere is robust to the additive noise, and achieves nearly optimal result in terms of the noise level. Motivated by this result, we also expect that both the formulation in \eqref{eqn:problem} and the RGD in \Cref{sec:main} are robust to the additive noise. Depending on the noise level and the parameter $\mu$, the LP rounding step may not be required in the noisy case. We leave the full investigation as future work.

\paragraph{Solving MCS-BD with extra data structures} In applications such as super-resolution microscopy imaging considered in \Cref{subsec:real-exp}, the data actually has more structures to be exploited. For example, the point sources $\Brac{\mb X_i}_{i=1}^p$ are often correlated that they share similar sparsity patterns, \revised{i.e., similar nonzero supports.} Therefore, one may want to enforce joint sparsity to capture this structure (\revised{e.g., by the $\norm{\cdot}{1,2}$ norm).} Analyzing this problem requires us to deal with probabilistic dependency across $\Brac{\mb X_i}_{i=1}^p$ \cite{li2018global}. On the other hand, we also want to solve the problem on a finer grid \revised{where the measurements are
\begin{align*}
    \mb Y_i \;=\; \mc D\brac{ \mb A \cconv \mb X_i },\quad 1\leq i\leq p	
\end{align*}
instead of \Cref{eqn:real-exp}. Here $\mc D\brac{\cdot}$ is a downsampling operator.
We leaves these MCS-BD  with the super-resolution problems for future research.}

\paragraph{Solving other nonconvex problems} This work joins recent line of research on provable nonconvex optimization \cite{jain2017non,sun2019link,chi2018nonconvex}. We believe the methodology developed here might be possible to be extended to other nonconvex bilinear problems. For instance, the blind gain and phase calibration problem \cite{li2016identifiability,ling2018self,li2018blind} is closely related to the MCS-BD problem, as mentioned by \cite{li2018global}. 
	\revised{It is also of great interest to extend our approach for solving the so-called convolutional dictionary learning problem \cite{chun2017convolutional,garcia2018convolutional}, in which each measurement consists of a superposition of multiple circulant convolutions:
	\begin{align*} \mb y_i \; = \; \sum_{k=1}^K \mb a_k \conv \mb x_{ik},\qquad 1\leq i\leq p. 	\end{align*} Given $\Brac{\mb y_i}_{i=1}^p$ we want to recover all the kernels $\Brac{\mb a_k}_{k=1}^K$ and sparse signals $\Brac{\mb x_{ik}}_{1\leq k\leq K, 1\leq i\leq p}$ simultaneously. We suspect our approach can be used to tackle this challenging problem and the number of samples will increase only proportionally to the number of kernels. We leave the full investigation as future work.}

%% file: sec/app_notation.tex
Throughout this paper, all vectors/matrices are written in bold font $\mb a$/$\mb A$; indexed values are written as $a_i, A_{ij}$. We use $\mb v_{-i}$ to denote a subvector of $\mb v$ without the $i$-th entry. Zeros or ones vectors are defined as $\mb 0_m$ or $\mb 1_m$ with $m$ denoting its length, and $i$-th canonical basis vector defined as $\mb e_i$. We use $\bb S^{n-1}$ to denote an $n$-dimensional unit sphere in the Euclidean space $\bb R^n$. We use $\mb z^{(k)}$ to denote an optimization variable $\mb z$ at $k$-th iteration. We let $[m] =\Brac{1,2,\cdots,m}$. Let $\mb F_n \in \bb C^{n \times n}$ denote a unnormalized $n\times n$ DFT matrix, with $\norm{\mb F_n}{} = \sqrt{n}$, and $\mb F_n^{-1} = n^{-1}\mb F_n^*$. In many cases, we just use $\mb F$ to denote the DFT matrix. We define $\sign(\cdot)$ as
\begin{align*}
   \sign(z) = \begin{cases}
   z/\abs{z}, & z\not = 0 \\
   0, & z =0 
 \end{cases}
\end{align*}

\paragraph{Some basic operators.} We use $\mc P_{\mb v}$ and $\mc P_{\mb v^\perp}$ to denote projections onto $\mb v$ and its orthogonal complement, respectively. We let $\mc P_{\bb S^{n-1}}$ to be the $\ell^2$-normalization operator. To sum up, we have
\begin{align*}
	\mc P_{\mb v^\perp} \mb u = \mb u - \frac{\mb v \mb v^\top }{\norm{\mb v}{}^2} \mb v,\quad  \mc P_{\mb v} \mb u = \frac{\mb v \mb v^\top }{\norm{\mb v}{}^2} \mb u,\quad \mc P_{\bb S^{n-1}} \mb v = \frac{\mb v}{\norm{\mb v}{}}.
\end{align*}

\paragraph{Circular convolution and circulant matrices.}
The convolution operator $\conv$ is \emph{circular} with modulo-$m$: $\paren{\mb a\conv \mb x}_i = \sum_{j=0}^{m-1} a_j x_{i-j}$, and we use $\cconv$ to specify the \emph{circular} convolution in 2D. For a vector $\mb v \in \bb R^m$, let $\mathrm{s}_\ell[\mb v]$ denote the cyclic shift of $\mb v$ with length $\ell$. Thus, we can introduce the circulant matrix $\mb C_{\mb v}\in \bb R^{m \times m}$ generated through $\mb v \in \bb R^m$,
\begin{align}\label{eqn:circulant matrx constrcut}
   \mb C_{\mb v} = \begin{bmatrix}
   v_1 & v_m & \cdots & v_3 & v_2 \\
   v_2 & v_1 & v_m  & & v_3 \\
   \vdots & v_2 & v_1 & \ddots &\vdots \\
   v_{m-1} &  & \ddots  & \ddots  &v_m\\
   v_m &  v_{m-1} & \cdots  & v_2 &v_1
 \end{bmatrix} = \begin{bmatrix}
 	\mathrm{s}_0\brac{ \mb v } & \mathrm{s}_1 \brac{\mb v} & \cdots & \mathrm{s}_{m-1}\brac{\mb v}
 \end{bmatrix}.
\end{align}
Now the circulant convolution can also be written in a simpler matrix-vector product form. For instance, for any $\mb u \in \bb R^m$ and $\mb v \in \bb R^m$,
\begin{align*}
   \mb u \conv \mb v = \mb C_{\mb u}  \cdot  \mb v = \mb C_{\mb v} \cdot \mb u.
\end{align*}
In addition, the correlation between $\mb u$ and $\mb v$ can be also written in a similar form of convolution operator which reverses one vector before convolution. Let $\wc{\mb v}$ denote a \emph{cyclic reversal} of $\mb v \in \bb R^m$, i.e., $\wc{\mb v} = \brac{ v_1, v_m, v_{m-1},\cdots,v_2 }^\top$, and define two correlation matrices $\mb C_{\mb v}^* \mb e_j = \mathrm{s}_j[\mb v]$ and $\wc{\mb C}_{\mb v} \mb e_j = \mathrm{s}_{-j}[\mb v]$. The two operators satisfy
\begin{align*}
   \mb C_{\mb v}^* \mb u = \wc{\mb v} \conv \mb u,\quad \wc{\mb C}_{\mb v} \mb u = \mb v \conv \wc{\mb u}.	
\end{align*}

\paragraph{Notation for several distributions.} We use $i.i.d.$ to denote \emph{identically} and \emph{independently distributed} random variables, and we introduce abbreviations for other distributions as follows.
\begin{itemize}
	\item we use $\mc N(\mu,\sigma^2)$ to denote Gaussian distribution with mean $\mu$ and variance $\sigma^2$;
	\item we use $\mc U( \bb S^{n-1} )$ to denote the uniform distribution over the sphere $\bb S^{n-1}$;
	\item we use $\mc B(\theta)$ to denote the Bernoulli distribution with parameter $\theta$ controling the nonzero probability;
	\item we use $\mc {BG}(\theta)$ to denote Bernoulli-Gaussian distribution, i.e., if $u \sim \mc {BG}(\theta) $, then $u = b \cdot g$ with $b \sim \mc B(\theta)$ and $g \sim \mc N(0,1)$;
	\item we use $\mc {BR}(\theta)$ to denote Bernoulli-Rademacher distribution, i.e., if $u \sim \mc {BR}(\theta) $, then $u = b \cdot r$ with $b \sim \mc B(\theta)$ and $r$ follows Rademacher distribution. 
\end{itemize}


%% file: sec/app_basic.tex
\begin{lemma}[Moments of the Gaussian Random Variable] \label{lem:gaussian_moment}
If $X \sim \mc N\left(0, \sigma_X^2\right)$, then it holds for all integer $m \geq 1$ that
\begin{align*}
\expect{\abs{X}^m} = \sigma_X^m \paren{m -1}!! \brac{ \sqrt{\frac{2}{\pi}} \indicator{m=2k+1 }+\indicator{m= 2k} } \leq \sigma_X^m \paren{m -1}!!,~k = \lfloor m/2 \rfloor.
\end{align*}
\end{lemma}


\begin{lemma}[sub-Gaussian Random Variables]\label{lem:sub-Gaussian}
	Let $X$ be a centered $\sigma^2$ sub-Gaussian random variable, such that
	\begin{align*}
	   \bb P\paren{ \abs{X} \geq t } \leq 2\exp\paren{ -\frac{t^2}{2\sigma^2} },
	\end{align*}
    then for any integer $p\geq 1$, we have
    \begin{align*}
       \bb E\brac{\abs{X}^p } \leq 	\paren{2\sigma^2}^{p/2} p \Gamma(p/2).
    \end{align*}
    In particular, we have
    \begin{align*}
       \norm{X}{L^p} = \paren{ \bb E\brac{\abs{X}^p } }^{1/p} \leq \sigma e^{1/e}	\sqrt{p}, \quad p \geq 2,
    \end{align*}
    and $\bb E\brac{\abs{X} }  \leq \sigma \sqrt{2\pi} $.
\end{lemma}

\begin{lemma}[Moment-Control Bernstein's Inequality for Random Variables \cite{foucart2013mathematical}] \label{lem:mc_bernstein_scalar}
Let $X_1, \cdots, X_N$ be i.i.d.\ real-valued random variables. Suppose that there exist some positive numbers $R$ and $\sigma_X^2$ such that
\begin{align*}
\expect{\abs{X_k}^m} \leq \frac{m!}{2} \sigma_X^2 R^{m-2}, \; \; \text{for all integers $m \ge 2$}.
\end{align*} 
Let $S \doteq \frac{1}{N}\sum_{k=1}^N X_k$, then for all $t > 0$, it holds  that 
\begin{align*}
\prob{\abs{S - \expect{S}} \ge t} \leq 2\exp\left(-\frac{Nt^2}{2\sigma_X^2 + 2Rt}\right).   
\end{align*}
\end{lemma}


\begin{lemma}[Gaussian Concentration Inequality]\label{lem:gauss-concentration}
Let $\mb g \in \bb R^n $ be a standard Gaussian random variable $\mb g \sim \mc N(\mb 0,\mb I)$, and let $f: \bb R^n \mapsto \bb R $ denote an $L$-Lipschitz function. Then for all $t>0$, 
	\begin{align*}
		\bb P\paren{ \abs{ f(\mb g) - \bb E \brac{f(\mb g)} }\geq t } \leq 2\exp \paren{ - \frac{t^2}{2L^2} }.
	\end{align*}
\end{lemma}

\begin{lemma}[Lemma VII.1, \cite{sun2017complete}]\label{lem:random-mtx-BG-l1}
Let $\mb M \in \bb R^{n_1\times n_2}$ with $\mb M \sim \mc {BG}(\theta)$ and $\theta \in (0,1/3)$. For a given set $\mc I \subseteq [n_2] $ with $\abs{\mc I} \leq \frac{9}{8} \theta n_2 $, whenever $n_2 \geq \frac{C}{\theta^2} n_1 \log \paren{ \frac{n_1}{ \theta } } $, it holds 
\begin{align*}
   \norm{ \mb v^\top \mb M_{\mc I^c} }{1} - \norm{ \mb v^\top \mb M_{\mc I} }{1}\;\geq \; \frac{n_2}{6} \sqrt{\frac{2}{\pi} } \theta \norm{\mb v}{}
\end{align*}
for all $\mb v \in \bb R^{n_1}$, with probability at least $1 - c n_2^{-6}$.
\end{lemma}

\begin{lemma}[Derivates of $h_{\mu}\left(z\right)$] \label{lem:derivatives_basic_surrogate}
The first two derivatives of $h_{\mu}\left(z\right)$ are 
\begin{align}
\nabla h_{\mu}\left(z\right) = 
\begin{cases}
\mathrm{sign}\left(z\right)  & \abs{z} \geq \mu \\
z/\mu  & \abs{z} < \mu
\end{cases}, 
\quad 
\nabla^2 h_{\mu}\left(z\right) = 
\begin{cases}
0  & \abs{z} > \mu \\
1/\mu  & \abs{z} < \mu
\end{cases}. 
\end{align}
Whenever necessary, we define $\nabla^2 h_{\mu}\left(\mu\right) = 0$, and write the ``second derivative'' as $\nabla^2 \overline{h}_{\mu}\left(\mu\right)$ instead. Moreover for all $z, z'$, 
\begin{align}
\abs{\nabla h_{\mu}\left(z\right) - \nabla h_{\mu}\left(z'\right)} \leq \frac{1}{\mu}\abs{z - z'}.
\end{align}
\end{lemma}

\begin{lemma}\label{lemma:aux_asymp_proof_b}
Let $X\sim \mc N(0,\sigma_x^2)$ and $Y\sim \mc N(0,\sigma_y^2)$ and $Z \sim \mc N\left(0, \sigma_z^2\right)$ be independent random variables. Then we have
\begin{align}
\expect{X\indicator{X+Y\ge \mu}} & = \frac{\sigma_x^2}{\sqrt{2\pi}\sqrt{\sigma_x^2+\sigma_y^2}}\exp\left(-\frac{\mu^2}{2(\sigma_x^2+\sigma_y^2)}\right), \\
\expect{XY\indicator{\abs{X+Y}\le \mu}} & = - \sqrt{\frac{2}{\pi}} \frac{\mu \sigma_x^2\sigma_y^2}{\left(\sigma_x^2+\sigma_y^2\right)^{3/2}}\exp\left(-\frac{\mu^2}{2\left(\sigma_x^2+\sigma_y^2\right)}\right), \\
\expect{\abs{X} \indicator{\abs{X} > \mu}} & = \sqrt{\frac{2}{\pi}}\sigma_x\exp\left(-\frac{\mu^2}{2\sigma_x^2}\right), \\
\expect{XY \indicator{\abs{X + Y + Z} < \mu}} & = -\sqrt{\frac{2}{\pi}} \mu \exp\left(-\frac{\mu^2}{2\left(\sigma_x^2 + \sigma_y^2 + \sigma_z^2\right)}\right) \frac{\sigma_x^2 \sigma_y^2}{\left(\sigma_x^2 + \sigma_y^2 + \sigma_z^2\right)^{3/2}}, \\
\expect{X^2 \indicator{\abs{X} < \mu}} & = -\sqrt{\frac{2}{\pi}} \sigma_x \mu \exp\left(-\frac{\mu^2}{2\sigma_x^2}\right) + \sigma_x^2 \prob{\abs{X} < \mu}, \\
\expect{X^2 \indicator{\abs{X + Y} < \mu}} & = -\sqrt{\frac{2}{\pi}}\mu \frac{\sigma_x^4 }{\left(\sigma_x^2 + \sigma_y^2\right)^{3/2}} \exp\left(-\frac{\mu^2}{2\left(\sigma_x^2 + \sigma_y^2\right)}\right) + \sigma_x^2 \prob{\abs{X + Y} < \mu}. 
\end{align} 
\end{lemma}
\begin{proof}
Direct calculations.
\end{proof}

\begin{lemma}[Calculus for Function of Matrices, Chapter X of \cite{bhatia2013matrix}]\label{lem:func_matrix_basics}
Let $\mc S^{n \times n}$ be the set of symmetric matrices of size $n \times n$. We define a map $f: \mc S^{n \times n} \mapsto \mc S^{n \times n} $ as
\begin{align*}
	f(\mb A) = \mb U f(\mb \Lambda) \mb U^*, 	
\end{align*}
where $\mb A \in \mc S^{n \times n}$ has the eigen-decomposition $\mb A = \mb U \mb \Lambda \mb U^*$. 
The map $f$ is called (Fr\'echet) differentiable at $\mb A$ if there exists a linear transformation on $\mc S^{n \times n}$ such that for all $\mb \Delta$
\begin{align*}
   \norm{ f(\mb A + \mb \Delta) - f(\mb A) - \mathrm{D} f(\mb A)[\mb \Delta] }{} = o\paren{ \norm{ \mb \Delta }{} }.	
\end{align*}
The linear operator $\mathrm{D} f(\mb A)$ is called the derivative of $f$ at $\mb A$, and $\mathrm{D} f(\mb A)[\mb \Delta]$ is the directional derivative of $f$ along $\mb \Delta$. If $f$ is differentiable at $\mb A$, then
\begin{align*}
   \mathrm{D} f(\mb A)[\mb \Delta] = \frac{d}{dt}  f(\mb A + t \mb \Delta)	\bigg|_{t=0}.
\end{align*}
We denote the operator norm of the derivative $\mathrm{D} f(\mb A)$ as
\begin{align*}
    \norm{  \mathrm{D} f(\mb A) }{}	\doteq \sup_{ \norm{\mb \Delta}{} =1 } \norm{ \mathrm{D} f(\mb A)[\mb \Delta] }{}.
\end{align*}
\end{lemma}

\begin{lemma}[Mean Value Theorem for Function of Matrices]
	Let $f$ be a differentiable map from a convex subset $\mc U$ of a Banach space $\mc X$ into the Banach space $\mc Y$. Let $\mb A,\mb B \in \mc U$, and let $\mc L$ be the line segment joining them. Then
	\begin{align*}
	    \norm{ f(\mb B) - f(\mb A) }{} \leq \norm{ \mb B - \mb A }{} \sup_{ \mb U \in \mc L } \norm{ \mathrm{D} f( \mb U ) }{}	.
	\end{align*}
\end{lemma}

\begin{lemma}[Theorem VII.2.3 of \cite{bhatia2013matrix}]\label{lem:S-equation}
	Let $\mb A$ and $\mb B$ be operators whose spectra are contained in the open right half-plane and open left half-plane, respectively. Then the solution of the equation $\mb A\mb X - \mb X\mb B = \mb Y$ can be expressed as
	\begin{align*}
	   \mb X = \int_0^\infty e^{-t\mb A} \mb Y e^{ t \mb B} dt	
	\end{align*}

\end{lemma}

\begin{lemma}\label{lem:grad-matrix-norm-bound}
    Let $f(\mb A) = \mb A^{-1/2}$, defined the set of all $n\times n$ positive definite matrices $\mc S_+^{n \times n}$, then we have
    \begin{align*}
       \norm{ \mathrm{D} f(\mb A) }{}\leq \frac{1}{\sigma_{\min}^2(\mb A)},
    \end{align*}
	where $\sigma_{\min}(\mb A)$ is the smallest singular value of $\mb A$.
\end{lemma}

\begin{proof}
	To bound the operator norm $\norm{ \mathrm{D} f(\mb A) }{} $, we introduce an auxiliary function
\begin{align*}
   g(\mb A) \;=\; \mb A^{-2}, \qquad f(\mb A) = g^{-1}(\mb A),
\end{align*}
such that $f$ and $g$ are the inverse function to each other. Whenever $g \paren{ f(\mb A)} \not = 0 $ (which is true for our case $\mb A \succ \mb 0$), this gives
\begin{align}\label{eqn:deriative-inverse}
   \mathrm{D} f(\mb A) \;=\; \brac{\mathrm{D} g \paren{f (\mb A)} }^{-1} \;=\; \brac{ \mathrm{D}g (\mb A^{-1/2})  }^{-1} .
\end{align}
This suggests that we can estimate $\mathrm{D}f(\mb A)$ via estimating $\mathrm{D}g(\mb A)$ of its inverse function $g$. Let 
\begin{align*}
    g \;=\; h \paren{ w(\mb A)},\quad  h(\mb A) \;=\; \mb A^{-1},\quad w(\mb A) = \mb A^2,	
\end{align*}
such that their directional derivatives have simple form
\begin{align*}
   \mathrm{D}h(\mb A)[\mb \Delta] \;=\; - \mb A^{-1} \mb \Delta \mb A^{-1},\quad \mathrm{D}w(\mb A) [\mb \Delta] \;=\; \mb \Delta \mb A + \mb A \mb \Delta.
\end{align*}
By using chain rule, simple calculation gives
\begin{align*}
   \mathrm{D}g(\mb A)[\mb \Delta] \;&=\;  \mathrm{D} h(w(\mb A)) \brac{	\mathrm{D} w(\mb A) [\mb \Delta] }, \\
   \;&=\; - \paren{ \mb A^{-2} \mb \Delta \mb A^{-1} + \mb A^{-1} \mb \Delta  \mb A^{-2} }.
\end{align*}
Now by \eqref{eqn:deriative-inverse}, the directional derivative
\begin{align*}
   \mb Z \doteq \mathrm{D}f(\mb A)[\mb \Delta]	
\end{align*}
satisfies 
\begin{align*}
      \mb A \mb Z \mb A^{1/2} + \mb A^{1/2} \mb Z \mb A \;=\; -\mb \Delta.
\end{align*}
Since $\mb A \succ  \mb 0$, we write the eigen decomposition as $\mb A = \mb U \mb \Lambda \mb U^* $, with $\mb U$ orthogonal and $\mb \Lambda> 0$ diagonal. Let $\wt{\mb Z} = \mb U^* \mb Z \mb U$ and $\wt{\mb \Delta} = \mb U^* \mb \Delta \mb U$ , then the equation above gives
\begin{align*}
      \mb \Lambda^{1/2} \wt{\mb Z} - \wt{\mb Z} \paren{ - \mb \Lambda^{1/2} } = - \mb \Lambda^{-1/2} \wt{\mb \Delta} \mb \Lambda^{-1/2},
\end{align*}
which is the Sylvester equation []. Since $\mb \Lambda^{1/2}$ and $- \mb \Lambda^{1/2}$ do not have common eigenvalues, Lemma \ref{lem:S-equation} gives
\begin{align*}
  \mathrm{D}f(\mb A)[\mb \Delta] \;=\; \mb U \brac{\int_0^\infty  e^{- \mb \Lambda^{1/2} \tau } \paren{- \mb \Lambda^{-1/2} \wt{\mb \Delta} \mb \Lambda^{-1/2} } e^{- \mb \Lambda^{1/2} \tau } d\tau } \mb U^*.
\end{align*}
Thus, by Lemma \ref{lem:func_matrix_basics} we know that
\begin{align*}
   \norm{ \mathrm{D} f(\mb A) }{} \;&=\; \sup_{ \norm{\mb \Delta}{}=1 } \norm{  \mathrm{D}f(\mb A)[\mb \Delta] }{} \\
   \;&\leq \; \int_0^\infty \norm{ e^{- \mb \Lambda^{1/2} \tau } \paren{- \mb \Lambda^{-1/2} \wt{\mb \Delta} \mb \Lambda^{-1/2} } e^{- \mb \Lambda^{1/2} \tau }  }{}  d\tau \\
   \;&\leq\; \norm{ \mb \Lambda^{-1/2} \wt{\mb \Delta} \mb \Lambda^{-1/2} }{} \int_0^\infty  e^{  - \sigma_{\min} \tau } d\tau \;\leq\; \frac{ 1 }{ \sigma_{\min}^2 (\mb A) }.
\end{align*}
\end{proof}

\begin{lemma}[Matrix Perturbation Bound]\label{lem:matrix-perturbation}
	Suppose $\mb A \succ \mb 0$. Then for any symmetric perturbation matrix $\mb \Delta$ with $\norm{\mb \Delta }{}\leq \frac{1}{2} \sigma_{\min}(\mb A)$, it holds that
	\begin{align*}
	   \norm{ \paren{ \mb A + \mb \Delta }^{-1/2} - \mb A^{-1/2} }{} \;\leq \; \frac{4 \norm{ \mb \Delta }{} }{ \sigma_{\min}^2(\mb A) },	
	\end{align*}
where $\sigma_{\min}(\mb A)$ denotes the minimum singular value of $\mb A$.
\end{lemma}

\begin{proof}
Let us denote $f(\mb A) = \mb A^{-1/2}$. Given a symmetric perturbation matrix $\mb \Delta$, by mean value theorem, we have
\begin{align*}
   \norm{ \paren{\mb A + \mb \Delta}^{-1/2} - \mb A^{-1/2} }{} \;&=\; \norm{ \int_0^1 \mathrm{D} f(\mb A+t \mb \Delta)[\mb \Delta]  dt }{} \\
   \;&\leq \; \paren{\sup_{t \in [0,1]} \norm{ \mathrm{D} f(\mb A + t \mb \Delta) }{} } \cdot  \norm{ \mb \Delta }{}.
\end{align*}
Thus, by Lemma \ref{lem:grad-matrix-norm-bound} and by using the fact that $\norm{\mb \Delta }{}\leq \frac{1}{2} \sigma_{\min}(\mb A)$, we have
\begin{align*}
   \norm{ \paren{\mb A + \mb \Delta}^{-1/2} - \mb A^{-1/2} }{} \;\leq \; \paren{ \sup_{t \in [0,1] } \frac{1}{ \sigma_{min}^2(\mb A + t \mb \Delta) } }  \norm{ \mb \Delta }{} \;\leq\; \frac{4 \norm{ \mb \Delta }{} }{  \sigma_{\min}^2(\mb A) },
\end{align*}
as desired.
\end{proof}

%% file: sec/app_main_geometry.tex
\edited{In this part of the appendix, we prove our main geometric result stated in \Cref{subsec:geometry}. Namely, we show the objective introduced in \eqref{eqn:problem-rotate}
\begin{align}\label{eqn:problem-rotate-app}
  \min_{ \mb q }\; f(\mb q) \;:=\; \frac{1}{np} \sum_{i=1}^p H_\mu \paren{ \mb C_{\mb x_i} \mb R \mb Q^{-1} \mb q }, \qquad \text{s.t.} \quad \norm{\mb q}{} \;=\; 1
\end{align}
with
\begin{align*}
   \mb R \;=\; \mb C_{\mb a}\paren{ \frac{1}{\theta np} \sum_{i=1}^p \mb C_{\mb y_i}^\top \mb C_{\mb y_i}  }^{-1/2},\quad \mb Q \;=\; \mb C_{\mb a} \paren{ \mb C_{\mb a}^\top \mb C_{\mb a} }^{-1/2},
\end{align*}
have benign first-order geometric structure. Namely, we prove that the function satisfies the regularity condition in Proposition \ref{prop:regularity-precond} and implicit regularization in Proposition \ref{prop:grad_orth-precond} properties over every one of the sets
\begin{align*}
   \mc S_\xi^{i\pm} \; := \; \Brac{ { \mb q } \in \bb S^{n-1} \; \mid\; \frac{\abs{{q}_i}}{ \norm{ {\mb q}_{-i} }{\infty } }\ge \sqrt{1 + \xi}, \; q_i \gtrless 0  }, \quad \xi\;\in\;(0,\infty),
\end{align*}
and we also show that the gradient is bounded all over the sphere (Proposition \ref{prop:grad-bound-precond}). These geometric properties enable efficient optimization via vanilla Riemannian gradient descent methods. In Appendix \ref{app:convergence}, we will leverage on these properties for proving convergence of our proposed optimization methods.

As aforementioned in \Cref{subsec:geometry}, the basic idea of our analysis is first reducing \eqref{eqn:problem-rotate-app} to a simpler objective
\begin{align}\label{eqn:problem-simple-app}
    \min_{\mb q} \wt{f}({\mb q}) = \frac{1}{np} \sum_{i=1}^p H_\mu \paren{  \mb C_{\mb x_i} {\mb q} },\quad \text{s.t.}\quad \norm{ {\mb q}}{} = 1.
\end{align}
by using the fact that $\mb R \approx \mb Q$ and assuming $\mb R\mb Q^{-1} = \mb I$. In Appendix \ref{app:regularity-population} and Appendix \ref{app:implicit-population}, we show the geometric properties hold in population for $\wt{f}({\mb q})$. We turn these results into non-asymptotic version via concentration analysis in Appendix \ref{app:gradient-concentration}. Finally, we prove these results for $f(\mb q)$ in \eqref{eqn:problem-rotate-app} via a perturbation analysis in Appendix \ref{app:precond}.

First, we show that regularity condition of the Riemannian gradient of $f(\mb q)$ over the set $\mc S_\xi^{i\pm}$ as follows.
} 
\begin{proposition}[Regularity condition]\label{prop:regularity-precond}
Suppose $\theta \geq \frac{1}{n}$ and $\mu \leq c_0\min\Brac{ \theta, \frac{1}{\sqrt{n}} } $. There exists some numerical constant $\gamma\in (0,1)$, when the sample complexity
\begin{align*}
  p \geq C \max\Brac{n , \frac{\kappa^8 }{ \theta \mu^2 \sigma_{\min}^2 } \log^4 n } \xi^{-2}\theta^{-2} n^4  	\log \paren{ \frac{ \theta n }{ \mu} },
\end{align*}
with probability at least $1- n^{-c_1} - c_2 np^{-c_3 n \theta }$ over the randomness of $\Brac{\mb x_i}_{i=1}^p$, we have
\begin{align}
  	\innerprod{ \grad f(\mb q)  }{q_i\mb q - \mb e_i } \; &\geq \;  c_4 \theta (1-\theta)  q_i \norm{ \mb q - \mb e_i }{}  ,\quad \sqrt{1-q_i^2} \in \brac{\mu,\;\gamma}, \label{eq:final RC1-near}\\
  	\innerprod{ \grad f(\mb q)  }{q_i\mb q - \mb e_i } \; &\geq \;  c_4 \theta (1-\theta) q_i n^{-1}\norm{ \mb q - \mb e_i }{},\quad \sqrt{1-q_i^2} \in \brac{\gamma, \;\sqrt{ \frac{n-1}{n} } }, \label{eq:final RC1-far}
\end{align}
holds for any $\mb q \in \mc S_\xi^{i+}$ and each index $i \in [n]$. Here, $c_0$, $c_1$, $c_2$, $c_3$, $c_4$, and $C$ are positive numerical constants.
\end{proposition}

\begin{proof}
Without loss of generality, it is enough to consider the case $i = n$. For all $\mb q \in \mc S_\xi^{n+}$, we have
\begin{align*}
    & \innerprod{ \grad f(\mb q)  }{q_n\mb q - \mb e_n } \\
   =\;&   \innerprod{ \grad f(\mb q) - \grad \wt{f}(\mb q) +  \grad \wt{f}(\mb q) - \grad \bb E\brac{\wt{f}(\mb q) } + \grad \bb E\brac{\wt{f}(\mb q) } }{q_n\mb q - \mb e_n } \\
   \geq \;& \innerprod{  \grad \bb E\brac{\wt{f}(\mb q) } }{ q_n \mb q - \mb e_n } - \abs{\innerprod{ \grad f(\mb q) - \grad \wt{f}(\mb q)}{ q_n\mb q - \mb e_n } } \\
   & -\abs{\innerprod{\grad \wt{f}(\mb q) - \grad \bb E\brac{\wt{f}(\mb q) } }{ q_n\mb q - \mb e_n } }.
\end{align*}
From Proposition \ref{prop:regularity-population}, when $\theta \geq \frac{1}{n}$ and $\mu \leq c_0\min\Brac{ \theta, \frac{1}{\sqrt{n}} } $, we know that in the worst case scenario, 
\begin{align*}
	 \innerprod{  \grad \bb E\brac{\wt{f}(\mb q) } }{ q_n \mb q - \mb e_n }  \;\geq \; c_1 \theta (1-\theta)  \xi n^{-3/2} \norm{ \mb q_{-n}}{}
\end{align*}
holds for all $\mb q \in \mc S_\xi^{n+}$. On the other hand, by Corollary \ref{cor:gradient-concentration}, when $p \geq C_1 \theta^{-2} \xi^{-2} n^5 \log \paren{ \frac{ \theta n }{ \mu} } $, we have
\begin{align*}
     \abs{\innerprod{\grad \wt{f}(\mb q) - \grad \bb E\brac{\wt{f}(\mb q) } }{ q_n\mb q - \mb e_n } } 
    \;\leq\;& \norm{\grad \wt{f}(\mb q) - \grad \bb E\brac{\wt{f}(\mb q) }}{}  \norm{ q_n \mb q - \mb e_n }{} \\
    \;\leq\;& \frac{c_1}{3} \theta(1-\theta) \xi n^{-3/2} \norm{ q_n \mb q - \mb e_n }{}
\end{align*}
holds for all $\mb q \in \mc S_\xi^{n+}$ with probability at least $ 1- n p^{-c_2\theta n} -n \exp\paren{ -c_3 n^2 }$. Moreover, from Proposition \ref{prop:preconditioning}, we know that when $p \geq C \frac{\kappa^8 n^4}{  \mu^2 \theta^3  \sigma_{\min}^2 \xi^2 }  \log^4 n \log \paren{ \frac{ \theta n }{ \mu } }$
\begin{align*}
    \abs{\innerprod{ \grad f(\mb q) - \grad \wt{f}(\mb q)}{ q_n\mb q - \mb e_n } } \;&\leq\; \norm{ q_n \mb q - \mb e_n }{}  \cdot \norm{  \grad f(\mb q) - \grad \wt{f}(\mb q) }{} \\
    \;&\leq\; \frac{c_1}{3} \theta(1-\theta) \xi n^{-3/2} \norm{ q_n \mb q - \mb e_n }{}
\end{align*}
holds for all $\mb q \in \mc S_\xi^{n+}$ with probability at least $1 -c_4p^{-c_5 n \theta } - n^{-c_6} - n e^{-c_7 \theta np } $. By combining all the bounds above, we obtain the desired result.
\end{proof}

Second, we show that the Riemannian gradient of $f(\mb q)$ also satisfies implicit regularization over $\mc S_\xi^{i\pm}$, such that iterates of the RGD method stays within one of the sets $\mc S_\xi^{i\pm}$ for sufficiently small stepsizes.

\begin{proposition}[Implicit Regularization]\label{prop:grad_orth-precond}
Suppose $\theta \geq \frac{1}{n}$ and $\mu \leq \frac{c_0}{\sqrt{n}} $. For any index $i \in [n]$, when the sample 
\begin{align*}
  p \geq C \max\Brac{n , \frac{\kappa^8 }{ \theta \mu^2 \sigma_{\min}^2 } \log^4 n } \xi^{-2}\theta^{-2} n^4  	\log \paren{ \frac{ \theta n }{ \mu} },
\end{align*}
with probability at least $1- n^{-c_1} - c_2 np^{-c_3 n \theta }$ over the randomness of $\Brac{\mb x_i}_{i=1}^p$, we have
\begin{align}
	\innerprod{ \grad f(\mb q) }{ \frac{1}{q_j} \mb e_j - \frac{1}{q_i} \mb e_i }	\;\geq \; c_4 \frac{\theta(1-\theta)}{n} \frac{\xi}{1+\xi},
\label{eq:final RC2}
\end{align}
holds for all $\mb q \in \mc S_\xi^{i+}$ and any $q_j$ such that $j \neq i$ and $q_j^2\geq \frac{1}{3}q_i^2 $. Here, $c_0$, $c_1$, $c_2$, $c_3$, $c_4$, and $C$ are positive numerical constants.

\end{proposition}

\begin{proof}
Without loss of generality, it is enough to consider the case $i = n$. For all $\mb q \in \mc S_\xi^{n+}$, we have
\begin{align*}
    & \innerprod{ \grad f(\mb q)  }{\frac{1}{q_j} \mb e_j - \frac{1}{q_n} \mb e_n} \\
   =\;&   \innerprod{ \grad f(\mb q) - \grad \wt{f}(\mb q) +  \grad \wt{f}(\mb q) - \grad \bb E\brac{\wt{f}(\mb q) } + \grad \bb E\brac{\wt{f}(\mb q) } }{\frac{1}{q_j} \mb e_j - \frac{1}{q_n} \mb e_n } \\
   \geq \;& \innerprod{  \grad \bb E\brac{\wt{f}(\mb q) } }{ \frac{1}{q_j} \mb e_j - \frac{1}{q_n} \mb e_n } - \abs{\innerprod{ \grad f(\mb q) - \grad \wt{f}(\mb q)}{ \frac{1}{q_j} \mb e_j - \frac{1}{q_n} \mb e_n } } \\
   & -\abs{\innerprod{\grad \wt{f}(\mb q) - \grad \bb E\brac{\wt{f}(\mb q) } }{ \frac{1}{q_j} \mb e_j - \frac{1}{q_n} \mb e_n } }.
\end{align*}
From Proposition \ref{prop:orth_stable_manifold_population}, when $\theta \geq \frac{1}{n}$ and $\mu \leq \frac{c_0}{\sqrt{n}}  $, we know that
\begin{align*}
	 \innerprod{  \grad \bb E\brac{\wt{f}(\mb q) } }{\frac{1}{q_j} \mb e_j - \frac{1}{q_n} \mb e_n }  \;\geq \; \frac{\theta(1-\theta)}{4n} \frac{\xi}{1+\xi}
\end{align*}
holds for all $\mb q \in \mc S_\xi^{n+}$ and any $q_j$ such that $q_j^2\geq \frac{1}{3}q_i^2 $. On the other hand, by Corollary \ref{cor:gradient-concentration}, when $p \geq C_1 \theta^{-2} \xi^{-2} n^5 \log \paren{ \frac{ \theta n }{ \mu} } $, we have
\begin{align*}
     \abs{\innerprod{\grad \wt{f}(\mb q) - \grad \bb E\brac{\wt{f}(\mb q) } }{ \frac{1}{q_j} \mb e_j - \frac{1}{q_n} \mb e_n } } 
    \;\leq\;& \norm{\grad \wt{f}(\mb q) - \grad \bb E\brac{\wt{f}(\mb q) }}{} \cdot \norm{ \frac{1}{q_j} \mb e_j - \frac{1}{q_n} \mb e_n }{} \\
    \;\leq\;& \frac{\theta(1-\theta)}{12n} \frac{\xi}{1+\xi}
\end{align*}
holds for all $\mb q \in \mc S_\xi^{n+}$ with probability at least $ 1- n p^{-c_2\theta n} -n \exp\paren{ -c_3 n^2 }$. For the last inequality, we used the fact that
\begin{align*}
   	 \norm{ \frac{1}{q_j} \mb e_j - \frac{1}{q_n} \mb e_n }{} \;=\; \sqrt{ \frac{1}{q_j^2} + \frac{1}{q_n^2} } \;\leq\; 2 \sqrt{n}.
\end{align*}
Moreover, from Proposition \ref{prop:preconditioning}, we know that when $p \geq C \frac{\kappa^8 n^4}{  \mu^2 \theta^3  \sigma_{\min}^2 \xi^2 }  \log^4 n \log \paren{ \frac{ \theta n }{ \mu } }$
\begin{align*}
    \abs{\innerprod{ \grad f(\mb q) - \grad \wt{f}(\mb q)}{ q_n\mb q - \mb e_n } } \;&\leq\;\norm{  \grad f(\mb q) - \grad \wt{f}(\mb q) }{} \cdot \norm{ \frac{1}{q_j} \mb e_j - \frac{1}{q_n} \mb e_n }{} \\
    \;&\leq\; \frac{\theta(1-\theta)}{12n} \frac{\xi}{1+\xi}
\end{align*}
holds for all $\mb q \in \mc S_\xi^{n+}$ with probability at least $1 -c_4p^{-c_5 n \theta } - n^{-c_6} - n e^{-c_7 \theta np } $. By combining all the bounds above, we obtain the desired result.
\end{proof}

Finally, we prove that the Riemannian gradient of $f(\mb q)$ are uniformly bounded over the sphere.

\begin{proposition}[Bounded gradient]\label{prop:grad-bound-precond}
Suppose $\theta \geq \frac{1}{n}$ and $\mu \leq \frac{c_0}{\sqrt{n}} $. For any index $i \in [n]$, when the sample 
\begin{align*}
  p \geq C \max\Brac{n , \frac{\kappa^8 }{ \theta \mu^2 \sigma_{\min}^2 } \log^4 n } \theta^{-2} n 	\log \paren{ \frac{ \theta n }{ \mu} },
\end{align*}
with probability at least $1- n^{-c_1} - c_2 np^{-c_3 n \theta }$ over the randomness of $\Brac{\mb x_i}_{i=1}^p$, we have
\begin{align}
	 \abs{\innerprod{ \grad f(\mb q) }{\mb e_i}  } \;&\leq \; 2,\label{eqn:grad-infty-bound-precond}\\
     \norm{\grad f(\mb q)}{} \;&\leq \; 2\sqrt{ \theta n }. \label{eqn:grad-bound-precond}
\end{align}
holds for all $\mb q \in \bb S^{n-1}$ and any index $i \in [n]$. Here, $c_0$, $c_1$, $c_2$, $c_3$ and $C$ are positive numerical constants.

\end{proposition}

\begin{proof}
For any index $i \in [n]$, we have
\begin{align*}
    \sup_{q \in \bb S^{n-1} }\;	 \abs{\innerprod{ \grad f(\mb q) }{\mb e_i}  } 
    \;&\leq\; \sup_{q \in \bb S^{n-1} }\;	\abs{\innerprod{ \grad \wt{f}(\mb q) }{\mb e_i} } +  \sup_{q \in \bb S^{n-1} }\;\abs{\innerprod{\grad f(\mb q) - \grad \wt{f}(\mb q)   }{\mb e_i}  } \\
    \;&\leq\; \sup_{q \in \bb S^{n-1} }\;	\abs{\innerprod{ \grad \wt{f}(\mb q) }{\mb e_i} } +  \norm{ \grad f(\mb q) - \grad \wt{f}(\mb q)     }{}.
\end{align*} 
By Corollary \ref{cor:gradient-bound-infty}, when $p \geq C_1 n \log \paren{ \frac{ \theta n }{\mu} } $, we have
\begin{align*}
   	\sup_{q \in \bb S^{n-1} }\;	\abs{\innerprod{ \grad \wt{f}(\mb q) }{\mb e_i} } \; \leq \; \frac{3}{2}
\end{align*}
holds for any index $i \in [n]$ with probability at least $1 - np^{-c_1\theta n} - n\exp\paren{-c_2p}$. On the other hand, Proposition \ref{prop:preconditioning} implies that, when $p \geq C_2 \frac{\kappa^8 n}{  \mu^2 \theta \sigma_{\min}^2 }  \log^4 n \log \paren{ \frac{ \theta n }{ \mu } },$, we have
\begin{align*}
   \norm{ \grad f(\mb q) - \grad \wt{f}(\mb q)  }{}	\;\leq\; \frac{1}{2},
\end{align*}
holds with probability at least $1 -c_3p^{-c_4 n \theta } - n^{-c_5} - n e^{-c_6 \theta np } $. Combining the bounds above gives \eqref{eqn:grad-infty-bound-precond}. The bound \eqref{eqn:grad-bound-precond} can be proved in a similar fashion.
\end{proof}

%% file: sec/app_convergence.tex
In this section, we prove the convergence result of proposed two-stage optimization method for Huber-loss stated in \Cref{subsec:algorithm}. Firstly, we prove that the vanilla RGD converges to an approximate solution in polynomial steps with linear rate. Second, we show linear convergence of subgradient method to the target solution, which solves Phase-2 LP rounding problem.

Our analysis leverages on the geometric properties of the optimization landscape showed in Appendix \ref{app:geometry-main}. Namely, our following proofs are based on the results in Proposition \ref{prop:regularity-precond}, Proposition \ref{prop:grad_orth-precond}, and Proposition \ref{prop:grad-bound-precond} (i.e., \eqref{eq:final RC1-near}, \eqref{eq:final RC1-far}, \eqref{eqn:grad-infty-bound-precond}, and \eqref{eqn:grad-bound-precond}) holding for the rest of this section.

\subsection{Proof of linear convergence for vanilla RGD}

First, assuming the geometric properties in Appendix \ref{app:geometry-main} hold, we show that starting from a random initialization, optimizing
\begin{align}\label{eqn:huber-loss-app}
   \min_{\mb q}\; f(\mb q) \;=\; \frac{1}{np} \sum_{i=1}^p H_\mu \paren{ \mb C_{\mb x_i} \mb R \mb Q^{-1} \mb q }, \qquad \text{s.t.} \quad \mb q \in \bb S^{n-1}
\end{align}
via vanilla RGD in \eqref{eqn:grad-descent}
\begin{align*}
   \mb q^{(k+1)} = \mc P_{\bb S^{n-1}} \paren{ \mb q^{(k)} - \tau \cdot \grad f(\mb q^{(k)} ) }
\end{align*}
recovers an approximate solution with linear rate.

\begin{theorem}[Linear convergence of RGD] \label{thm:linear convergence of grad}
Given an initialization $\mb q^{(0)} \sim \mc U(\bb S^{n-1})$ uniform random drawn from the sphere, choose a stepsize
\begin{align*}
   \tau \;=\; c\min\Brac{ \frac{1}{n^{5/2}},\frac{\mu}{n} },
\end{align*}
then the vanilla gradient descent method for \eqref{eqn:huber-loss} produces a solution  
\begin{align*}
	\norm{ \mb q^{(k)} - \mb e_i }{} \le 2\mu
\end{align*}
for some $i \in [n]$, whenever 
\begin{align*}
   k\;\geq\; K\;:=\; \frac{C}{\theta} \max\Brac{ n^4, \frac{n^{5/2}}{\mu }  }  \log\paren{ \frac{1}{\mu} }.
\end{align*}
\end{theorem}

\begin{proof}[Proof of \Cref{thm:linear convergence of grad}]

\paragraph{Initialization and iterate stays within the region.} First, from Lemma \ref{lem:initialization}, we know that when $\xi = \frac{1}{5\log n }$, with probability at least $1/2$, our random initialization $\mb q^{(0)}$ falls into one of the sets $\Brac{\mc S_\xi^{1+},\mc S_\xi^{1-}, \ldots, \mc S_\xi^{n+}, \mc S_\xi^{n-}}$. Without loss of generality, we assume that $\mb q^{(0)} \in \mc S_\xi^{n+} $. 

Once $\mb q^{(0)}$ initialized within the region $\mc S_\xi^{n+}$, from Lemma \ref{lem: within region}, whenever the stepsize $\tau \leq c_0/\sqrt{n}$, we know that our gradient descent stays within the region $\mc S_\xi^{n+}$ when the stepsize $\tau \leq c_1/\sqrt{n} $ for some $c_1>0$. Based on this, to complete the proof, we now proceed by proving the following results.

\paragraph{Linear convergence until reaching $\norm{\mb q - \mb e_n}{} \leq \mu $.} From Proposition \ref{prop:regularity-precond}, there exists some numerical constant $\gamma \in (\mu, 1) $, such that the regularity condition
\begin{align}
\innerprod{ \grad f(\mb q)  }{q_n\mb q - \mb e_n } \; &\geq \;  \underbrace{c_2 \theta (1-\theta) n^{-3/2} }_{ \alpha_1 } \cdot \norm{ \mb q - \mb e_n }{},\quad \sqrt{1-q_n^2} \in \brac{ \gamma, \sqrt{ \frac{n-1}{n} } }, \label{eqn:regularity-alpha-1} \\
  	\innerprod{ \grad f(\mb q)  }{q_n\mb q - \mb e_n } \; &\geq \; \underbrace{ c_2' \theta (1-\theta)  }_{ \alpha_2 } \cdot \norm{ \mb q - \mb e_n }{}  ,\quad \sqrt{1-q_n^2} \in [\mu, \gamma ], \label{eqn:regularity-alpha-2}
\end{align}
holds w.h.p. for all  $\mb q \in \mc S_\xi^{n+} $. As $\alpha_2\geq \alpha_1$, the regularity condition holds for all $\mb q$ with $\alpha = \alpha_1$. Select a stepsize $\tau$ such that $\tau \leq \gamma \frac{\alpha_1 }{ 2\sqrt{2} \theta n }$. By Lemma~\ref{lem:RC decay} and the regularity condition \eqref{eqn:regularity-alpha-1}, we have
\begin{align*}
   \norm{ \mb q^{(k)} - \mb e_n }{}^2 - \frac{\gamma^2}{2} 	\;\leq\; \paren{1 - \tau \alpha_1}^k \brac{ \norm{ \mb q^{(0)} - \mb e_n }{}^2 - \frac{\gamma^2}{2} } \leq 2 \paren{ 1 - \tau \alpha_1 }^k,
\end{align*}
where the last inequality utilizes the fact that $\norm{\mb q^{(0)} -  \mb e_n}{}^2 \leq 2$. This further implies that
\begin{align*}
  1-q_n^2\;\leq\; \norm{ \mb q^{(k)} - \mb e_n }{}^2 \;\leq\; \frac{\gamma^2}{2} + 2 \paren{ 1 - \tau \alpha_1 }^k \leq \gamma^2,
\end{align*}
when 
\begin{align*}
   2\paren{ 1 - \tau \alpha_1 }^k \leq \frac{\gamma^2}{2} \quad \Longrightarrow\quad k \geq K_1 := \frac{ \log\paren{ \gamma^2/4 } }{ \log \paren{ 1- \tau \alpha_1 } }.
\end{align*}
This implies that $\sqrt{1-q_n^2} \leq \gamma $ for $\forall \;k \geq K_1$. Thus, from \eqref{eqn:regularity-alpha-2}, we know that the regularity condition holds with $\alpha = \alpha_2$. Choose stepsize $\tau \leq \frac{ \mu \alpha_2 }{2\sqrt{2} \theta n } $, apply Lemma~\ref{lem:RC decay} again with $\alpha = \alpha_2$, for all $k \geq 1$, we have
 \begin{align*}
    \norm{ \mb q^{(K_1 + k)} - \mb e_n }{}^2 - \frac{\mu^2}{2}  \;\leq\; \paren{1 - \tau \alpha_2}^k \paren{ \norm{\mb q^{(0)} -  \mb e_n}{}^2 - \frac{\mu^2}{2} } \;\leq\; \paren{\gamma^2 - \mu^2}\paren{1 - \tau \alpha_2}^k.
  \end{align*} 
This further implies that
\begin{align*}
 	\norm{ \mb q^{(K_1 + k)} - \mb e_n }{}^2 \;\leq\; \frac{\mu^2}{2} + \paren{\gamma^2 - \frac{\mu^2}{2}}\paren{1 - \tau \alpha_2}^k \;\leq\; \mu^2
\end{align*}
whenever
\begin{align*}
   \paren{\gamma^2 - \frac{\mu^2}{2}}\paren{1 - \tau \alpha_2}^k\; \leq \; \frac{ \mu^2 }{2} \quad \Longrightarrow\quad k \geq K_2 := \frac{ \log\paren{ \mu^2/\paren{ 2\gamma^2- \mu^2 }  } }{ \log \paren{ 1- \tau \alpha_2 } }.
\end{align*}

Therefore, combining the results above, by using the fact that $\alpha_1 = c_2 \theta (1-\theta) n^{-3/2}$ and $\alpha_2= c_2' \theta (1-\theta)$, we have $\norm{ \mb q^{(k)} - \mb e_n }{} \leq \mu$ whenever 
\begin{align*}
   	\tau \leq \min\Brac{  \frac{\gamma \alpha_1 }{ 2\sqrt{2} \theta n }, \frac{ \mu \alpha_2 }{2\sqrt{2} \theta n }} \;=\; C\min\Brac{ \frac{1}{n^{5/2}},\frac{\mu}{n} } 
\end{align*}
and $k \;\geq\; K:= K_1+K_2$ with
\begin{align*}
    K \;&=\;  \frac{ \log\paren{ 4/\gamma^2 } }{ \log \paren{ (1- \tau \alpha_1)^{-1} } }	\;+\; \frac{ \log\paren{ \paren{ 2\gamma^2- \mu^2 }/\mu^2  } }{ \log \paren{ (1- \tau \alpha_2)^{-1} } } \\
   \;&\leq \; \frac{c_3}{ \tau \alpha_1 }  + \frac{c_4}{\tau \alpha_2} \log\paren{ \frac{1}{\mu} } \;\leq \; \frac{c_5}{\theta} \max\Brac{ n^4, \frac{n^{5/2}}{\mu }  }  \log\paren{ \frac{1}{\mu} },
\end{align*}
where we used the fact that $\log^{-1}\paren{(1-x)^{-1}} \leq 2/x $ for small $x$.

\paragraph{No jump away from an approximate solution $\mb e_n$.}

Finally, we show that once our iterate reaches the region
\begin{align*}
   \mc S \;:=\; \Brac{ \mb q \in \bb S^{n-1} \mid \norm{\mb q - \mb e_n }{} \leq 2 \mu  },
\end{align*}
it will stay within the region $\mc S$, such that our final iterates will always stay close to an approximate solution $\mb e_n$. Towards this end, suppose $\mb q^{(k)} \in \mc S$. Therefore two possibilities: (i) $\mu \leq \norm{\mb q^{(k)} - \mb e_n }{} \leq 2\mu $ (ii) $ \norm{\mb q^{(k)} - \mb e_n }{} \leq \mu $. If the case (i) holds, then our argument above implies that $\norm{ \mb q^{(k+1)} - \mb e_n }{} \leq \norm{ \mb q^{(k)} - \mb e_n }{} \leq 2\mu  $. Otherwise $ \norm{\mb q^{(k)} - \mb e_n }{} \leq \mu $, for which we have
\begin{align*}
    \norm{ \mb q^{(k+1)} - \mb e_n }{}  \;&\leq\; \norm{\mb q^{(k)} - \tau \grad f(\mb q)  -  \mb e_n }{}  \\
    \;& \leq\; \norm{\mb q^{(k)}  -  \mb e_n }{} \;+\; \tau \norm{ \grad f(\mb q)  }{} \; \leq\; \mu + 2 \tau \sqrt{ \theta n } \;\leq\;  2\mu,
\end{align*}
where we used the fact that $\tau \leq \frac{\mu }{ \sqrt{\theta n} } $. Thus, by induction, we have $\mb q^{(k')}\in \mc S$ for all future iterates $k' = k+1,k+2,\cdots$. This completes the proof.
\end{proof}

\begin{lemma}
For any $\mb q\in \mc S_\xi^{n+}$, we have
\begin{align*}
  1 - q_n^2 \;\leq \; \norm{ \mb q - \mb e_n }{}^2  \;\leq\; 2 \paren{1 - q_n^2}\;\leq \;2.
\end{align*}
\end{lemma}
\begin{proof}
We have
\begin{align*}
   1 - q_n^2 \;\leq\; \norm{ \mb q - \mb e_n }{}^2 \;=\; \norm{\mb q_{-n}}{}^2 + (1-q_n)^2\norm{ \mb e_n }{}^2 \; = \; 2(1-q_n) \; =\; 2\frac{1-q_n^2}{1+q_n^2} \;\leq\; 2(1-q_n^2)
\end{align*}
as desired.
\end{proof}

\begin{lemma}[Random initialization falls into good region] \label{lem:initialization}
 Let $\mb q^{(0)} \sim \mc U(\bb S^{n-1})$ be uniformly random generated from the unit sphere $\bb S^{n-1}$. When $\xi= \frac{1}{5\log n} $, then with probability at least $1/2$, $\mb q^{(0)}$ belongs to one of the $2n$ sets $\Brac{\mc S_\xi^{1+},\mc S_\xi^{1-}, \ldots, \mc S_\xi^{n+}, \mc S_\xi^{n-}}$. The set $\mb q^{(0)}$ belongs to is uniformly at random.
\end{lemma}

\begin{proof}
We refer the readers to Lemma 3.9 of \cite{bai2018subgradient} and Theorem 1 of \cite{gilboa2018efficient} for detailed proofs.
\end{proof}

\begin{lemma}[Stay within the region $\mc S_\xi^{n+}$] \label{lem: within region}
Suppose $\mb q^{(0)}\in \mc S_\xi^{n+}$ with $\xi \le 1$. There exists some constant $c>0$, such that when the stepsize satisfies $\tau \leq \frac{c}{\sqrt{n}}$, our Riemannian gradient iterate $\mb q^{(k)} = \mc P_{\bb S^{n-1}} \paren{ \mb q^{(k-1)} - \tau \cdot \grad  f(\mb q^{(k-1)}) }  $ satisfies $ \mb q^{(k)} \in \mc S_\xi^{n+}$ for all $k \geq 1$.
\end{lemma}

\begin{proof} We prove this by induction. For any $k\geq 1$, suppose $\mb q^{(k)} \in \mc S_\xi^{n+}$. For convenience, let $\mb g^{(k)} = \grad f(\mb q^{(k)})$. Then, for any $j \not = k$, we have
\begin{align*}
    \paren{ \frac{ q_n^{(k+1)}  }{ q_j^{(k+1)}  } }^2 \;=\; \paren{ \frac{ q_n^{(k)} - \tau g_n^{(k)} }{ q_j^{(k)} - \tau g_j^{(k)} }  }^2.
\end{align*}
We proceed by considering the following two cases.

\paragraph{Case (i): $\abs{q_n^{(k)}/q_j^{(k)}} \geq \sqrt{3}$.}	 In this case, we have
\begin{align*}
   	\paren{ \frac{ q_n^{(k+1)}  }{ q_j^{(k+1)}  } }^2 \;=\; \paren{ \frac{ q_n^{(k)} - \tau g_n^{(k)} }{ q_j^{(k)} - \tau g_j^{(k)} }  }^2 \;\geq\; \paren{ \frac{1 - \tau \cdot g_n^{(k)} / q_n^{(k)} }{ q_j^{(k)} / q_n^{(k)} - \tau g_j^{(k)} / q_n^{(k)} } }^2\; \geq \; \paren{ \frac{1- 2 \tau \sqrt{n}}{ 1/\sqrt{3} +2 \tau \sqrt{n} }  }^2\;\geq \; 2,
\end{align*}
where the second inequality utilizes \eqref{eqn:grad-infty-bound-precond} and the fact $q_n^{(k)} \geq \frac{1}{\sqrt{n}}$, and the last inequality follows when $\tau \leq \frac{\sqrt{3} - \sqrt{2}}{2(\sqrt{6} + \sqrt{3})}\frac{1}{\sqrt{n}}$.

\paragraph{Case (ii): $\abs{q_n^{(k)}/q_j^{(k)}} \leq \sqrt{3}$.} Proposition \ref{prop:regularity-precond} and Proposition \ref{prop:grad_orth-precond} implies that
\begin{align}\label{eqn:stay-region-1}
   	\frac{g_j^{(k)}}{q_j^{(k)}}\;\geq\; 0, \quad \frac{g_j^{(k)}}{q_j^{(k)}} - \frac{g_n^{(k)}}{q_n^{(k)}} \;\geq\; 0.
\end{align}
By noting that $\abs{q_j^{(k)} }\geq \abs{q_n^{(k)}} /\sqrt{3}\ge 1/\sqrt{3n}$ and $\abs{g_j^{(k)}}\leq 2$, we have
\begin{align}\label{eqn:stay-region-2}
   \tau \;\leq\; \frac{1}{2\sqrt{3n}} \;\leq\; \frac{ q_j^{(k)} }{ g_j^{(k)} } \quad \Longrightarrow\quad \tau \cdot  \frac{ g_j^{(k)} }{ q_j^{(k)} } \;\leq\; 1.
\end{align}
Thus, we have
\begin{align*}
   \paren{ \frac{ q_n^{(k+1)}  }{ q_j^{(k+1)}  } }^2 \;&=\; \paren{ \frac{ q_n^{(k)} }{ q_j^{(k)} } }^2	 \paren{ 1+ \tau \cdot \frac{ g_j^{(k)}/ q_j^{(k)} - g_n^{(k)}/q_n^{(k)} }{ 1 - \tau g_j^{(k)}/ q_j^{(k)} }  }^2 \\
   \;&\geq\; \paren{ \frac{ q_n^{(k)} }{ q_j^{(k)} } }^2 \paren{ 1+ \tau \cdot \paren{ \frac{g_j^{(k)}}{ q_j^{(k)} } - \frac{ g_n^{(k)} }{ q_n^{(k)} } } }^2\;\geq\;\paren{ \frac{ q_n^{(k)} }{ q_j^{(k)} } }^2 \paren{1 + \tau \cdot \frac{\theta(1-\theta)}{4n} \frac{\xi}{1+\xi}}^2.
\end{align*}
The first inequality follows from \eqref{eqn:stay-region-1} and \eqref{eqn:stay-region-2}, and 
the second inequality directly follows from Proposition \ref{prop:grad_orth-precond}. Therefore, when $\xi \leq 1$, this implies that $\mb q^{(k+1)} \in \mc S_\xi^{n+}$. By induction, this holds for all $k\geq 1$.
\end{proof}

In the following, we show that the iterates get closer to $\mb e_n$.
\begin{lemma}[Iterate contraction]\label{lem:RC decay}
For any $\mb q\in \mc S_\xi^{n+}$, assuming the following regularity condition 
\begin{align}\label{eq:RC general}
    \innerprod{\grad f(\mb q)}{ q_i \mb q - \mb e_n } \;\geq \; \alpha \norm{\mb q - \mb e_n}{}	
\end{align}
holds for a parameter $\alpha>0$. Then if $\mb q^{(k)} \in \mc S_\xi^{n+}$ and the stepsize $\tau \leq c\frac{\alpha}{\theta n} $, the iterate $\mb q^{(k+1)} = \mc P_{\bb S^{n-1}}\paren{  \mb q - \tau \cdot \grad f(\mb q)} $ satisfies
\begin{align*}
   \norm{ \mb q^{(k+1)} - \mb e_n }{}^2 -  \paren{\frac{2\tau \theta n}{\alpha}}^2 \;\leq\;\paren{ 1 - \tau \alpha  } \brac{ \norm{ \mb q^{(k)} - \mb e_n }{}^2 - \paren{\frac{2\tau \theta n}{\alpha}}^2 }.
\end{align*}
\end{lemma}

\begin{proof} First, note that
\begin{align*}
   \norm{ \mb q^{(k+1)} - \mb e_n }{}^2 \;&=\; \norm{ \mc P_{\bb S^{n-1}} \paren{\mb q^{(k)} - \tau \cdot \grad f(\mb q^{(k)}) } -  \mc P_{\bb S^{n-1}} (\mb e_n ) }{}^2 \\
   \;&\leq \; \norm{ \mb q^{(k)} - \tau \cdot \grad f(\mb q^{(k)}) -  \mb e_n  }{}^2 \\
   \;&= \; \norm{ \mb q^{(k)} - \mb e_n }{}^2 - 2\tau \cdot \innerprod{ \grad f(\mb q^{(k)}) }{ \mb q^{(k)} - \mb e_n } + \tau^2 \norm{ \grad f(\mb q^{(k)}) }{}^2 \\
   \;&\leq \; \norm{ \mb q^{(k)} - \mb e_n }{}^2 - 2 \tau \alpha \norm{ \mb q^{(k)} - \mb e_n }{} + 4 \tau^2 \theta n,
\end{align*}
where the first inequality utilizes the fact that $\mc P_{\bb S^{n-1}} (\cdot)$ is $1$-Lipschitz continuous, and the last line follows from  \eqref{eq:RC general} and \eqref{eqn:grad-bound-precond} in Proposition \ref{prop:grad-bound-precond}. We now subtract both sides by $\paren{\frac{2\tau \theta n}{\alpha}}^2$,
\begin{align*}
    \norm{ \mb q^{(k+1)} - \mb e_n }{}^2 - \paren{\frac{2\tau \theta n}{\alpha}}^2 \;&\leq\; \norm{\mb q^{(k)} -  \mb e_n}{}^2- \paren{\frac{2\tau \theta n}{\alpha}}^2 - 2\tau \alpha \paren{ \norm{\mb q^{(k)} - \mb e_n}{} - \frac{2\tau \theta n}{\alpha} }	\\
    \;&=\; \brac{ 1 - 2\tau \alpha \paren{ \norm{\mb q^{(k)} - \mb e_n}{} + \frac{2\tau \theta n}{\alpha} }^{-1} } \brac{ \norm{\mb q^{(k)} - \mb e_n}{}^2 - \paren{\frac{2\tau \theta n}{\alpha}}^2 } \\
    \;&\leq \; \paren{ 1 - \tau \alpha } \brac{ \norm{\mb q^{(k)} - \mb e_n}{}^2 - \paren{\frac{2\tau \theta n}{\alpha}}^2 },
\end{align*}
where the last inequality follows because 
\begin{align*}
   \norm{ \mb q^{(k)} - \mb e_n }{}^2 \;\leq\; 2,\quad \tau \leq \paren{1 - \frac{1}{\sqrt{2}}}\frac{ \alpha }{ \theta n },
\end{align*}
such that 
\begin{align*}
  	\norm{\mb q - \mb e_n}{}  + \frac{2\tau \theta n}{\alpha} \;\leq\; 2.
\end{align*}
This completes the proof.
\end{proof}

\subsection{Proof of exact recovery via LP rounding}

To obtain exact solutions, we use the approximate solution $\mb q_\star$ from Phase-1 gradient descent method as a warm start $\mb r = \mb q_\star$, and consider solving a \emph{convex} Phase-2 LP rounding problem introduced in \eqref{eqn:LP-rounding}
\begin{align*}
   \min_{\mb q} \; \zeta(\mb q) \;:=\;\frac{1}{np} \sum_{i=1}^p \norm{ \mb C_{\mb x_i} \mb R \mb Q^{-1} \mb q }{1},\quad \text{s.t.}\quad \innerprod{ \mb r }{ \mb q } \;=\;1.
\end{align*}
In the following, we show the function is sharp around \cite{burke1993weak,li2019incremental} the target solution, so that projected subgradient descent methods converge linearly to the truth with geometrically decreasing stepsizes.

\subsubsection{Sharpness of the objective function.}

\begin{proposition}\label{prop:sharp LP with en}
Suppose $\theta\in\paren{ \frac{1}{n}, \frac{1}{3}}$ and $\mb r$ satisfies
\begin{align}\label{eqn:r_n-constraint}
	\frac{ \norm{\mb r_{-n} }{} }{ r_n } \leq \frac{1}{20}.
\end{align}
Whenever $p \;\geq\; C  \frac{ \kappa^8 }{\theta  \sigma_{\min}^{2}(\mb C_{\mb a}) } \log^3 n$, with probability at least $ 1- p^{-c_1 n \theta } - n^{-c_2}$, the function $\zeta(\mb q)$ is sharp in a sense that
\begin{align}
 \zeta(\mb q) - \zeta\paren{  \paren{\mb R\mb Q^{-1}}^{-1} \frac{ \mb e_n}{\wt{r}_n} } \;\geq\; \frac{1}{50}\sqrt{\frac{2}{\pi}} \theta  \norm{\mb q -  \paren{\mb R\mb Q^{-1}}^{-1}\frac{\mb e_n}{\wt{r}_n} }{}
\label{eq:sharp LP with en}
\end{align}
for any feasible $\mb q$ with $\innerprod{\mb r}{\mb q}=1$. Here, $\wt{\mb r} = \paren{ \mb R\mb Q^{-1} }^{-\top}\mb r$. 
\end{proposition}

\begin{proof} Let us denote $\wt{\mb q} = \mb R\mb Q^{-1} \mb q$. Then we can rewrite our original problem as
\begin{align*}
   \min_{\wt{\mb q} }\;  \wt\zeta(\wt{\mb q} )\;=\;\frac{1}{np} \sum_{i=1}^p \norm{ \mb C_{\mb x_i} \wt{\mb q} }{1} \quad\text{s.t.} \quad \innerprod{\wt{\mb r}}{\wt{\mb q}}= 1,	
\end{align*}
which is reduced to the orthogonal problem in \eqref{eq:LP rounding no precond} of Lemma \ref{lem:sharp LP no precond}. To utilize the result in Lemma~\ref{lem:sharp LP no precond}, we first prove that $\wt {\mb r}$ satisfies \eqref{eq:LP condition on r} if $\mb r$ satisfies \eqref{eqn:r_n-constraint}. Towards that end, note that
\begin{align*}
   \wt{\mb r} \;=\;	 \paren{ \mb R\mb Q^{-1} }^{-\top} \mb r \;=\; \mb r  + \paren{\paren{ \mb R\mb Q^{-1} }^{-\top} - \mb I } \mb r.
\end{align*}
By Lemma \ref{lem:precond-bounds}, we know that, for any $\delta \in (0,1)$, whenever $p \;\geq\; C  \frac{ \kappa^8 }{\theta \delta^2 \sigma_{\min}^{2}(\mb C_{\mb a}) } \log^3 n$,
\begin{align*}
   	\norm{\paren{\paren{ \mb R\mb Q^{-1} }^{-\top} - \mb I } \mb r }{} \;\leq\; \norm{ \paren{ \mb R\mb Q^{-1} }^{-1} - \mb I   }{} \norm{\mb r}{} \;\leq\; 2\delta \norm{\mb r}{}
\end{align*}
holds with probability at least $ 1- p^{-c_1 n \theta } - n^{-c_2}$. This further implies that
\begin{align*}
   \wt{r}_n \;\geq\;r_n - 2\delta \norm{\mb r}{}, 	\quad \norm{ \wt{\mb r}_{-n} }{}\;\leq\; \norm{ \mb r_{-n} }{} + 2\delta \norm{ \mb r }{}.
\end{align*}
Therefore, by choose $\delta$ sufficiently small, we have
\begin{align*}
   \frac{ \norm{ \wt{\mb r}_{-n} }{} }{ \wt{r}_n }	\;\leq\; \frac{  \norm{ \mb r_{-n} }{}  + 2\delta \norm{\mb r}{} }{  r_n - 2\delta \norm{\mb r}{} }\;=\; \frac{ \norm{\mb r_{-n}}{}/r_n + 2\delta \sqrt{ 1+ \paren{\norm{\mb r_{-n}}{}/r_n}^2 } }{ 1 - 2\delta \sqrt{ 1+ \paren{\norm{\mb r_{-n}}{}/r_n}^2 } } \;\leq \; \frac{1}{10},
\end{align*}
where the last inequality follows from \eqref{eqn:r_n-constraint}. Therefore, by Lemma~\ref{lem:sharp LP no precond}, we obtain
\begin{align*}
   \zeta(\mb q) - \zeta\paren{  \paren{\mb R\mb Q^{-1}}^{-1} \frac{ \mb e_n}{\wt{r}_n} } \;&=\; \wt{\zeta}(\mb q) - \wt{\zeta}\paren{ \frac{ \mb e_n }{\wt{r}_n} } \\
   \;&\geq \; \frac{1}{25} \sqrt{ \frac{2}{\pi} } \theta \norm{ \wt{\mb q} - \frac{\mb e_n}{ \wt{r}_n } }{} \\
   \;&= \; \frac{1}{25} \sqrt{ \frac{2}{\pi} } \theta \norm{ \paren{\mb R\mb Q^{-1}} \cdot  \paren{ \mb q -  \paren{\mb R\mb Q^{-1}}^{-1}\frac{\mb e_n}{\wt{r}_n} }  }{} \\
   \;&\geq \; \frac{1}{25} \sqrt{ \frac{2}{\pi} } \theta \cdot \sigma_{\min} \paren{\mb R\mb Q^{-1} } \cdot \norm{\mb q -  \paren{\mb R\mb Q^{-1}}^{-1}\frac{\mb e_n}{\wt{r}_n} }{} \\
\end{align*}
By Lemma \ref{lem:precond-bounds}, we know that $\norm{ \paren{\mb R\mb Q^{-1}}^{-1} }{} \leq 1+2\delta$, so that 
\begin{align*}
   	\sigma_{\min} \paren{\mb R\mb Q^{-1} } \;=\; \norm{ \paren{\mb R\mb Q^{-1}}^{-1} }{}^{-1} \geq \frac{1}{1+2\delta}.
\end{align*}
Thus, this further implies that
\begin{align*}
   \zeta(\mb q) - \zeta\paren{  \paren{\mb R\mb Q^{-1}}^{-1} \frac{ \mb e_n}{\wt{r}_n} } \;\geq \; \frac{1}{25} \sqrt{ \frac{2}{\pi} } \frac{\theta}{1+2\delta} \cdot \norm{\mb q -  \paren{\mb R\mb Q^{-1}}^{-1}\frac{\mb e_n}{\wt{r}_n} }{},
\end{align*}
as desired.
\end{proof}

\begin{lemma}[Sharpness for the orthogonal case]\label{lem:sharp LP no precond}
Consider the following problem
\begin{align}\label{eq:LP rounding no precond}
   \min_{\mb q}  \wt\zeta(\mb q):=\frac{1}{np} \sum_{i=1}^p \norm{ \mb C_{\mb x_i} \mb q }{1} \quad\text{s.t.} \quad \innerprod{\mb r}{\mb q}= 1,	
\end{align}
with $\mb r\in\bb S^{n-1}$ satisfying
\begin{align}
\frac{\norm{\bm r_{-n} }{}}{r_n} \le \frac{1}{10}, \quad r_n>0.
\label{eq:LP condition on r}\end{align}
 Whenever $ p \geq \frac{C}{\theta^2} n \log \paren{\frac{n}{\theta}  }  $, with probability at least $1- c_1 np^{-6} - c_2 n e^{-c_3\theta^2 p} $, the function $\wt{\zeta}(\mb q)$ is sharp in a sense that
\begin{align*}
	\wt{\zeta}(\mb q) - \wt{\zeta}\parans{\frac{\mb e_n}{r_n}} \;\geq \frac{1}{25}\sqrt{\frac{2}{\pi}} \theta \norm{\mb q - \frac{\mb e_n}{r_n} }{}
\end{align*}
for any feasible $\mb q$ with $\innerprod{\mb r}{\mb q}= 1$.
\end{lemma}
\begin{proof}
Observing that $\innerprod{\mb r }{ \mb q } = \mb r_{-n}^\top \mb q_{-n} + r_n q_n = 1 $, we have
\begin{align*}
   \norm{ \mb r_{-n} }{} \norm{ \mb q_{-n} }{} \;\geq \; \mb r_{-n}^\top \mb q_{-n} \;=\; r_n \paren{ \frac{1}{r_n} - q_n} \;\geq \; r_n \paren{ \frac{1}{r_n} - \abs{q_n}}.
\end{align*}
This further implies that
\begin{align}
   \frac{1}{r_n} - \abs{q_n} \;&\leq \; \frac{ \norm{ \mb r_{-n} }{}  }{r_n } \norm{ \mb q_{-n} }{} \label{eq:LP proof bound dist qn}.
\end{align}
Second, we have
\begin{align*}
\norm{ \mb q - \frac{ \mb e_n }{ r_n} }{} \;= \; \sqrt{\paren{ \frac{1}{r_n} - q_n}^2 + \norm{ \mb q_{-n} }{}^2}  \;&\leq \; \sqrt{1 + \paren{ \frac{ \norm{ \mb r_{-n} }{} }{ r_n } }^2}\norm{ \mb q_{-n} }{},
\end{align*}
which implies that
\begin{align}\label{eq:LP proof bound dist}	
	\paren{1 + \paren{ \frac{ \norm{ \mb r_{-n} }{} }{ r_n } }^2}^{-1/2} \norm{ \mb q - \frac{ \mb e_n }{ r_n} }{} \;\leq \; \norm{ \mb q_{-n} }{}.
\end{align}
We now proceed by considering the following two cases.
\paragraph{Case i: $\abs{q_n} \;\geq\; \frac{1}{r_n}$.} In this case, we have
\begin{align*}
   \wt{\zeta}(\mb q) - \wt{\zeta}\paren{ \frac{\mb e_n}{ r_n } }	\; \geq \; \frac{1}{6}\sqrt{\frac{2}{\pi}} \theta \norm{ \mb q_{-n} }{} 
   \;&\geq \; \frac{1}{6}\sqrt{\frac{2}{\pi}} \theta \paren{ 1 + \paren{ \frac{ \norm{ \mb r_{-n} }{} }{ r_n } }^2 }^{-1/2}  \norm{ \mb q - \frac{ \mb e_n }{ r_n} }{} \\
   \;&\geq \; \frac{5}{33} \sqrt{\frac{2}{\pi}} \theta \norm{\mb q - \frac{\mb e_n}{r_n}}{},
\end{align*}
where the first inequality follows by \eqref{eq:LP proof obj distance}, the second inequality follows by \eqref{eq:LP proof bound dist}, and the last inequality follows because $\frac{\norm{ \bm r_{-n} }{}}{r_n} \le \frac{1}{10}$.

\paragraph{Case ii: $|q_n| \;\leq\; \frac{1}{r_n}$.} In this case, we have
\begin{align*}
     \wt{\zeta}(\mb q) - \wt{\zeta}\paren{ \frac{\mb e_n}{ r_n } }	\; &\geq \; 	\frac{1}{6}\sqrt{\frac{2}{\pi}} \theta  \norm{ \mb q_{-n} }{} - \frac{5}{4}\sqrt{\frac{2}{\pi}} \theta \paren{\frac{1}{r_n} - |q_n|} \\
     \;&\geq \; \theta \paren{ \frac{1}{6}\sqrt{\frac{2}{\pi}}   -  \frac{5}{4}\sqrt{\frac{2}{\pi}} \frac{ \norm{ \mb r_{-n} }{} }{ r_n} } \norm{ \mb q_{-n} }{} \\
     \;&\geq\; \theta \paren{ \frac{1}{6}\sqrt{\frac{2}{\pi}}   -  \frac{5}{4}\sqrt{\frac{2}{\pi}} \frac{ \norm{ \mb r_{-n} }{} }{ r_n} } \paren{ 1 + \paren{ \frac{ \norm{ \mb r_{-n} }{} }{ r_n } }^2 }^{-1/2}  \norm{ \mb q - \frac{ \mb e_n }{ r_n} }{} \\
     \;&\geq \; \frac{\theta}{25}\sqrt{\frac{2}{\pi}}  \norm{\mb q - \frac{\mb e_n}{r_n} }{},
\end{align*}
where the first inequality follows by \eqref{eq:LP proof obj distance}, the second inequality follows from \eqref{eq:LP proof bound dist qn}, the third inequality follows from \eqref{eq:LP proof bound dist}, and the last one follows because $\frac{\norm{\mb r_{-n}}{}}{r_n} \leq \frac{1}{10}$.

Combining the results in both cases, we obtain the desired result.
\end{proof}

\begin{lemma}\label{lem:sharp-orth-concentration}
Suppose $\theta \in \paren{ \frac{1}{n},\frac{1}{3} }$. Whenever $ p \geq \frac{C}{\theta^2} n \log \paren{\frac{n}{\theta}  }  $, we have
\begin{align}
\wt\zeta(\mb q) - \wt\zeta\paren{\frac{\mb e_n}{r_n}} \;\geq\; \begin{cases}\frac{1}{6}\sqrt{\frac{2}{\pi}}\theta \norm{ \mb q_{-n}}{}, & \text{if}\ |q_n| - \frac{1}{r_n} \ge 0,\\
\frac{1}{6}\sqrt{\frac{2}{\pi}}\theta \norm{\ol {\mb q}}{} - \frac{5}{4}\sqrt{\frac{2}{\pi}} \theta \parans{\frac{1}{r_n} - |q_n|}, & \text{if} \ |q_n| - \frac{1}{r_n} < 0,
 \end{cases}
\label{eq:LP proof obj distance}
\end{align}
holds with probability at least $1- c_1 np^{-6} - c_2 n e^{-c_3\theta^2 p} $.
\end{lemma}

\begin{proof}
For each $j \in [n]$, let us define an index set $ \mc I_j:=\left\{i\in[p]: \paren{\shift{\wc{\mb x}_i}{j} }_n \neq 0\right\}$, and let us define events
\begin{align*}
   \mc E:= \bigcap_{j=0}^{n-1} \mc E_j,\quad  \mc E_j := \Brac{ \abs{ \mc I_i } \leq \frac{9}{8} \theta p },\;(0\leq j\leq n-1).
\end{align*}
By Hoeffding's inequality and a union bound, we know that
\begin{align*}
    \bb P\paren{ \mc E^c } \leq \sum_{j=0}^{n-1}	 \bb P\paren{ \mc E_j^c } \leq n \exp\paren{ - p\theta^2/2 }.
\end{align*}
Based on this, we have
\begin{align*}
   &\wt{\zeta}(\mb q) - \wt{ \zeta } \paren{ \frac{\mb e_n}{r_n} } \\
   =\;& \frac{1}{np} \sum_{i=1}^p \norm{ \mb C_{\mb x_i} \mb q }{1} - \frac{1}{np} \frac{1}{r_n} \sum_{i=1}^p \norm{\mb x_i}{1}  \\
     =\;& \frac{1}{np} \sum_{i=1}^p \sum_{j=0}^{n-1} \abs{ \innerprod{ \shift{ \wc{\mb x}_i }{j}  }{ \mb q }  } - \frac{1}{np} \frac{1}{r_n} \sum_{i=1}^p \norm{\mb x_i}{1} \\
     \geq \;& \frac{ 1 }{np}\paren{  \abs{q_n} - \frac{1}{r_n} }  \sum_{i=1}^p \norm{ \mb x_i }{1}   + \frac{1}{np} \sum_{j=0}^{n-1} \paren{ \sum_{i \in \mc I_j^c} \abs{ \innerprod{ (\shift{ \wc{\mb x}_i }{ j })_{-n}  }{ \mb q_{-n} }  } - \sum_{i \in \mc I_j} \abs{ \innerprod{ (\shift{ \wc{\mb x}_i }{ j })_{-n}  }{ \mb q_{-n} }   } } \\
     =\;& \frac{ 1 }{np} \paren{  \abs{q_n} - \frac{1}{r_n} } \sum_{i=1}^p \norm{ \mb x_i }{1} +\frac{1}{np} \sum_{j=0}^{n-1} \paren{ \norm{ \mb q_{-n}^\top \mb M_{\mc I_j^c}^j }{1} -  \norm{ \mb q_{-n}^\top \mb M_{\mc I_j}^j }{1} },
\end{align*}
where we denote $\mb M^j = \begin{bmatrix}  (\shift{ \wc{\mb x}_1 }{ j })_{-n} & (\shift{ \wc{\mb x}_2 }{ j })_{-n} & \cdots & (\shift{ \wc{\mb x}_p }{ j })_{-n} \end{bmatrix}$, and $\mb M_{\mc I}^j$ denote a submatrix of $\mb M^j$ with columns indexed by $\mc I$. Conditioned on the event $\mc E$, by Lemma \ref{lem:random-mtx-BG-l1} and a union bound, whenever $ p \geq \frac{C}{\theta^2} n \log \paren{\frac{n}{\theta}  }  $, we have
\begin{align*}
   \norm{ \mb q_{-n}^\top \mb M_{\mc I_j^c}^j }{1} -  \norm{ \mb q_{-n}^\top \mb M_{\mc I_j}^j }{1} \;\geq\; \frac{p}{6}\sqrt{\frac{2}{\pi}}\theta\norm{\mb q_{-n} }{}, \ \forall \ \mb q_{-n} \in \bb R^{n-1},\;(0\leq j\leq n-1) 
\end{align*}
with probability at least $1- c n p^{-6}$. On the other hand, by Gaussian concentration inequality, we have
\begin{align*}
   \bb P\paren{ \frac{1}{np} \sum_{i=1}^p \norm{ \mb x_i }{1} \geq \frac{5}{4}\sqrt{\frac{2}{\pi}} \theta }	\;\leq\; \exp\paren{ - \frac{ \theta^2 p }{ 64 \pi  } }.
\end{align*}
Therefore, combining all the results above, we have
\begin{align*}
\wt\zeta(\mb q) - \wt\zeta\paren{\frac{\mb e_n}{r_n}} \;\geq\; \begin{cases}\frac{1}{6}\sqrt{\frac{2}{\pi}}\theta \norm{ \mb q_{-n}}{}, & \text{if}\; \abs{q_n} - \frac{1}{r_n} \geq 0,\\
\frac{1}{6}\sqrt{\frac{2}{\pi}}\theta \norm{\ol {\mb q}}{} - \frac{5}{4}\sqrt{\frac{2}{\pi}} \theta \parans{\frac{1}{r_n} - |q_n|}, & \text{if}\;  \abs{q_n} - \frac{1}{r_n} < 0,
 \end{cases}
\end{align*}
as desired.
\end{proof}

\subsection{Linear convergence for projection subgradient descent for rounding}

Now based on the sharpness condition, we are ready to show that the projected subgradient descent method 
\begin{align*}
	\mb q^{(k+1)} \;=\; \mb q^{(k)} - \tau^{(k)} \mc P_{\mb r^\perp} \mb g^{(k)},\quad \mb g^{(k)} \;=\;  \sum_{i=1}^p \paren{\mb R\mb Q^{-1}}^\top \mb C_{\mb x_i}^\top \sign\paren{ \mb C_{\mb x_i} \mb R \mb Q^{-1}  \mb q^{(k)} }.	
\end{align*}
on $\zeta(\mb q)$ converges linearly to the target solution up to a scaling factor. For convenience, let us first define the distance between the iterate and the target solution
\begin{align*}
   d^{(k)} \;:=\; \norm{ \mb s^{(k)} }{},\quad \mb s^{(k)} \;:=\; \mb q^{(k)} - \paren{ \mb R\mb Q^{-1} }^{-1} \frac{\mb e_n}{ \wt{r}_n },
\end{align*}
and several parameters
\begin{align*}
   \alpha := \frac{1}{50}\sqrt{\frac{2}{\pi}} \theta ,\quad \beta := 36 \log(np).
\end{align*}
We show the following result.
\begin{proposition}
	Suppose $\theta \in \paren{\frac{1}{n}, \frac{1}{3} }$ and $\mb r$ satisfies
\begin{align}\label{eqn:r-bound-2}
	\frac{ \norm{\mb r_{-n} }{} }{ r_n } \leq \frac{1}{20}, \quad r_n>0,\quad \norm{\mb r}{} = 1.
\end{align}
Let $\mb q^{(k)}$ be the sequence generated by the projected subgradient method (cf. \Cref{alg:subgradient}) with initialization $\mb q^{(0)} = \mb r$ and geometrically decreasing step size
\begin{align}\label{eqn:subgradient-stepsize}
	\tau^{(k)} \;=\; \eta^k  \tau^{(0)},\quad \tau^{(0)} = \frac{16}{25} \frac{ \alpha }{ \beta^2 } , \quad  \sqrt{ 1 - \frac{ \alpha^2 }{ 2\beta^2 } } \leq \eta <1
\end{align}
Whenever $p \;\geq\; C  \frac{ \kappa^8 }{\theta \sigma_{\min}^{2}(\mb C_{\mb a}) } \log^3 n$, with probability at least $ 1- p^{-c_1 n \theta } - n^{-c_2}$, the sequence $\Brac{\mb q^{(k)}}_{k\geq 0}$ satisfies
\begin{align}
\norm{\mb q^{(k)} - \paren{ \mb R\mb Q^{-1} }^{-1}  \frac{\mb e_n}{\wt r_n} }{}\;\leq\; \frac{2}{5} \eta^k ,\label{eq:linear decay dist}
\end{align}
for all iteration $k =0,1,2,\cdots$.
\end{proposition}

\begin{proof}
Given the initialization $\mb q^{(0)} = \mb r $, we have
\begin{align*}
   d^{(0)} \;=\; \norm{ \mb r - \paren{ \mb R\mb Q^{-1} }^{-1} \frac{\mb e_n}{ \wt{r}_n } }{}
   \; &\leq\; \norm{\paren{ \mb R\mb Q^{-1} }^{-1}}{} \norm{ \wt{\mb r} - \frac{\mb e_n}{ \wt{r}_n }  }{}\\
   \;&\leq \;  \frac{10}{9}  \cdot \paren{ \norm{ \wt{\mb r}_{-n} }{}^2 + \paren{\wt{r}_n - \frac{1}{\wt{r}_n} }^2  }^{1/2},
\end{align*}
where the last inequality we used Lemma \ref{lem:precond-bounds}. From the argument in Proposition \ref{prop:sharp LP with en}, we know that \eqref{eqn:r-bound-2} implies $\norm{ \wt{r}_{-n} }{} / \wt{r}_n \leq 1/10$. By the fact that $\norm{\wt{\mb r}}{}\leq 10/9$, we have
\begin{align}\label{eqn:d-0-bound}
    \norm{ \wt{\mb r}_{-n} }{} \;\leq\; \frac{1}{9}, \quad \abs{ \wt{r}_n - \frac{1}{\wt{r}_n} } \;\leq\; \abs{ \frac{8}{9} - \frac{9}{8} }^2 \;\leq\; \frac{1}{4} \quad \Longrightarrow \quad d^{(0)} \;\leq\; \frac{2}{5}.
\end{align}
On the other hand, notice that
\begin{align*}
   \paren{ d^{(k+1)} }^2	 \; &=\; \norm{ \mb q^{(k)} - \tau^{(k)} \mc P_{\mb r^\perp} \mb g^{(k)} - \paren{ \mb R\mb Q^{-1} }^{-1} \frac{\mb e_n}{ \wt{r}_n } }{}^2 \\
   \;&=\; \paren{ d^{(k)} }^2 - 2  \tau^{(k)} \innerprod{ \mb s^{(k)} }{ \mc P_{\mb r^\perp} \mb g^{(k)} } +  \paren{\tau^{(k)}}^2 \norm{ \mc P_{\mb r^\perp} \mb g^{(k)} }{}^2
\end{align*}
By Lemma \ref{lem:rho-bound}, we know that when $p \;\geq\; C  \frac{ \kappa^8 }{\theta \sigma_{\min}^{2}(\mb C_{\mb a}) } \log^3 n$, for any $k = 1,2,\cdots$,
\begin{align*}
	\norm{ \mc P_{\mb r^\perp} \mb g^{(k)} }{}^2 \;\leq \; 36 \log\paren{np} = \beta 
\end{align*}
holds with probability at least $ 1- p^{-c_1 n \theta } - n^{-c_2}$. On the other hand, by the sharpness property of the function in Proposition \ref{prop:sharp LP with en}, for any $k=1,2,\cdots$, 
\begin{align*}
   \innerprod{ \mb s^{(k)} }{ \mc P_{\mb r^\perp} \mb g^{(k)} }\;= \; 	\innerprod{ \mb s^{(k)} }{  \mb g^{(k)} } \;&\geq \; \zeta\paren{ \mb q^{(k)} } - \zeta\paren{ \paren{ \mb R\mb Q^{-1} }^{-1} \frac{\mb e_n}{ \wt{r}_n } }  \\
   \;&\geq \; \frac{1}{50}\sqrt{\frac{2}{\pi}} \theta  \norm{\mb q^{(k)} -  \paren{\mb R\mb Q^{-1}}^{-1}\frac{\mb e_n}{\wt{r}_n} }{}\;= \; \alpha \cdot d^{(k)},
\end{align*}
where the first equality follows from the fact that $\innerprod{\mb r }{ \mb s^{(k)}  }=0$ so that $ \mc P_{\mb r^\perp} \mb s^{(k)} = \mb s^{(k)} $, the first inequality follows from the fact that $\zeta(\mb q)$ is convex, and the second inequality utilizes the sharpness of the function in Proposition \ref{prop:sharp LP with en} given the condition \eqref{eqn:r-bound-2}. Thus, we have
\begin{align*}
	\paren{ d^{(k+1)} }^2 \; &\leq\; \paren{ d^{(k)} }^2 - 2  \alpha \cdot \tau^{(k)}  \cdot d^{(k)} + \beta^2 \cdot  \paren{\tau^{(k)}}^2.
\end{align*}
Now we proceed to prove \eqref{eq:linear decay dist} by induction. It is clear that \eqref{eq:linear decay dist} holds for $\mb q^{(0)}$. Suppose $\mb q^{(k)}$ satisfies \eqref{eq:linear decay dist}, i.e., $d^{(k)}\leq \eta^k d^{(0)}$ for some $k\geq 1$. The quadratic term of $d^{(k)}$ on the right hand side of the inequality above will obtain its maximum at $ \frac{2}{5} \eta^k$ due to the definition of $\tau^{(0)}$ and $d^{(0)} \leq \frac{2}{5}$ as shown in \eqref{eqn:d-0-bound}. This, together with $\tau^{(k)} = \eta \tau^{(k-1)} $, it gives 
\begin{align*}
   \paren{ d^{(k+1)} }^2 \; &\leq\; \frac{4}{25} \eta^{2k}   - \frac{4}{5} \alpha \cdot \eta^{2k}  \tau^{(0)}  + \beta^2 \cdot \eta^{2k} \paren{\tau^{(0)}}^2 \\
   \;&= \; \frac{4}{25}\eta^{2k}  \cdot  \brac{ 1 -  5 \alpha \tau^{(0)} + \frac{25}{4} \beta^2 \paren{ \tau^{(0)} }^2 } \;\leq \; \eta^{2k+2}\cdot \paren{ d^{(0)} }^2
\end{align*}
where the last inequality follows from \eqref{eqn:subgradient-stepsize}, where
\begin{align*}
   1 -  5 \alpha \tau^{(0)} + \frac{25}{4} \beta^2 \paren{ \tau^{(0)} }^2 \;\leq\; 1 -   \alpha \tau^{(0)}  \;\leq \;1 - \frac{ \alpha^2 }{ 2\beta^2 } \;\leq\; \eta^2<1.
\end{align*}
This completes the proof.
\end{proof}

\begin{lemma}\label{lem:rho-bound}
Suppose $\theta \in \paren{ \frac{1}{n}, \frac{1}{3} }$. Whenever $p \;\geq\; C  \frac{ \kappa^8 }{\theta \sigma_{\min}^{2}(\mb C_{\mb a}) } \log^3 n$, we have 
\begin{align}
\rho\; :=\; \sup_{\mb q: \mb q^\top \mb r = 1} \frac{1}{np}\norm{ \mc P_{\mb r^\perp} \sum_{i=1}^p \paren{\mb R\mb Q^{-1}}^\top \mb C_{\mb x_i}^\top \sign\paren{ \mb C_{\mb x_i} \mb R \mb Q^{-1}  \mb q }}{} \;\leq\; 6 \sqrt{\log (np) }
\label{eq:LP subgradient bound}
\end{align}
holds with probability at least $ 1- p^{-c_1 n \theta } - n^{-c_2}$.
\end{lemma}

\begin{proof}
We have  
\begin{align*}
\rho\;\leq \; \frac{1}{np} \norm{ \mb R\mb Q^{-1} }{} \sum_{i=1}^p \paren{\norm{ \mb C_{\mb x_i}}{} \sup_{\mb q: \mb q^\top \mb r = 1}  \norm{ \sign\paren{ \mb C_{\mb x_i} \mb R \mb Q^{-1}  \mb q } }{} }.
\end{align*}
Since the $\sign(\cdot)$ function is bounded by $1$, we have
\begin{align*}
  	\rho \;\leq \; \frac{1}{ np} \norm{ \mb R\mb Q^{-1} }{} \cdot \paren{\sum_{i=1}^p \norm{ \mb F \mb x_i}{\infty}} \cdot \sqrt{n},
\end{align*}
where we used the fact that $\norm{\mb C_{\mb x_i}}{} = \norm{ \mb F\mb x_i }{\infty} $. As $\mb x_i \sim_{i.i.d.} \mc {BG}(\theta) $, let $\mb x_i = \mb b_i \odot \mb g_i$ with $\mb b_i \sim \mc B(\theta)$ and $\mb g_i \sim \mc N(\mb 0,\mb I)$. Then we have
\begin{align*}
   \norm{\mb C_{\mb x_i}}{} = \norm{ \mb F\mb x_i }{\infty} = \max_{1\leq j \leq n} \abs{ \paren{\mb f_j\odot \mb b_i}^* \mb g_i }.
\end{align*}
By Gaussian concentration inequality in Lemma \ref{lem:gauss-concentration} and a union bound, we have
\begin{align*}
   \bb P\paren{ \max_{1\leq i\leq p}\; \norm{ \mb F\mb x_i }{} \;\geq\; t } \;\leq\; (np)\cdot  \exp\paren{ - \frac{t^2}{ 2n } }.
\end{align*}
Choose $t = 4\sqrt{n \log\paren{np} }  $, then we have
\begin{align*}
   \max_{1\leq i\leq p}\; \norm{ \mb F\mb x_i }{}\;\leq \; 4\sqrt{n \log\paren{np} },
\end{align*}
with probability at least $1 -  (np)^{-7} $. On the other hand, by Lemma \ref{lem:precond-bounds}, we know that whenever $p \;\geq\; C  \frac{ \kappa^8 }{\theta \sigma_{\min}^{2}(\mb C_{\mb a}) } \log^3 n$, we have
\begin{align*}
   \norm{ \mb R\mb Q^{-1} }{}\;\leq\; \frac{3}{2},
\end{align*}
holds with probability at least $ 1- p^{-c_1 n \theta } - n^{-c_2}$. 
Combining all the results above, we obtain
\begin{align*}
   \rho \;\leq\; \frac{1}{np} \cdot \frac{3}{2} \cdot \paren{4p\sqrt{n \log\paren{np} } } \cdot \sqrt{n} \;=\; 6 \sqrt{ \log (np) },
\end{align*}
as desired.
\end{proof}

%% file: sec/app_geometry.tex
Here, we show that the reduced objective introduced in \eqref{eqn:problem-simple-app}
\begin{align*}
    \min_{\mb q} \wt{f}({\mb q}) = \frac{1}{np} \sum_{i=1}^p H_\mu \paren{  \mb C_{\mb x_i} {\mb q} },\quad \text{s.t.}\quad \norm{ {\mb q}}{} = 1.
\end{align*}
satisfies the regularity condition in population ($p \rightarrow +\infty$) on the set
\begin{align*}
   \mc S_\xi^{i\pm} \; := \; \Brac{ { \mb q } \in \bb S^{n-1} \; \mid\; \frac{\abs{{q}_i}}{ \norm{ {\mb q}_{-i} }{\infty } }\ge \sqrt{1 + \xi}, \; q_i \gtrless 0  }, 
\end{align*}
for every $i \in [n]$ and $\xi>0$.
\begin{proposition}\label{prop:regularity-population}
Whenever $\theta \in \paren{ \frac{1}{n}, c_0 }$ and $\mu \leq c_1\min\Brac{ \theta, \frac{1}{\sqrt{n}} } $, we have
\begin{align}
  	\innerprod{ \bb E\brac{\grad \wt{f}(\mb q) } }{q_i\mb q - \mb e_i } \; &\geq \;  c_2 \theta (1-\theta)  q_i \norm{\mb q_{-i} }{},\quad \sqrt{1-q_i^2} \in [\mu,c_3] \\
  	\innerprod{ \bb E\brac{\grad \wt{f}(\mb q) } }{q_i\mb q - \mb e_i } \; &\geq \;  c_2 \theta (1-\theta) q_i n^{-1}\norm{\mb q_{-i} }{},\quad \sqrt{1-q_i^2} \in \brac{c_3, \sqrt{ \frac{n-1}{n} } },
\end{align}
hold for any $\mb q \in \mc S_\xi^{i\pm }$ and each $i \in [n]$.
\end{proposition}

\paragraph{Remarks.} For proving this result, we first introduce some basic notations. We use $\mc I$ to denote the generic support set of $\mb q\in \bb S^{n-1}$ of i.i.d. $\mc B(\theta)$ law. Since the landscape is symmetric for each $i\in [n]$, without loss of generality, it is enough to consider the case when $i=n$. We reparameterize $\mb q\in \bb S^{n-1}$ by 
\begin{align}\label{eqn:repara-w}
   \mb q(\mb w): \; \mb w \mapsto \begin{bmatrix}
   \mb w \\ 
   \sqrt{ 1 - \norm{\mb w}{}^2 }
 \end{bmatrix},
\end{align}
where $\mb w \in \bb R^{n-1}$ with $ \norm{\mb w }{}\leq \sqrt{\frac{n-1}{n}}$. We write
\begin{align*}
   \mb q_{\mc I} = \begin{bmatrix}
 	\mb w_{\mc J} \\
 	q_n \indicator{ n \in \mc I }
 \end{bmatrix},
\end{align*}
where we use $\mc J$ to denote the support set of $\mb w$ of i.i.d. $\mc B(\theta)$ law.

\begin{proof}
 We denote
\begin{align}\label{eqn:function-g-w}
   	g(\mb w) = h_\mu\paren{ \mb w^\top \mb x_{-n} + x_n\sqrt{ 1 - \norm{\mb w}{}^2 } } 
\end{align}
Note that if $\mb e_n$ is a local minimizer of $\bb E\brac{\wt{f}(\mb q)}$, then $\bb E\brac{g(\mb w)}$ has a corresponding local minimum at $\mb 0$.  Since $g(\cdot)$ satisfies chain rule when computing its gradient, we have
\begin{align*}
	\innerprod{\bb E\brac{\nabla g(\mb w)} }{ \mb w - \mb 0}  &= \innerprod{\brac{\mb I_{n-1} \quad \frac{-\mb w}{\sqrt{1-\|\mb w\|^2}}} \nabla \bb E\brac{ \wt{f}(\mb q)}}{ \mb w} \\
 &= \innerprod{ \bb E\brac{ \nabla\wt{f}(\mb q)} }{\mb q -\frac{1}{q_n} \mb e_n} =
\frac{1}{q_n}\innerprod{ \bb E\brac{ \grad \wt{f}(\mb q)}}{q_n \mb q - \mb e_n}, 
\end{align*}
which gives
\begin{align}\label{eqn:grad-equiv}
\innerprod{ \bb E\brac{ \grad \wt{f}(\mb q)}}{q_n \mb q - \mb e_n} = q_n \innerprod{\bb E\brac{\nabla g(\mb w)} }{ \mb w }.
\end{align}
Thus, the above relationship implies that we can work on the ``unconstrained" function $g(\mb w)$ and establish the following: for any $\mb q(\mb w) \in \mc S_\xi^{n+}$ with $\xi> 0$, or equivalently,  
\begin{align*}
   \norm{\mb w}{}^2 + \paren{1+ \xi} \norm{\mb w}{ \infty }^2 \leq 1	,
\end{align*}
the following holds
\begin{align*}
	\innerprod{ \nabla \bb E\brac{g(\mb w) } }{ \mb w - \mb 0} \gtrsim \norm{\mb w}{}.
\end{align*} 
When $ \norm{\mb w}{} \in \brac{c_0\mu,c_1}$, Lemma \ref{lem:gradient-lower-bound} implies that
\begin{align*}
    \mb w^\top \nabla \bb E\brac{g(\mb w) } \geq c_2 \theta (1-\theta) \norm{\mb w}{}.
\end{align*}
By Lemma \ref{lem:geometry_asymp_curvature}, we know that when $c_1 \leq \norm{\mb w}{} \leq \sqrt{ \frac{n-1}{n} } $,
\begin{align*}
		\mb w^\top \nabla^2 \bb E\brac{g(\mb w)} \mb w \;\leq\; - c_3 \theta (1-\theta) \norm{\mb w}{}^2,
\end{align*}
which implies concavity of $g(\mb w)$ along the $\mb w$ direction. Let us denote $\mb v = \mb w/ \norm{\mb w}{}$, then the directional concavity implies that
\begin{align*}
    t\mb v^\top \nabla \bb E\brac{ g(t\mb v)} \;\geq \; (t'\mb v)^\top \nabla \bb E\brac{g(t' \mb v)} + c_4\theta (1-\theta) \paren{ t'   - t  },
\end{align*}
for any $t,t'\in \brac{ c_1, \sqrt{\frac{n-1}{n} } } $. Choose $t' = \frac{\norm{\mb w}{} }{ \sqrt{ \norm{\mb w}{}^2 + \norm{\mb w}{\infty }^2 } }  $ and $t = \norm{\mb w}{} $, by Lemma \ref{lem:gradient-positive}, we know that
\begin{align*}
   \mb w^\top \nabla \bb E\brac{g(\mb w)} \;\geq\; c_4 \theta(1-\theta)  \norm{\mb w}{} \paren{ \frac{1}{ \sqrt{ \norm{\mb w}{}^2 + \norm{\mb w}{ \infty }^2  } } -1 }.
\end{align*}
The function
\begin{align*}
   h_{\mb v}(t) \doteq \frac{ \norm{t\mb v}{} }{ \sqrt{ \norm{t\mb v}{}^2 + \norm{t\mb v}{ \infty }^2  }   }	- \norm{t\mb v}{} = \frac{1}{ \sqrt{1+ \norm{ \mb v  }{ \infty }^2 } } - t
\end{align*}
is obviously monotonically decreasing w.r.t. $t$. Since $\mb q \in \mc S_\xi^{n+}$, we have
\begin{align*}
    \norm{t\mb v}{}^2  + (1+\xi) \norm{t\mb v}{\infty}^2	\;\leq \; 1 \;\;\Longrightarrow\;\; t \leq \frac{1}{\sqrt{ 1+ (1+\xi) \norm{ \mb v }{ \infty }^2 }}.
\end{align*}
Therefore, we can uniformly lower bound $h_{\mb v}(t)$ by
\begin{align*}
   	h_{\mb v}(t) \; \geq \; \frac{1}{ \sqrt{1+ \norm{ \mb v  }{ \infty }^2 } } - \frac{1}{\sqrt{ 1+ (1+\xi) \norm{ \mb v }{ \infty }^2 }} \;\geq\; \xi \norm{ \mb v }{ \infty }^2 \geq \xi n^{-1}
\end{align*}
Therefore, we have
\begin{align*}
	\mb w^\top \nabla \bb E\brac{g(\mb w)} \;\geq\; c_4 \xi \theta(1-\theta) n^{-1}  \norm{\mb w}{},
\end{align*}
when $\norm{\mb w}{} \in \brac{ c_1, \sqrt{\frac{n-1}{n} } } $. Combining the bounds above, we obtain the desired results. 
\end{proof}

\begin{lemma}\label{lem:expectation-w-grad-w}
Suppose $\mb g \in \mc N(\mb 0, \mb I_n) $, we have
\begin{align}\label{eqn:w-grad-g}
   \mb w^\top \nabla \bb E\brac{g(\mb w)} = 	 \frac{1}{\mu} \bb E_{\mc I} \brac{ \paren{ \norm{\mb q_{\mc I}}{}^2 - \indicator{ n \in \mc I } } \bb P\paren{ \abs{ \mb q_{\mc I}^\top \mb g } \leq \mu  } }.
\end{align}

\end{lemma}

\begin{proof}
In particular, exchange of gradient and expectation operator can again be justified. By simple calculation, we obtain that
\begin{align}\label{eqn:gradient-w}
   \nabla g(\mb w) \;=\; \nabla h_\mu\paren{ \mb q^\top \mb x } \paren{ \mb x_{-n} -  \frac{x_n}{q_n} \mb w } 
   \;=\; \begin{cases}
   \frac{ \mb q^\top \mb x }{\mu} \paren{ \mb x_{-n} -  \frac{x_n}{q_n} \mb w }, &  \abs{ \mb q^\top \mb x } \leq \mu \\
 	\sign\paren{ \mb q^\top \mb x } \paren{ \mb x_{-n} -  \frac{x_n}{q_n} \mb w }, & \abs{ \mb q^\top \mb x } > \mu. 
 \end{cases}
\end{align}
Thus, we obtain
\begin{align*}
   	&  \mb w^\top \nabla \bb E\brac{g(\mb w)} \\
   	=\;& \bb E\brac{ \sign\paren{ \mb q^\top \mb x } \paren{ \mb w^\top\mb x_{-n} -  \frac{x_n}{q_n} \norm{ \mb w}{}^2 } \indicator{ \abs{\mb q^\top \mb x} \geq \mu  } } + \bb E\brac{ \frac{ \mb q^\top \mb x }{\mu} \paren{ \mb w^\top\mb x_{-n} -  \frac{x_n}{q_n} \norm{ \mb w}{}^2 } \indicator{ \abs{\mb q^\top \mb x} \leq \mu  } } \\
   	=\;& \bb E\brac{ \sign\paren{ \mb q^\top \mb x } \paren{ \mb q^\top\mb x -  \frac{x_n}{q_n}  } \indicator{ \abs{\mb q^\top \mb x} \geq \mu  } } + \frac{1}{\mu} \bb E\brac{ \paren{ \mb q^\top \mb x } \paren{ \mb q^\top\mb x -  \frac{x_n}{q_n}  } \indicator{ \abs{\mb q^\top \mb x} \leq \mu  } }, 
\end{align*}
where we used the fact that 
\begin{align*}
   \mb w^\top\mb x_{-n} -  \frac{x_n}{q_n} \norm{ \mb w}{}^2 = \mb w^\top \mb x_{-n} + q_n x_n - x_n \frac{ \norm{ \mb w}{}^2 + q_n^2 }{q_n} = \mb q^\top \mb x - \frac{x_n}{q_n}.
\end{align*}
Let $Z = X + Y$, with
\begin{align}\label{eqn:X-Y-Z}
	X = \mb w^\top \mb x_{-n} \sim \mc N(\mb 0, \norm{ \mb w_{\mc J} }{}^2 ),\; Y = q_n x_n \sim \mc N(0, q_n^2 \indicator{ n \in \mc I }  ), \; Z \sim \mc N(\mb 0, \norm{ \mb q_{\mc I} }{}^2).
\end{align}
This gives
\begin{align*}
 \mb w^\top \nabla \bb E\brac{g(\mb w)} \;&=\; \bb E\brac{ \abs{ \mb q^\top \mb x } \indicator{ \abs{ \mb q^\top \mb x } \geq \mu } } - \frac{1}{q_n} \bb E\brac{ \sign\paren{ \mb q^\top \mb x } x_n \indicator{ \abs{\mb q^\top 
 \mb x} \geq \mu } } \\
  &\; + \frac{1}{\mu} \bb E\brac{ \paren{ \mb q^\top \mb x }^2 \indicator{ \abs{\mb q^\top \mb x} \leq \mu }  }  - \frac{1}{q_n \mu} \bb E\brac{ x_n \paren{\mb w^\top \mb x_{-n}} \indicator{ \abs{\mb q^\top \mb x} \leq \mu } } - \frac{1}{\mu} \bb E\brac{ x_n^2 \indicator{ \abs{\mb q^\top \mb x} \leq \mu } } \\
  \;&=\; \bb E\brac{ \abs{Z} \indicator{ \abs{Z} \geq \mu } } - \frac{1}{q_n^2} \bb E\brac{ \sign\paren{X+Y} Y \indicator{ \abs{X+Y} \geq \mu } } + \frac{1}{\mu} \bb E\brac{ Z^2 \indicator{ \abs{Z}\leq \mu }  } \\
  &\; - \frac{1}{\mu q_n^2} \bb E\brac{ XY \indicator{ \abs{X+Y} \leq \mu  } } - \frac{1}{\mu q_n^2} \bb E\brac{ Y^2 \indicator{\abs{X+Y}\leq \mu } }.
\end{align*}
Now by Lemma~\ref{lemma:aux_asymp_proof_b}, we have 
\begin{align*}
   \bb E\brac{ \abs{Z} \indicator{ \abs{Z} \geq \mu  }  }	\;&=\; \sqrt{\frac{2}{\pi}} \bb E_{\mc I}\brac{ \norm{\mb q_{\mc I}}{} \exp\paren{ - \frac{\mu^2}{2 \norm{\mb q_{\mc I}}{}^2  } }  } \\
       \bb E\brac{ \sign\paren{X+Y} Y \indicator{ \abs{X+Y} \geq \mu } }\;&=\; q_n^2 \sqrt{ \frac{2}{\pi} }\bb E\brac{ \frac{ \indicator{ n \in \mc I } }{ \norm{ \mb q_{\mc I} }{} } \exp\paren{ - \frac{\mu^2}{ 2 \norm{\mb q_{\mc I}}{}^2 } }   }  \\
    \bb E\brac{ Z^2 \indicator{ \abs{Z} \leq \mu } } \;&=\; - \mu \sqrt{ \frac{2}{\pi} } \bb E_{\mc I} \brac{ \norm{ \mb q_{\mc I} }{} \exp\paren{ - \frac{\mu^2}{2 \norm{\mb q_{\mc I}}{}^2} }  } + \bb E_{\mc I} \brac{ \norm{\mb q_{\mc I}}{}^2 \bb P\paren{ \abs{ \mb q_{\mc I}^\top \mb g } \leq \mu  }  } \\
    \bb E\brac{ XY \indicator{ \abs{X+Y} \leq \mu  } } \;&=\; - \mu q_n^2 \sqrt{ \frac{2}{\pi} } \bb E_{\mc I} \brac{ \frac{ \indicator{n \in \mc I } \norm{\mb w_{\mc J}}{}^2  }{ \norm{ \mb q_{\mc I} }{}^3 }  \exp\paren{ - \frac{\mu^2}{ 2 \norm{\mb q_{\mc I}}{}^2 }  } } \\
    \bb E\brac{ Y^2  \indicator{\abs{X+Y}\leq \mu } } \;&=\; - \mu q_n^4 \sqrt{ \frac{2}{\pi} } \bb E_{\mc I}\brac{ \frac{ \indicator{n \in \mc I } }{ \norm{ \mb q_{\mc I} }{}^3  } \exp\paren{ - \frac{\mu^2}{ 2 \norm{ \mb q_{\mc I} }{}^2 } } } + q_n^2 \bb E_{\mc I} \brac{ \indicator{n \in \mc I }  \bb P\paren{ \abs{ \mb q_{\mc I}^\top \mb g } \leq \mu  } } 
\end{align*}

Putting the above calculations together and simplify, we obtain the desired result in \eqref{eqn:w-grad-g}.

\end{proof}

\begin{lemma}\label{lem:gradient-positive}
   When for any $\mb w\in \bb R^{n-1} $ satisfies $\norm{\mb w}{}^2+ \norm{\mb w}{\infty}^2 \leq 1$, we have
\begin{align*}
   \mb w^\top \nabla \bb E\brac{g(\mb w)} \;\geq\; 0.
\end{align*}
\end{lemma}

\begin{proof}
From Lemma \ref{lem:expectation-w-grad-w}, we know that
\begin{align}
&\mu \cdot \mb w^\top \nabla \bb E\brac{g(\mb w)} \nonumber \\
   =\;& \bb E_{\mc I} \brac{ \paren{ \norm{\mb q_{\mc I}}{}^2 - \indicator{ n \in \mc I } } \bb P\paren{ \abs{ \mb q_{\mc I}^\top \mb g } \leq \mu  } } \nonumber \\
  =\;&  \bb E_{\mc J}\brac{  (1-\theta) \norm{ \mb w_{\mc J} }{}^2 \bb P\paren{ \abs{ \mb g_{-n}^\top \mb w_{\mc J}} \leq \mu } - \theta \norm{ \mb w_{\mc J^c} }{}^2   \bb P\paren{ \abs{ \mb g_{-n}^\top \mb w_{\mc J} + q_ng_n} \leq \mu }  } \nonumber \\
  =\;& \bb E_{\mc J} \brac{ \int_{-\mu}^\mu \paren{ \frac{1-\theta}{\sqrt{2\pi}} \frac{\norm{ \mb w_{\mc J} }{}^2}{\norm{ \mb w_{\mc J} }{} }  \exp\paren{ - \frac{t^2}{2\norm{ \mb w_{\mc J} }{}^2} } - \frac{\theta}{\sqrt{2\pi}} \frac{ \norm{ \mb w_{\mc J^c} }{}^2 }{ \sqrt{ 1 - \norm{\mb w_{\mc J^c}  }{}^2 } } \exp\paren{\frac{-t^2}{2 - 2\norm{\mb w_{\mc J^c} }{}^2 }} } dt } \nonumber \\
  =\;& \frac{1-\theta}{\sqrt{2\pi}}  \sum_{i=1}^{n-1} \int_{-\mu}^\mu \bb E_{\mc J} \brac{ \frac{ w_i^2 \indicator{i \in \mc J } }{ \sqrt{ w_i^2\indicator{i \in \mc J }  +  \norm{ \mb w_{\mc J\setminus\{i\} } }{}^2  } }  \exp\paren{ - \frac{t^2}{ 2w_i^2\indicator{i \in \mc J }  +  2\norm{ \mb w_{\mc J \setminus \{i\} } }{}^2} } }  dt \nonumber \\
  &\; - \frac{\theta}{\sqrt{2\pi}} \sum_{i=1}^{n-1} \int_{-\mu}^\mu \bb E_{\mc J} \brac{ \frac{ w_i^2 \indicator{ i \not \in \mc J } }{ \sqrt{ 1 - w_i^2\indicator{ i \not \in \mc J } - \norm{ \mb w_{\mc J^c \setminus \{ i \} } }{}^2  } }  \exp\paren{ - \frac{t^2}{  2 - 2w_i^2\indicator{ i \not \in \mc J } - 2\norm{ \mb w_{\mc J^c \setminus\{ i \} } }{}^2 } }   }  dt \nonumber \\
  =\;& \frac{(1-\theta)\theta}{\sqrt{2\pi}}  \sum_{i=1}^{n-1} \int_{-\mu}^\mu \bb E_{\mc J} \brac{ \frac{ w_i^2  }{ \sqrt{ w_i^2  +  \norm{ \mb w_{\mc J\setminus\{i\} } }{}^2  } }  \exp\paren{ - \frac{t^2}{ 2w_i^2  +  2\norm{ \mb w_{\mc J\setminus\{i\} } }{}^2} } }  dt \nonumber \\
  &\; - \frac{(1-\theta)\theta}{\sqrt{2\pi}} \sum_{i=1}^{n-1} \int_{-\mu}^\mu \bb E_{\mc J} \brac{ \frac{ w_i^2  }{ \sqrt{ 1 - \norm{\mb w }{}^2  + \norm{ \mb w_{\mc J \setminus\{ i \} } }{}^2  } }  \exp\paren{ - \frac{t^2}{ 2 - 2\norm{\mb w }{}^2  + 2\norm{ \mb w_{\mc J \setminus\{ i \} } }{}^2  } }   }  dt \nonumber \\
  =\;& (1-\theta)\theta \sum_{i=1}^{n-1} w_i^2 \bb E_{\mc J}\brac{ \bb P\paren{ \abs{Z_{i1}} \leq \mu  } - \bb P\paren{ \abs{Z_{i2}} \leq \mu  } }, \label{eqn:grad-prob}
\end{align}
where
\begin{align}\label{eqn:Zi}
    \quad Z_{i1} \sim \mc N\paren{0, w_i^2 + \norm{ \mb w_{\mc J \setminus \{i\} } }{}^2  },\quad Z_{i2} \sim \mc N\paren{ 0, 1- \norm{\mb w}{}^2+  \norm{\mb w_{\mc J \setminus \{i\}  } }{}^2  }.
\end{align}
Since we have $1 - \norm{ \mb w}{}^2 \geq \norm{ \mb w}{\infty}^2 \ge w_i^2$, the variance of $Z_i^2$ is larger than that of $Z_i^1$. Therefore, we have $\bb P\paren{ \abs{Z_{i1}} \leq \mu  } \geq \bb P\paren{ \abs{Z_{i2}} \leq \mu  }$ for each $i = 1,\cdots, n-1$. Hence, we obtain
\begin{align*}
	 \mb w^\top \nabla \bb E\brac{g(\mb w)} \;=\; \frac{1}{\mu}\theta (1-\theta) \sum_{i=1}^{n-1} w_i^2 \bb E_{\mc J}\brac{ \bb P\paren{ \abs{Z_{i1}} \leq \mu  } - \bb P\paren{ \abs{Z_{i2}} \leq \mu  } } \geq 0.
\end{align*}	
\end{proof}

\begin{lemma}\label{lem:gradient-lower-bound}
	For any $\mb w$ with $ c_0\mu \leq \norm{\mb w}{} \leq c_1$, we have
	\begin{align*}
	   \mb w^\top \nabla \bb E\brac{g(\mb w)} \geq c \theta (1-\theta) \norm{\mb w}{}
	\end{align*}

\end{lemma}

\begin{proof} Recall from \eqref{eqn:grad-prob}, we have
\begin{align*}
\mb w^\top \nabla \bb E\brac{g(\mb w)}
  \;=\; \frac{1}{\mu}(1-\theta)\theta \sum_{i=1}^{n-1} w_i^2 \bb E_{\mc J}\brac{ \bb P\paren{ \abs{Z_{i1}} \leq \mu  } - \bb P\paren{ \abs{Z_{i2}} \leq \mu  } },
\end{align*}
where $Z_{i1}$ and $Z_{i2}$ are defined the same as \eqref{eqn:Zi}. Let us denote
\begin{align*}
    \quad Z_{1} \sim \mc N\paren{0, \norm{\mb w}{}^2  },\quad Z_{2} \sim \mc N\paren{ 0, 1- \norm{\mb w}{}^2 }.
\end{align*}
Since we have $\norm{\mb w}{}^2 \ge w_i^2 + \norm{ \mb w_{\mc J \setminus \{i\} } }{}^2$, the variance of $Z_1$ is larger than that of $Z_{i1}$.
Therefore, we have $\bb P\paren{ \abs{Z_{i1}} \leq \mu  } \geq \bb P\paren{ \abs{Z_{1}} \leq \mu  }$ for each $i = 1,\cdots, n-1$. By a similar argument, we have $\bb P\paren{ \abs{Z_{i2}} \leq \mu  } \leq \bb P\paren{ \abs{Z_{2}} \leq \mu  }$ for each $i = 1,\cdots, n-1$. Thus, we obtain
\begin{align}
&\bb P\paren{ \abs{Z_{i1}} \leq \mu  } - \bb P\paren{ \abs{Z_{i2}} \leq \mu  } \nonumber \\
  \ge \;& \bb P\paren{ \abs{Z_{1}} \leq \mu  } - \bb P\paren{ \abs{Z_{2}} \leq \mu  } \nonumber \\
  = \;&  \sqrt{\frac{2}{\pi} } \frac{1}{ \norm{\mb w}{}}  \int_{0}^\mu \exp\paren{ - \frac{t^2}{2\norm{ \mb w }{}^2} }dt  - \sqrt{\frac{2}{\pi}}
   \frac{1}{ \sqrt{1-\norm{\mb w}{}^2}} \int_{0}^\mu  \exp\paren{ - \frac{t^2}{2-2\norm{\mb w}{}^2} } dt  \nonumber \\
  \geq \;& \sqrt{\frac{2}{\pi}}\brac{ \frac{1}{\norm{\mb w}{}} \int_{0}^\mu \paren{1 - \frac{t^2}{2\norm{ \mb w }{}^2} } dt   - \frac{\mu}{\sqrt{1-\norm{\mb w}{}^2}}} \nonumber \\
  = \; & \sqrt{\frac{2}{\pi}} \brac{\frac{1}{\norm{\mb w}{}}\paren{\mu - \frac{1}{6}\frac{\mu^3}{\norm{ \mb w }{}^2} }  - \frac{\mu}{\sqrt{1-\norm{\mb w}{}^2}} } \nonumber \\
  \geq \; &\mu \sqrt{\frac{2}{\pi}} \paren{\frac{1}{\norm{\mb w}{}} - 2\frac{1}{\sqrt{1-\norm{\mb w}{}^2}}     } \; \geq \; \frac{\mu}{2\sqrt{2\pi}} \frac{1}{\norm{\mb w}{}} \label{eqn:gradient-prob-bound}
\end{align}
where we used the fact that $ \mu /\sqrt{3} \leq \norm{\mb w}{}  \leq 1/\sqrt{17}$ for the last two inequalities.  Plugging \eqref{eqn:gradient-prob-bound} back into \eqref{eqn:grad-prob} gives
\begin{align*}
\mb w^\top \nabla \bb E\brac{g(\mb w)} \;&=\; \frac{1}{\mu} (1-\theta)\theta \sum_{i=1}^{n-1} w_i^2 \bb E_{\mc J}\brac{ \bb P\paren{ \abs{Z_{i1}} \leq \mu  } - \bb P\paren{ \abs{Z_{i2}} \leq \mu  } }\\
  \;&\geq \; \frac{(1-\theta)\theta}{2\sqrt{2\pi}\norm{\mb w}{}}   \sum_{i=1}^{n-1} w_i^2 \;=\; \frac{1}{2\sqrt{2\pi}}(1-\theta)\theta \norm{\mb w}{},
\end{align*}
as desired.
\end{proof}

\begin{lemma} \label{lem:geometry_asymp_curvature}
When $\mu \leq c_0 \min\Brac{ \frac{1}{\sqrt{n}}, \theta } $ and $\theta \in \paren{ \frac{1}{n}, c_1}$, we have
\begin{align*}
	\mb w^\top \nabla^2 \bb E\brac{g(\mb w)} \mb w \;\leq\; - c_2 \theta (1-\theta) \norm{\mb w}{}^2
\end{align*}
for all $\mb w$ with $ c_3 \leq \norm{\mb w}{}\leq \sqrt{ \frac{n-1}{n} } $. Here, $c_0,\;c_1,\;c_2,$ and $c_3$ are some numerical constants.
\end{lemma}
\begin{proof}
Since the expectation and derivative are exchangeable, we have
 \begin{align*}
    \mb w^\top \nabla^2 \bb E\brac{g(\mb w)} \mb w = \mb w^\top \bb E\brac{\nabla^2  g(\mb w)	 } \mb w.
 \end{align*}
From \eqref{eqn:gradient-w}, we obtain
\begin{align*}
	\mb w^\top \nabla^2 g (\mb w) \mb w = \begin{cases}
	\frac{1}{\mu} \brac{ \paren{\mb q^\top \mb x}^2 - \frac{x_n}{q_n} \paren{ \mb q^\top \mb x } - \frac{x_n}{q_n^3} \paren{ \mb x_{-n}^\top \mb w }  }, & \abs{\mb q^\top \mb x} \leq \mu \\
 	- \frac{x_n}{q_n^3} \norm{\mb w}{}^2 \sign\paren{ \mb q^\top \mb x }, & \abs{\mb q^\top \mb x}\geq \mu.
 \end{cases}
\end{align*}
Thus, we have
\begin{align*}
   \bb E\brac{  \mb w^\top \nabla^2 g (\mb w) \mb w \indicator{ \abs{\mb q^\top \mb x}\geq \mu } }	\;&=\; - \frac{ \norm{\mb w}{}^2  }{q_n^4} \bb E\brac{ q_n x_n \sign\paren{ \mb q^\top \mb x }  \indicator{ \abs{\mb q^\top \mb x}\geq \mu } } \\
   \;&=\;  - \sqrt{ \frac{2}{\pi} } \frac{ \norm{\mb w}{}^2 }{ q_n^2 }  \bb E_{\mc I}\brac{  \frac{\indicator{n \in \mc I} }{ \norm{\mb q_{\mc I} }{} } \exp\paren{ - \frac{\mu^2}{2\norm{\mb q_{\mc I}}{}^2} } }
\end{align*}
and 
\begin{align*}
   &\bb E\brac{ \mb w^\top  \nabla^2 g (\mb w) \mb w \indicator{ \abs{ \mb q^\top \mb x } \leq \mu }  }	\\
   =\;& \frac{1}{\mu} \bb E\brac{ \paren{\mb q^\top \mb x}^2 \indicator{ \abs{\mb q^\top \mb x} \leq \mu   } } - \frac{1}{\mu}\bb E\brac{ \frac{x_n}{q_n} \paren{ \mb q^\top \mb x } \indicator{ \abs{\mb q^\top \mb x} \leq \mu   } } - \frac{1}{\mu} \bb E\brac{ \frac{x_n}{q_n^3} \paren{ \mb x_{-n}^\top \mb w } \indicator{ \abs{\mb q^\top \mb x} \leq \mu   } } \\
   =\;& \frac{1}{\mu} \bb E\brac{Z^2 \indicator{ \abs{Z} \leq \mu }  }  - \frac{1}{\mu q_n^2} \bb E\brac{ Y^2 \indicator{\abs{X+Y} \leq \mu } }  - \frac{1}{\mu} \paren{\frac{1}{  q_n^2 } + \frac{1}{  q_n^4 }} \bb E\brac{ XY \indicator{ \abs{X+Y} \leq \mu } },
\end{align*}
where $X$, $Y$ and $Z=X+Y$ are defined the same as \eqref{eqn:X-Y-Z}. Similar to Lemma \ref{lem:expectation-w-grad-w}, by using Lemma~\ref{lemma:aux_asymp_proof_b}, we obtain
\begin{align*}
   &\bb E\brac{ \mb w^\top \nabla^2 g (\mb w) \mb w \indicator{ \abs{ \mb q^\top \mb x } \leq \mu }  } \\
   =\;& - \sqrt{ \frac{2}{\pi} }  \bb E_{\mc I}\brac{ \norm{\mb q_{\mc I} }{} \exp\paren{ - \frac{\mu^2}{ 2 \norm{\mb q_{\mc I}}{}^2  } } } +  \frac{1}{\mu} \bb E\brac{\paren{ \norm{\mb q_{\mc I} }{}^2 - \indicator{n \in \mc I} } \bb P\paren{ \abs{ \mb q_{\mc I}^\top \mb g } \leq \mu }} \\
   & + \sqrt{ \frac{2}{\pi} } \bb E_{\mc I} \brac{   \frac{q_n^2 \indicator{ n \in \mc I }  }{ \norm{ \mb q_{\mc I} }{}^3 } \exp\paren{ - \frac{\mu^2}{ 2 \norm{\mb q_{\mc I}}{}^2 }  }   } +  \sqrt{\frac{2}{\pi}}  \paren{1 + \frac{1}{  q_n^2 }} \bb E_{\mc I} \brac{ \frac{ \norm{\mb w_{\mc J}}{}^2 \indicator{ n \in \mc I }  }{ \norm{\mb q_{\mc I}}{}^3  } \exp\paren{ - \frac{ \mu^2 }{ 2 \norm{\mb q_{\mc I}}{}^2 } } }.
\end{align*}
Combining the results above and using integral by parts, we obtain
\begin{align*}
&\mb w^\top  \nabla^2 \bb E\brac{g(\mb w) } \mb w\\
    =\;& - \sqrt{\frac{2}{\pi}} \bb E_{\mc I} \brac{  \frac{ \indicator{n \in \mc I} }{ \norm{\mb q_{\mc I}}{}^3 }  \exp\paren{ - \frac{ \mu^2 }{ 2 \norm{\mb q_{\mc I}}{}^2 }} } + 2 \sqrt{\frac{2}{\pi}} \bb E_{\mc I} \brac{  \frac{ \indicator{n \in \mc I} }{ \norm{\mb q_{\mc I}}{} }  \exp\paren{ - \frac{ \mu^2 }{ 2 \norm{\mb q_{\mc I}}{}^2 }} } \\
   & -  \sqrt{ \frac{2}{\pi} }  \bb E_{\mc I}\brac{ \norm{\mb q_{\mc I} }{} \exp\paren{ - \frac{\mu^2}{ 2 \norm{\mb q_{\mc I}}{}^2  } } } +  \frac{1}{\mu} \bb E\brac{\paren{ \norm{\mb q_{\mc I} }{}^2 - \indicator{n \in \mc I} } \bb P\paren{ \abs{ \mb q_{\mc I}^\top \mb g } \leq \mu }} \\
   =\;& - \sqrt{\frac{2}{\pi} }  \bb E_{\mc I} \brac{  \frac{ \norm{ \mb w_{\mc J^c} }{}^2 \indicator{n \in \mc I} }{ \norm{\mb q_{\mc I}}{}^3 }  \exp\paren{ - \frac{ \mu^2 }{ 2 \norm{\mb q_{\mc I}}{}^2 }} } \\
   &\;+  \sqrt{\frac{2}{\pi}} \bb E_{\mc I} \brac{  \frac{ \indicator{n \in \mc I} }{ \norm{\mb q_{\mc I}}{} } \paren{\exp\paren{ - \frac{\mu^2}{ 2 \norm{\mb q_{\mc I}}{}^2  } } - \frac{ \norm{\mb q_{\mc I} }{} }{\mu} \int_0^{\mu / \norm{ \mb q_{\mc I} }{}  } \exp\paren{-t^2/2 }dt }} \\
   &\; - \sqrt{ \frac{2}{\pi} }  \bb E_{\mc I}\brac{   \norm{\mb q_{\mc I} }{} \paren{ \exp\paren{ - \frac{\mu^2}{ 2 \norm{\mb q_{\mc I}}{}^2  } } - \frac{\norm{\mb q_{\mc I} }{}}{\mu} \int_0^{\mu/ \norm{ \mb q_{\mc I} }{} } \exp\paren{-t^2/2}  dt }} \\
   =\;&  - \sqrt{\frac{2}{\pi} }   \bb E_{\mc I} \brac{ \norm{ \mb w_{\mc J^c} }{}^2 \frac{ \indicator{n \in \mc I}   }{ \norm{\mb q_{\mc I}}{}^3 }  \exp\paren{ - \frac{ \mu^2 }{ 2 \norm{\mb q_{\mc I}}{}^2 }} } - \frac{1}{\mu} \sqrt{\frac{2}{\pi}} \bb E_{\mc I} \brac{  \indicator{n \in \mc I} \int_0^{\mu / \norm{ \mb q_{\mc I} }{}  } t^2 \exp\paren{-t^2/2}  dt } \\
   & + \frac{1}{\mu} \sqrt{ \frac{2}{\pi} } \bb E_{\mc I} \brac{ \norm{\mb q_{\mc I}}{}^2 \int_0^{\mu/ \norm{ \mb q_{\mc I} }{} } t^2\exp\paren{-t^2/2}  dt } \\
    \leq\;& - \sqrt{\frac{2}{\pi} }   \bb E_{\mc I} \brac{ \norm{ \mb w_{\mc J^c} }{}^2 \frac{ \indicator{n \in \mc I}   }{ \norm{\mb q_{\mc I}}{}^3 }  \exp\paren{ - \frac{ \mu^2 }{ 2 \norm{\mb q_{\mc I}}{}^2 }} } + \frac{1}{\mu} \sqrt{ \frac{2}{\pi} }  \int_0^\mu t^2 \bb E_{\mc I} \brac{ \frac{1}{\norm{ \mb q_{\mc I} }{}}  \exp\paren{ - \frac{t^2}{ 2\norm{\mb q_{\mc I}}{}^2 }  }  }  dt.
\end{align*}




First,  when $\sqrt{ \frac{n-1}{n} }  \geq \norm{\mb w}{} \geq c_0 $, we have
\begin{align*}
   	 &\bb E_{\mc I} \brac{ \norm{ \mb w_{\mc J^c} }{}^2 \frac{ \indicator{n \in \mc I}   }{ \norm{\mb q_{\mc I}}{}^3 }  \exp\paren{ - \frac{ \mu^2 }{ 2 \norm{\mb q_{\mc I}}{}^2 }} } \\
   	 =\;& \theta  \bb E_{\mc J} \brac{ \norm{ \mb w_{\mc J^c} }{}^2 \frac{ 1   }{ \paren{ q_n^2 + \norm{\mb w_{\mc J} }{}^2  }^{3/2} }  \exp\paren{ - \frac{ \mu^2 }{ 2 \paren{q_n^2 + \norm{\mb w_{\mc J}}{}^2  } }} } \\
   	 \geq\; & \theta \bb E_{\mc J} \brac{ \norm{ \mb w_{\mc J^c} }{}^2  \exp\paren{ - \frac{\mu^2}{2q_n^2 + 2 \norm{\mb w_{\mc J}}{}^2 } }  }  \\
   	 \geq\; &  \theta \bb E_{\mc J} \brac{ \norm{ \mb w_{\mc J^c} }{}^2  \exp\paren{ - \frac{\mu^2}{2q_n^2} }  } \;\geq\; c_1\theta (1 - \theta) \norm{ \mb w }{}^2 .
\end{align*}
Second, notice that the function
\begin{align*}
   h(x) = x^{-1} \exp\paren{ - \frac{t^2}{2x^2} },\quad x\in [0,1]
\end{align*}
reaches the maximum when $x = t $. Thus, we have
\begin{align*}
   	\frac{1}{\mu} \sqrt{ \frac{2}{\pi} } \int_0^\mu t^2 \bb E_{\mc I} \brac{ \frac{1}{\norm{ \mb q_{\mc I} }{}}  \exp\paren{ - \frac{t^2}{ 2\norm{\mb q_{\mc I}}{}^2 }  }  }  dt \leq \frac{1}{\mu} \sqrt{ \frac{2}{\pi} } \int_0^\mu t \exp\paren{ -\frac{1}{2} } dt \leq  \frac{1}{\sqrt{2\pi} }  e^{-1/2} \mu.
\end{align*}
Therefore, when $\mu \leq \frac{1}{n}\leq \theta $, we have
\begin{align*}
   \mb w^\top \nabla^2 \bb E\brac{ g(\mb w) } \mb w \leq -c_2\theta (1-\theta) \norm{\mb w}{}^2 
\end{align*}
for any $\sqrt{ \frac{n-1}{n} }  \geq \norm{\mb w}{} \geq c_0 $.



\end{proof}

%% file: sec/app_stay.tex
Under the same settings of Appendix \ref{app:regularity-population}, we show that the simplified function $\wt{f}(\mb q)$ satisfies the following implicit regularization property over $\mb q \in \mc S_\xi^{i\pm }$ for each $i \in [n]$.

\begin{proposition}\label{prop:orth_stable_manifold_population}
Suppose $\theta \geq \frac{1}{n}$. Given any index $i \in [n]$, when $\mu \leq \frac{1}{\sqrt{3n}}$, we have
	\begin{align*}
	   \innerprod{ \grad \bb E\brac{\wt{f}(\mb q)}}{ \frac{1}{q_j} \mb e_j - \frac{1}{q_i} \mb e_i }	\;\geq \; \frac{\theta(1-\theta)}{4n} \frac{\xi}{1+\xi},
	\end{align*}
	holds for all $\mb q \in \mc S_\xi^{i\pm }$ and any $q_j$ such that $j \not = i$ and $q_j^2\geq \frac{1}{3}q_i^2 $
\end{proposition}

\begin{proof}
Without loss of generality, let us consider the case $i=n$. For any $j \not = n$, we have
\begin{align*}
   	&\innerprod{ \grad \bb E\brac{\wt{f}(\mb q)}}{ \frac{1}{q_j} \mb e_j - \frac{1}{q_n} \mb e_n } \\
   	=\;& \paren{ \frac{1}{q_j}\mb e_j - \frac{1}{q_n} \mb e_n }^\top \mc P_{\mb q^\perp} \bb E\brac{ \mb x \cdot \nabla h_\mu(\mb x^\top \mb q) } \\
   	=\;& \paren{ \frac{1}{q_j}\mb e_j - \frac{1}{q_n} \mb e_n }^\top \bb E\brac{ \mb x \cdot \nabla h_\mu(\mb x^\top \mb q) }.
\end{align*}
Let 
\begin{align*}
   Z = Z_1 + Z_2, \quad Z_1 \;=\; q_i x_i \sim \mc N(0, (b_iq_i)^2  ) ,\quad Z_2 \;=\; \mb q_{-i}^\top \mb x_{-i}   	\sim \mc N(0, \norm{ \mb q_{-i} \odot \mb b_{-i} }{}^2 ).
\end{align*}
Notice that for every $i \in [n]$, we have
\begin{align*}
  &\frac{1}{q_i} \mb e_i^\top \bb E\brac{ \mb x \cdot \nabla h_\mu(\mb x^\top \mb q) } \\
  =\;& \frac{1}{q_i^2} \frac{1}{\mu} \bb E\brac{ Z_1^2 \indicator{\abs{Z_1 +Z_2}\leq \mu } }+ \frac{1}{q_i^2} \frac{1}{\mu} \bb E\brac{ Z_1 Z_2 \indicator{\abs{Z_1 +Z_2}\leq \mu } } +  \frac{1}{q_i^2} \bb E\brac{ Z_1 \sign\paren{Z_1+Z_2} \indicator{\abs{Z_1 +Z_2}\geq \mu } }. 
\end{align*}
By Lemma \ref{lemma:aux_asymp_proof_b}, we have
\begin{align*}
   	\bb E\brac{ Z_1^2 \indicator{\abs{Z_1 +Z_2}\leq \mu } } \;&=\; - \sqrt{ \frac{2}{\pi} }\mu \bb E_{\mc I}\brac{\frac{ q_i^4 \indicator{ i \in \mc I } }{ \norm{ \mb q_{\mc I} }{}^3 } \exp\paren{ - \frac{ \mu^2 }{ 2 \norm{ \mb q_{\mc I} }{}^2 } } } \\
   	&\quad+ \bb E\brac{ q_i^2 \indicator{ i \in \mc I } \bb P\paren{ \abs{Z} \leq \mu }  } , \\
   	\bb E\brac{ Z_1 Z_2 \indicator{\abs{Z_1 +Z_2}\leq \mu } } \;&=\; - \sqrt{ \frac{2}{\pi} }\mu \bb E_{\mc I}\brac{\frac{ q_i^2 \indicator{ i \in \mc I } \norm{ (\mb q_{-i})_{\mc J} }{}^2  }{ \norm{ \mb q_{\mc I} }{}^3 } \exp\paren{ - \frac{ \mu^2 }{ 2 \norm{ \mb q_{\mc I} }{}^2 } } } \\
   	\bb E\brac{ Z_1 \sign\paren{Z_1+Z_2} \indicator{\abs{Z_1 +Z_2}\geq \mu } } \;&=\; \sqrt{\frac{2}{\pi}} \bb E_{\mc I}\brac{ \frac{ q_i^2 \indicator{ i \in \mc I } }{ \norm{ \mb q_{\mc I} }{}  } \exp\paren{ - \frac{ \mu^2 }{ 2 \norm{ \mb q_{\mc I} }{}^2  } } }.
\end{align*}
Combining the results above, we obtain
\begin{align*}
   \frac{1}{q_i} \mb e_i^\top \bb E\brac{ \mb x \cdot \nabla h_\mu(\mb x^\top \mb q) } \;=\; \frac{1}{\mu} \bb E\brac{  \indicator{ i \in \mc I } \bb P\paren{ \abs{Z} \leq \mu }  }.\end{align*}
Therefore, we have
\begin{align*}
   &\innerprod{ \grad \bb E\brac{\wt{f}(\mb q)}}{ \frac{1}{q_j} \mb e_j - \frac{1}{q_n} \mb e_n } \\	
   =\;& \frac{1}{\mu} \paren{ \bb E\brac{  \indicator{ j \in \mc I } \bb P\paren{ \abs{Z} \leq \mu }  } - \bb E\brac{  \indicator{ n \in \mc I } \bb P\paren{ \abs{Z} \leq \mu }  } } \\
   =\;& \frac{\theta}{\mu} \sqrt{\frac{2}{\pi}} \bb E_{\mc I} \brac{ \frac{1}{ \sqrt{ q_j^2 + \norm{ \mb q_{\mc I \setminus j } }{}^2 }  } \int_0^\mu \exp\paren{ - \frac{t^2}{  q_j^2 + \norm{ \mb q_{\mc I \setminus j } }{}^2  }  } dt  } \\
   &\quad - \frac{\theta}{\mu} \sqrt{\frac{2}{\pi}} \bb E_{\mc I} \brac{ \frac{1}{ \sqrt{ q_n^2 + \norm{ \mb q_{\mc I \setminus n } }{}^2 }  } \int_0^\mu \exp\paren{ - \frac{t^2}{  q_n^2 + \norm{ \mb q_{\mc I \setminus n } }{}^2  }  } dt  } \\
   =\;& \frac{\theta(1-\theta)}{\mu } \sqrt{\frac{2}{\pi}} \bb E_{\mc I} \brac{ \frac{1}{ \sqrt{ q_j^2 + \norm{ \mb q_{\mc I \setminus \{j,n\} } }{}^2 }  } \int_0^\mu \exp\paren{ - \frac{t^2}{  q_j^2 + \norm{ \mb q_{\mc I \setminus \{j,n\} } }{}^2  }  } dt  } \\
   &\quad - \frac{\theta(1-\theta)}{\mu } \sqrt{\frac{2}{\pi}} \bb E_{\mc I} \brac{ \frac{1}{ \sqrt{ q_n^2 + \norm{ \mb q_{\mc I \setminus \{j,n\} } }{}^2 }  } \int_0^\mu \exp\paren{ - \frac{t^2}{  q_n^2 + \norm{ \mb q_{\mc I \setminus \{j,n\} } }{}^2  }  } dt  } \\
   =\; & \frac{\theta(1-\theta)}{\mu } \bb E_{\mc I}\brac{ \mathrm{erf}\paren{ \frac{\mu}{\sqrt{ q_i^2 + \norm{ \mb q_{\mc I \setminus \{j,n\} } }{}^2 } }  } - \mathrm{erf}\paren{ \frac{\mu}{\sqrt{ q_n^2 + \norm{ \mb q_{\mc I \setminus \{j,n\} } }{}^2 } }  } }
\end{align*}
where $\mathrm{erf}(x)$ is the Gaussian error function
\begin{align*}
   \mathrm{erf}(x) =  \frac{1}{\sqrt{2\pi}} \int_{-x}^x \exp\paren{-t^2/2} dt = \sqrt{\frac{2}{2\pi}} \int_{0}^x \exp\paren{-t^2/2} dt,\quad x\geq 0.
\end{align*}
When $\mu \leq \frac{1}{\sqrt{3n}}$ such that $\frac{\mu}{\sqrt{ q_n^2 + \norm{ \mb q_{\mc I \setminus \{j,n\} } }{}^2 } } \leq 1$ for $\mb q \in \mc S_\xi^{n+}$, by Taylor approximation we have
\begin{align*}
    &\mathrm{erf}\paren{ \frac{\mu}{\sqrt{ q_i^2 + \norm{ \mb q_{\mc I \setminus \{j,n\} } }{}^2 } }  } - \mathrm{erf}\paren{ \frac{\mu}{\sqrt{ q_n^2 + \norm{ \mb q_{\mc I \setminus \{j,n\} } }{}^2 } }  } \\
     \geq\;& \frac{\mu}{2} \brac{ \frac{1}{\sqrt{ q_i^2 + \norm{ \mb q_{\mc I \setminus \{j,n\} } }{}^2 } }  -  \frac{1}{\sqrt{ q_n^2 + \norm{ \mb q_{\mc I \setminus \{j,n\} } }{}^2 } }  }\;=\; \frac{\mu}{4} \int_{q_i^2}^{q_n^2}  \frac{1}{\paren{ t^2 + \norm{ \mb q_{\mc I \setminus \{j,n\} } }{}^2 }^{3/2} } dt.
\end{align*}
Therefore, we have
\begin{align*}
&\innerprod{ \grad \bb E\brac{\wt{f}(\mb q)}}{ \frac{1}{q_j} \mb e_j - \frac{1}{q_n} \mb e_n } \\
\geq\;& \frac{\theta(1-\theta)}{4} \int_{q_i^2}^{q_n^2}  \frac{1}{\paren{ t^2 + \norm{ \mb q_{\mc I \setminus \{j,n\} } }{}^2 }^{3/2} } dt \\
\geq\;& \frac{\theta(1-\theta)}{4} \paren{ q_n^2 - \norm{ \mb q_{-n} }{\infty}^2} \;\geq\; \frac{\theta(1-\theta)}{4} \frac{\xi}{1+\xi} q_n^2 \;\geq\; \frac{\theta(1-\theta)}{4n} \frac{\xi}{1+\xi}.
\end{align*}
This gives the desired result.
\end{proof}

%% file: sec/app_concentration.tex
In this section, under the same settings of Appendix \ref{app:regularity-population}, we uniformly bound the deviation between the empirical process $\grad \wt{f}(\mb q)$ and its mean $\bb E\brac{ \grad \wt{f}(\mb q) }$ over the sphere. Namely, we show the following results.
\begin{proposition}\label{prop:gradient-coordinate-concentration}
For every $i \in [n]$ and any $\delta \in (0,1)$, when 
\begin{align}\label{eqn:bound-p-1}
	p \geq C \delta^{-2}n \log \paren{ \frac{ \theta n }{ \mu \delta } },
\end{align}
we have
\begin{align*}
	 \sup_{\mb q\in \bb S^{n-1} }\; \abs{\innerprod{\grad \wt{f}(\mb q) - \bb E\brac{\grad \wt{f}(\mb q)}}{ \mb e_i}} \leq \delta 
\end{align*}
holds with probability at least $ 1- n p^{-c_1\theta n} -n \exp\paren{ -c_2 p \delta^2 }$, for any $\mb e_i$. Here, $c_1,\;c_2$, and $C$ are some universal positive numerical constants.
\end{proposition}

\paragraph{Remarks.} Here, our bound is loose by roughly a factor of $n$ because of the looseness in handling the probabilistic dependency due to the convolution measurement. We believe this bound can be improved by an order of $\mc O(n)$ using more advanced probability tools, such as decoupling and chaining \cite{de2012decoupling,krahmer2014suprema,qu2017convolutional}.

\begin{proof}
First, note that
\begin{align}\label{eqn:wt-f-grad}
   \wt{f}(\mb q) \;=\; \frac{1}{np} \sum_{i=1}^p H_\mu\paren{ \mb C_{ \mb x_i } \mb q },\quad 
   \grad \wt{f}(\mb q) \;=\; \frac{1}{np} \mc P_{\mb q^\perp}\sum_{i=1}^p \mb C_{\mb x_i}^\top \nabla h_\mu\paren{ \mb C_{\mb x_i} \mb q }.
\end{align}
Thus, we have
\begin{align*}
   &\innerprod{ \grad \wt{f}(\mb q) -  \bb E\brac{ \grad \wt{f}(\mb q)} }{ \mb e_n }	\\
   =\;& \frac{1}{np} \sum_{i=1}^p \sum_{j=0}^{n-1} \brac{ \innerprod{\mc P_{\mb q^\perp} \shift{ \wc{\mb x}_i }{j} }{\mb e_n} \nabla h_\mu\paren{  \shift{ \wc{\mb x}_i }{j}^\top \mb q }  -  \bb E\brac{ \paren{ \mb e_n^\top \mc P_{\mb q^\perp}\mb x} \nabla h_\mu\paren{ \mb x^\top \mb q }   }  }.
\end{align*}
This is a summation of dependent random variables, which is very difficult to show measurement concentration in general. We alleviate this difficulty by only considering a partial summation of independent random variables, namely,
\begin{align*}
   \mc L(\mb q) \;=\; \frac{1}{p} \frac{1}{ \norm{  \mb P_{\mb q^\perp} \mb e_n }{} } \sum_{i=1}^p 	 \brac{ \innerprod{\mc P_{\mb q^\perp}  \mb x_i  }{\mb e_n} \nabla h_\mu\paren{  \mb x_i^\top \mb q }  -  \bb E\brac{  \paren{ \mb e_n^\top \mc P_{\mb q^\perp}\mb x} \nabla h_\mu\paren{ \mb x^\top \mb q }   }  },
\end{align*}
where $\mb x_i \sim_{i.i.d.} \mc {BG}(\theta)$. Note that the bound of $\mc L(\mb q)$ automatically gives an upper bound of $$\innerprod{ \grad \wt{f}(\mb q) -  \bb E\brac{ \grad \wt{f}(\mb q)} }{ \mb e_n }$$ in distribution. To uniformly control $\mc L(\mb q)$ over the sphere, we first consider controlling $\mc L(\mb q)$ for a fixed $\mb q\in \bb S^{n-1}$. For each $\ell = 1, 2, \cdots$, we have the moments
\begin{align*}
   \bb E\brac{ \abs{  \innerprod{\mc P_{\mb q^\perp}  \mb x_i  }{\mb e_n} \nabla h_\mu\paren{  \mb x_i^\top \mb q } }^\ell   } \;\leq\; \bb E\brac{ \abs{ \mb e_n^\top \mc P_{\mb q^\perp}  \mb x_i }^\ell } \;=\; \bb E\brac{ \abs{Z_i}^\ell },
\end{align*}
where conditioned on the Bernoulli distribution, we have $Z_i \sim \mc N\paren{ 0, \norm{ \paren{\mc P_{\mb q^\perp} \mb e_n}_{\mc J} }{}^2 } $. By Lemma \ref{lem:gaussian_moment}, we have
\begin{align*}
    	\bb E\brac{ \abs{  \innerprod{\mc P_{\mb q^\perp}  \mb x_i  }{\mb e_n} \nabla h_\mu\paren{  \mb x_i^\top \mb q } }^\ell   } \;\leq\; \bb E_{\mc J}\brac{ \paren{\ell -1}!! \norm{ \paren{\mc P_{\mb q^\perp} \mb e_n}_{\mc J} }{}^\ell } \leq \frac{\ell!}{2} \norm{  \mb P_{\mb q^\perp} \mb e_n }{}^{\ell} ,
\end{align*}
where we used the fact that $\abs{\nabla h_\mu(z)}\leq 1$ for any $z$. Thus, we are controlling the concentration of summation of sub-Gaussian r.v., for which we have
\begin{align*}
   \bb P\paren{  \abs{ \mc L(\mb q) } \geq t } \;\leq \; \exp\paren{ - C\frac{pt^2}{ 2 }  }.	
\end{align*}
Next, we turn this point-wise concentration into a uniform bound for all $\mb q\in \bb S^{n-1}$  via a standard covering argument. Let $\mc N(\eps)$ be an $\eps$-net of the sphere, whose cardinality can be controlled by
\begin{align*}
   \abs{ \mc N(\eps) } \;\leq\; \paren{ \frac{3}{ \eps } }^{n-1}. 	
\end{align*}
Thus, we have
\begin{align*}
   \bb P\paren{ \sup_{\mb q \in \mc N(\eps) } \abs{\mc L(\mb q)} \;\geq\; t } \;\leq\;\paren{ \frac{3}{ \eps } }^{n-1} 	\exp\paren{ - \frac{pt^2}{ 2+2t } }.
\end{align*}
For any point $\mb q \in \bb S^{n-1}$, it can written as $\mb q = \mb q' + \mb e$, where $\mb q' \in \mc N(\eps)$ and $\norm{\mb e}{}\leq \eps$. Now we control the all points over the sphere through the Lipschitz property of $\mc L$.
\begin{align*}
  	&\sup_{\mb q \in \bb S^{n-1} } \abs{\mc L(\mb q)} \\
  	=\;& \sup_{\mb q'\in \mc N(\eps), \norm{\mb e}{} \leq\eps   } \abs{\mc L(\mb q'+\mb e)} \\
  	\leq\;& \sup_{\mb q'\in \mc N(\eps) } \abs{ \mc L(\mb q') } + \underbrace{ \sup_{\mb q'\in \mc N(\eps), \norm{\mb e}{} \leq\eps   }  \abs{  \bb E\brac{ \paren{ \mb e_n^\top \mc P_{ (\mb q'+\mb e  )^\perp}\mb x - \mb e_n^\top \mc P_{(\mb q')^\perp}\mb x } \nabla h_\mu\paren{ \mb x^\top\mb q' }   } } }_{\mc L_1} \\
  	&+\underbrace{ \sup_{\mb q'\in \mc N(\eps), \norm{\mb e}{} \leq\eps   }  \abs{  \bb E\brac{ \paren{ \mb e_n^\top \mc P_{ (\mb q'+\mb e  )^\perp}\mb x } \paren{\nabla h_\mu\paren{ \mb x^\top (\mb q' +\mb e) }   - \nabla h_\mu\paren{ \mb x^\top \mb q' } }  } } }_{\mc L_2} \\
  	& + \underbrace{\sup_{\mb q'\in \mc N(\eps), \norm{\mb e}{} \leq\eps   } \abs{ \frac{1}{p} \sum_{i=1}^p \brac{ \mb e_n^\top \mc P_{(\mb q'+\mb e)^\perp}  \mb x_i - \mb e_n^\top \mc P_{(\mb q')^\perp}  \mb x_i }   \nabla h_\mu(\mb x_i^\top \mb q') } }_{\mc L_3}\\
  	&+ \underbrace{\sup_{\mb q'\in \mc N(\eps), \norm{\mb e}{} \leq\eps   } \abs{ \frac{1}{p} \sum_{i=1}^p  \paren{\mb e_n^\top \mc P_{(\mb q'+\mb e)^\perp}  \mb x_i  } \brac{ \nabla h_\mu\paren{  \mb x_i^\top (\mb q' + \mb e)  }  -  \nabla h_\mu\paren{  \mb x_i^\top \mb q' } }  } }_{\mc L_4}.
\end{align*}
By Lipschitz continuity and the fact that $\nabla h_\mu(z)\leq 1$ for any $z$, we obtain
\begin{align*}
   \mc L_1 \;&\leq \; \sup_{\mb q'\in \mc N(\eps), \norm{\mb e}{} \leq\eps   } \sqrt{\theta} \norm{ \paren{ \mc P_{(\mb q' + \mb e)^\perp } - \mc P_{(\mb q')^\perp } } \mb e_n }{} \;\leq\; 3\sqrt{\theta} \eps \\
   \mc L_2 \;&\leq \;  \sup_{\mb q'\in \mc N(\eps), \norm{\mb e}{} \leq\eps   }  \frac{1}{\mu} \bb E\brac{  \norm{\mb x}{}  \norm{ \mb x^\top \mb e }{} } \;\leq \; \frac{\theta n}{\mu} \eps.
\end{align*}
For each $\mb x_i$, we know that $\mb x_i = \mb g_i \odot \mb b_i$ with $\mb g_i \sim \mc N(\mb 0,\mb I)$ and $\mb b_i \sim_{i.i.d.} \mc B(\theta)$. By Gaussian concentration inequality, we know that for each $\mb x_i$,
\begin{align*}
   \bb P\paren{ \norm{\mb x_i}{} - \sqrt{\theta n} \geq t  }  \leq \bb P\paren{ \norm{\mb x_i}{} - \bb E\brac{ \norm{\mb x_i}{} } \geq t  } \leq \exp\paren{ -\frac{t^2}{2 \norm{\mb b_i}{\infty}  } }	 \leq \exp\paren{ -\frac{t^2}{2  } }	. 
\end{align*}
 Therefore, by a union bound, we have
 \begin{align*}
     \max_{1\leq i\leq p} \norm{\mb x_i}{} \leq 5\sqrt{ \theta n \log p }
 \end{align*}
holds with probability at least $1 - p^{-8\theta n}$. Therefore, w.h.p we have
\begin{align*}
   \mc L_3 \;&\leq \; \paren{\max_{1\leq i\leq p} \norm{\mb x_i}{} } \sup_{\mb q'\in \mc N(\eps), \norm{\mb e}{} \leq\eps   }  \norm{ \mc P_{(\mb q' + \mb e)^\perp } - \mc P_{(\mb q')^\perp } }{} \;\leq\; 15  \sqrt{ \theta n \log p } \eps, \\
   \mc L_4 \;&\leq \; \frac{1}{\mu} \paren{\max_{1\leq i\leq p} \norm{\mb x_i}{}^2 } \sup_{\mb q'\in \mc N(\eps), \norm{\mb e}{} \leq\eps } \norm{\mb e}{} \;\leq\; 25\frac{\theta n \log p}{\mu} \eps.
\end{align*}
Combining the bounds above, choose $\eps = \frac{\mu t}{ c\theta n \log p } $, we have
\begin{align*}
   \sup_{\mb q \in \bb S^{n-1} } \abs{\mc L(\mb q)} \leq \sup_{\mb q'\in \mc N(\eps) } \abs{ \mc L(\mb q') } + c \frac{\theta n\log p}{\mu} \eps \leq 2t
\end{align*}
holds with probability at least
\begin{align*}
 1 - p^{-8\theta n} - \exp\paren{ - C \frac{pt^2}{2} + c'n \log\paren{ \frac{ \theta n }{\mu t}  }  }.
\end{align*}
Thus, applying a union bound, we obtain the desired result holding for every $i \in [n]$.
\end{proof}

Similarly, we also show the following result.

\begin{corollary}\label{cor:gradient-concentration}
For any $\delta \in (0,1)$, when 
\begin{align}\label{eqn:bound-p-2}
	p \geq C \delta^{-2}n^2 \log \paren{ \frac{ \theta n }{ \mu \delta } },
\end{align}
we have
\begin{align*}
    \sup_{\mb q\in \bb S^{n-1} }\; \norm{\grad \wt{f}(\mb q) - \bb E\brac{\grad \wt{f}(\mb q)}}{ } \;&\leq\; \delta, \\
    \sup_{\mb q\in \bb S^{n-1} }\; \norm{\nabla \wt{f}(\mb q) - \bb E\brac{\nabla \wt{f}(\mb q)}}{ } \;&\leq\; \delta,
\end{align*}
hold with probability at least $ 1- p^{-c_1\theta n} -n\exp\paren{ -c_2 p \delta^2 }$. Here, $c_1,\;c_2$, and $C$ are some universal positive numerical constants.
\end{corollary}

\begin{proof}
From Proposition \ref{prop:gradient-coordinate-concentration}, we know that when $p \geq C_0 \eps^{-2}n \log \paren{ \frac{ \theta n }{ \mu \eps } }$,
\begin{align*}
    & \sup_{q \in \bb S^{n-1}} \norm{\grad \wt{f}(\mb q) - \bb E\brac{\grad \wt{f}(\mb q)}}{ }^2 \\
   \leq \;& \sum_{i=1}^n  \sup_{q \in \bb S^{n-1}} \abs{\innerprod{\grad \wt{f}(\mb q) - \bb E\brac{\grad \wt{f}(\mb q)}}{ \mb e_i}}^2 \;\leq\; n \eps^2.
\end{align*}
holds with probability at least $ 1- p^{-c_1\theta n} -n\exp\paren{ -c_2 p \delta^2 }$. Therefore, by letting $\delta = \sqrt{n }\eps$, w.h.p. we have 
\begin{align*}
   	\sup_{q \in \bb S^{n-1}} \norm{\grad \wt{f}(\mb q) - \bb E\brac{\grad \wt{f}(\mb q)}}{ } \leq \delta,
\end{align*}
 whenever $p \geq C\delta^{-2}n^2 \log \paren{ \frac{ \theta n }{ \mu \delta } }$. By a similar argument, we can also provide the same bound for $$ \sup_{\mb q\in \mc S^{n-1} } \norm{\nabla \wt{f}(\mb q) - \bb E\brac{\nabla \wt{f}(\mb q)}}{ }$$.
\end{proof}

\begin{corollary}\label{cor:gradient-bound-infty}
For each $i \in [n]$ and any $\delta \in (0,1)$, when $p \geq C \delta^{-2} n \log \paren{ \frac{ \theta n }{ \mu \delta } }$, we have
\begin{align*}
   \sup_{q \in \bb S^{n-1} }\;	 \abs{\innerprod{ \grad \wt{f}(\mb q) }{\mb e_i}  } \;\leq \; 1+ \delta,
\end{align*}
hold with probability at least $ 1- np^{-c_1\theta n} -n\exp\paren{ -c_2 p \delta^2 }$. Here, $c_1,\;c_2$, and $C$ are some universal positive numerical constants.
\end{corollary}

\begin{proof}
For any $\mb q \in \bb S^{n-1}$ and every $i \in [n]$, we have
\begin{align*}
   \bb E\brac{ \abs{\innerprod{ \grad \wt{f}(\mb q) }{\mb e_i}  }} \;=\; \bb E\brac{\abs{ \paren{\mb e_i^\top \mc P_{\mb q^\perp} \mb x}\cdot \nabla h_\mu(\mb x^\top \mb q) }{}  } \leq \bb E\brac{ \norm{\mb e_i^\top \mc P_{\mb q^\perp} \mb x }{} } \leq 1. 
\end{align*}
Thus, we have
\begin{align*}
   & \sup_{q \in \bb S^{n-1} } 	\abs{ \innerprod{ \grad \wt{f}(\mb q) - \bb E\brac{\grad \wt{f}(\mb q)} }{ \mb e_i } }\\ \geq \;& \sup_{q \in \bb S^{n-1} } \paren{	\abs{\innerprod{ \grad \wt{f}(\mb q) }{\mb e_i}  } - \bb E\brac{\abs{\innerprod{ \grad \wt{f}(\mb q) }{\mb e_i}  } } } \\
    \;\geq \;& \sup_{q \in \bb S^{n-1} }	\abs{\innerprod{ \grad \wt{f}(\mb q) }{\mb e_i}  } \;-\; \sup_{q \in \bb S^{n-1} }\bb E\brac{\abs{\innerprod{ \grad \wt{f}(\mb q) }{\mb e_i}  }}.
\end{align*}
Therefore, by using the result in Proposition \ref{prop:gradient-coordinate-concentration}, we obtain the desired result.
\end{proof}

\begin{corollary}\label{cor:gradient-bound}
For any $\delta \in (0,1)$, when $p$ satisfies \eqref{eqn:bound-p-2}, we have
\begin{align*}
   \sup_{q \in \bb S^{n-1} }\;	\norm{\grad \wt{f}(\mb q)}{} \;\leq \; \sqrt{ \theta n } + \delta,
\end{align*}
hold with probability at least $ 1- p^{-c_1\theta n} -n\exp\paren{ -c_2 p \delta^2 }$. Here, $c_1,\;c_2$, and $C$ are some universal positive numerical constants.
\end{corollary}

\begin{proof}
For any $\mb q \in \bb S^{n-1}$, we have
\begin{align*}
   \bb E\brac{\norm{\grad \wt{f}(\mb q)}{ } } \;=\; \bb E\brac{\norm{ \mc P_{\mb q^\perp} \mb x \nabla h_\mu(\mb x^\top \mb q) }{}} \leq \bb E\brac{ \norm{\mb x}{} } \leq \sqrt{\theta n}. 
\end{align*}
Note that
\begin{align*}
    \sup_{q \in \bb S^{n-1} } 	\norm{\grad \wt{f}(\mb q) - \bb E\brac{\grad \wt{f}(\mb q)}}{ } \;&\geq \;\sup_{q \in \bb S^{n-1} } \paren{	\norm{\grad \wt{f}(\mb q)}{} - \bb E\brac{\norm{\grad \wt{f}(\mb q)}{ } } } \\
    \;&\geq \; \sup_{q \in \bb S^{n-1} }	\norm{\grad \wt{f}(\mb q)}{} \;-\; \sup_{q \in \bb S^{n-1} }\bb E\brac{\norm{\grad \wt{f}(\mb q)}{ }}.
\end{align*}
Thus, by using the result in Corollary \ref{cor:gradient-concentration}, we obtain the desired result.
\end{proof}

%% file: sec/app_preconditioning.tex
In this section, given the Riemannian gradient of $\wt{f}(\mb q)$ in \eqref{eqn:problem-simple} and its preconditioned variant
\begin{align*}
   	\grad \wt{f}(\mb q) \;&=\; \frac{1}{np} \mc P_{\mb q^\perp}  \sum_{i=1}^p \mb C_{\mb x_i}^\top \nabla h_\mu \paren{ \mb C_{\mb x_i} \mb q }, \\
   	\grad f(\mb q) \;&=\; \frac{1}{np}  \mc P_{\mb q^\perp}  \sum_{i=1}^p \paren{ \mb R \mb Q^{-1} }^\top \mb C_{\mb x_i}^\top \nabla h_\mu \paren{ \mb C_{\mb x_i}\paren{ \mb R \mb Q^{-1} }\mb q },
\end{align*}
with
\begin{align*}
   \mb R \;=\; \mb C_{\mb a}\paren{ \frac{1}{\theta np} \sum_{i=1}^p \mb C_{\mb y_i}^\top \mb C_{\mb y_i}  }^{-1/2},\quad \mb Q \;=\; \mb C_{\mb a} \paren{ \mb C_{\mb a}^\top \mb C_{\mb a} }^{-1/2} ,
\end{align*}
we prove that they are very close via a perturbation analysis by using the Lipschitz property of first-order derivative of Huber loss.

\begin{proposition}\label{prop:preconditioning}
Suppose $\theta\geq \frac{1}{n}$. For any $\delta \in (0,1)$, whenever 
\begin{align*}
   p \geq C \frac{\kappa^8 n}{  \mu^2 \theta \delta^2 \sigma_{\min}^2 }  \log^4 n \log \paren{ \frac{ \theta n }{ \mu } },
\end{align*}
we have
\begin{align*}
    \sup_{ \mb q \in \bb S^{n-1} }\; \norm{ \grad \wt{f}(\mb q) - \grad f(\mb q)	}{} \;\leq\; \delta
\end{align*}
holds with probability at least $1 -c_1p^{-c_2 n \theta } - n^{-c_3} - n e^{-c_4 \theta np } $. Here, $\kappa$ and $\sigma_{\min}$ denote the condition number and minimum singular value of $\mb C_{\mb a}$, and $c_1,\;c_2,\;c_3\;,c_4$ and $C$ are some positive numerical constants.
\end{proposition}

\begin{proof}
Notice that
\begin{align*}
   \mb R\mb Q^{-1} \;=\;	  \mb C_{\mb a} \paren{ \frac{1}{\theta np} \sum_{i=1}^p \mb C_{\mb y_i}^\top \mb C_{\mb y_i}  }^{-1/2} \paren{ \mb C_{\mb a}^\top \mb C_{\mb a} }^{1/2} \mb C_{\mb a}^{-1}.
\end{align*}
Thus, we have
\begin{align}
   &\sup_{ \mb q \in \bb S^{n-1} }\;\norm{ \grad \wt{f}(\mb q) - \grad f(\mb q)	}{} \nonumber \\
   \leq\;& \frac{1}{np} \norm{ \mc P_{\mb q^\perp} \paren{\mb I - \paren{ \mb R \mb Q^{-1} }  }^\top \sum_{i=1}^p \mb C_{\mb x_i}^\top \nabla h_\mu\paren{ \mb C_{\mb x_i}\mb q } }{} \nonumber \\
   & + \frac{1}{np} \norm{ \mc P_{\mb q^\perp}\paren{ \mb R \mb Q^{-1} }^\top \sum_{i=1}^p \mb C_{\mb x_i}^\top \brac{ \nabla h_\mu \paren{ \mb C_{\mb x_i} \mb q } -  \nabla h_\mu \paren{ \mb C_{\mb x_i}\paren{ \mb R \mb Q^{-1} }\mb q } } }{} \nonumber \\
   \leq\;&  \norm{ \mb I - \mb R\mb Q^{-1} }{}  \norm{ \nabla \wt{f}(\mb q) }{} + \norm{ \mb R \mb Q^{-1} }{} \norm{ \frac{1}{np} \sum_{i=1}^p  \mb C_{\mb x_i}^\top \brac{ \nabla h_\mu \paren{ \mb C_{\mb x_i} \mb q } -  \nabla h_\mu \paren{ \mb C_{\mb x_i}\paren{ \mb R \mb Q^{-1} }\mb q } } }{} \nonumber \\
   \leq\;& \norm{ \mb I - \mb R\mb Q^{-1} }{}  \norm{ \nabla \wt{f}(\mb q) }{} + \frac{1}{\mu \sqrt{n} } \norm{ \mb R\mb Q^{-1} }{} \paren{ \max_{1\leq i \leq p}  \norm{\mb x_i}{} \norm{\mb F\mb x_i}{ \infty } } \norm{ \mb I - \mb R \mb Q^{-1} }{}. \label{eqn:precondition-0}
\end{align}
Here, by Lemma \ref{lem:precond-bounds}, for any given $\eps \in (0,1)$, when $p \;\geq\; C  \frac{ \kappa^8 }{\theta \eps^2 \sigma_{\min}^{2}(\mb C_{\mb a}) } \log^3 n$, we have
\begin{align}\label{eqn:preconding-bound}
    \norm{ \mb R\mb Q^{-1} - \mb I}{} \;\leq\; \eps,\quad  	\norm{ \mb R\mb Q^{-1} }{} \;\leq \; 1+ \eps,
\end{align}
holding with probability at least $ 1- p^{-c_1 n \theta } - n^{-c_2}$. On the other hand, by Gaussian concentration inequality and a union bound, we have
\begin{align}
   \max_{1\leq i\leq p} \norm{\mb x_i}{} \leq 4\sqrt{ n \log p }, \quad \max_{1\leq i \leq p} \norm{\mb F \mb x_i}{\infty} \leq 4 \sqrt{n \log p}, \label{eqn:precondition-3}
\end{align}
hold with probability at least $1 - p^{-c_3n} $. By Corollary \ref{cor:gradient-bound}, when $p \geq C_2\theta^{-1} n \log \paren{ \frac{ \theta n }{ \mu } } $, we have
\begin{align}
    \sup_{q \in \bb S^{n-1} }\;	\norm{\grad \wt{f}(\mb q)}{} \;\leq \; 2\sqrt{ \theta n } \label{eqn:precondition-4}
\end{align}
holds with probability at least $ 1- p^{-c_4\theta n} -ne^{ -c_5  \theta np }$. Plugging the bounds in \eqref{eqn:preconding-bound} and \eqref{eqn:precondition-3} into \eqref{eqn:precondition-0}, we obtain 
\begin{align*}
   \sup_{ \mb q \in \bb S^{n-1} }\;\norm{ \grad \wt{f}(\mb q) - \grad f(\mb q)	}{} 
   \;\leq \;  \eps  \brac{ 2\sqrt{ \theta n } \;+\; \frac{16\sqrt{n} \log p}{\mu} \cdot \paren{ 1+\eps }  }.
\end{align*}
By a change of variable, we obtain the desired result.
\end{proof}

\begin{lemma}\label{lem:C_x-concentration}
When $\theta \geq 1/n$,
\begin{align}
  \norm{\frac{1}{\theta n p } \sum_{i=1}^p \mb C_{\mb x_i}^\top \mb C_{\mb x_i} - \mb I }{} \leq t
\end{align}
holds with probability at least $1 - p^{-c_1n\theta} - n \exp\paren{ - c_2\min \Brac{ \frac{pt^2}{ \theta\log p} , \frac{pt}{ \sqrt{ \theta \log p}}  } } $ for some numerical constants $c_1, c_2 >0$.
\end{lemma}

\begin{proof}
Notice that
\begin{align*}
     \mb C_{\mb x_i}^\top \mb C_{\mb x_i} = \mb F^* \diag\paren{ \abs{ \mb F \mb x_i }^{\odot 2} } \mb F.
\end{align*}
Then 
\begin{align}
  \norm{\frac{1}{\theta n p } \sum_{i=1}^p \mb C_{\mb x_i}^\top \mb C_{\mb x_i} - \mb I }{} \;&=\; \norm{ \mb F^* \paren{ \diag\paren{ \frac{1}{ \theta np } \sum_{i=1}^p \abs{ \mb F \mb x_i }^{\odot 2}   }  - \mb F^{-1} (\mb F^*)^{-1}  }  \mb F  }{} \nonumber  \\
  \;&=\; \norm{ \frac{1}{ \theta np } \sum_{i=1}^p \abs{ \mb F \mb x_i }^{\odot 2}  - \mb 1 }{\infty }.\label{eqn:trans-circulant}
\end{align}

Let $\mb x_i = \mb b_i \odot \mb g_i$ with $\mb b_i \sim_{i.i.d.} \mc B(\theta)$ and $\mb g_i \sim \mc N(\mb 0,\mb I)$, and let us define events
\begin{align*}
   \mc E_{i,j} \doteq \Brac{ \norm{ \mb b_i \odot \mb f_j }{}^2 \leq 5n \sqrt{\theta \log p}   }, \quad 1\leq i \leq p,\; 1\leq j \leq n.
\end{align*}
We use $\mc E_j = \bigcap_{i=1}^p \mc E_{i,j} $. For each individual $i$ and $j$, by the Hoeffding's inequality, we have
\begin{align*}
	\bb P \paren{ \mc E_{i,j}^c  } \leq \exp \paren{ - 8n \theta \log p }
\end{align*}
For each $j = 1,\cdots , n$, by conditional probability and union bound, we have

\begin{align}
   \bb P\paren{\abs{  \frac{1}{\theta n p} \sum_{i=1}^p  \abs{ \mb f_j^* \mb x_i }^2  - 1 } \geq t}
  &\leq   \bb P \paren{ \bigcup_{i=1}^p \mc E_{i,j}^c  } + \bb P\paren{\abs{  \frac{1}{\theta n p} \sum_{i=1}^p  \abs{ \mb f_j^* \mb x_i }^2  - 1 } \geq t \mid   \mc E_j } \nonumber \\
  &\leq \sum_{i=1}^p \bb P \paren{ \mc E_{i,j}^c  } + \bb P\paren{\abs{  \frac{1}{\theta n p} \sum_{i=1}^p  \abs{ \mb f_j^* \mb x_i }^2  - 1 } \geq t \mid   \mc E_j } \nonumber \\
  &\leq p e^{ - 8n \theta \log p } + \bb P\paren{\abs{  \frac{1}{\theta n p} \sum_{i=1}^p  \abs{ \mb f_j^* \mb x_i }^2  - 1 } \geq t \mid   \mc E_j }. \label{eqn:bound-condition-event-j}
  \end{align}
For the second term, since $\mb x_i \sim \mc {BG}(\theta) $, we have
\begin{align*}
	\mb f_j^*\mb x_i = \sum_{k=1}^n f_{ji} b_{ik} g_{ik} \sim \mc N \paren{ 0, \norm{\mb b_i \odot \mb f_j }{}^2 }
\end{align*}
for all $\ell \geq 1$, by Lemma \ref{lem:gaussian_moment}, we have
\begin{align*}
   \bb E\brac{ (\theta n)^{-\ell} \abs{ \mb f_j^* \mb x_i }^{2\ell} \mid \mc E_{i,j}  } 
   \;&=\;  \frac{ (2\ell-1)!! }{(\theta n)^{\ell}} \bb E\brac{ \norm{\mb b \odot \mb f }{}^{2\ell} \mid \mc E_{i,j} } \\
    \;&\leq\;\frac{\ell!}{2} 10^\ell \theta^{-\ell/2} \log^{\ell/2} p.
\end{align*}
Thus, by Bernstein inequality in Lemma \ref{lem:mc_bernstein_scalar}, we have
\begin{align}
\bb P\paren{\abs{  \frac{1}{\theta n p} \sum_{i=1}^p  \abs{ \mb f_j^* \mb x_i }^2  - 1 } \geq t \mid  \mc E_j }	\;&\leq\; \exp\paren{  -\frac{ pt^2  }{  200 \theta \log p  + 20 \sqrt{ \theta \log p} t } }  \nonumber \\
 \;&\leq \; \exp\paren{ - \min \Brac{ \frac{pt^2}{ 400 \theta \log p} , \frac{pt}{ 40 \sqrt{\theta \log p}}  } }. \label{eqn:tail-event-j}
\end{align}
Plugging \eqref{eqn:tail-event-j} into \eqref{eqn:bound-condition-event-j}, we obtain
\begin{align*}
	\abs{  \frac{1}{\theta n p} \sum_{i=1}^p  \abs{ \mb f_j^* \mb x_i }^2  - 1 } \leq t
\end{align*}
holds with high probability for each $j = 1,\cdots, n$. We apply a union bound to control the $\ell_\infty$-norm in \eqref{eqn:trans-circulant}, and hence get the desired result.
\end{proof}

\begin{lemma}\label{lem:purterbation-bound}
	For any $\eps \in (0,1)$, when $p \geq C \theta^{-1} \eps^{-2} \log^3 n $, we have
\begin{align*}
    \norm{ \frac{1}{\theta n p } \sum_{i=1}^p \mb C_{\mb y_i}^\top \mb C_{\mb y_i} }{} \;&\leq\; \paren{1 +\eps }\norm{\mb C_{\mb a} }{}^2 \\
	\norm{  \paren{  \frac{1}{\theta n p } \sum_{i=1}^p \mb C_{\mb y_i}^\top \mb C_{\mb y_i} }^{-1/2} - \paren{\mb C_{\mb a}^\top \mb C_{\mb a}}^{-1/2} }{} \;&\leq\; \frac{4 \kappa^2 \eps }{ \sigma_{\min}^2(\mb C_{\mb a}) } 
\end{align*}
holds with probability at least $1 - p^{-c_1 n \theta} - n^{-c_2}$. Here, $\kappa$ is the condition number of $\mb C_{\mb a}$, and $\sigma_{\min}(\mb C_{\mb a})$ is the smallest singular value of $\mb C_{\mb a}$.
\end{lemma}

\begin{proof}
For any $\eps \in (0,1)$, from Lemma \ref{lem:C_x-concentration},  when $p \geq C \theta^{-1}\eps^{-2} \log^3 n $ we know that the event
\begin{align*}
  	\mc E(\eps) \;\doteq\; \Brac{\norm{ \frac{1}{\theta np} \sum_{i=1}^p \mb C_{\mb x_i}^\top \mb C_{\mb x_i} - \mb I }{} \leq \eps}
\end{align*}
holds with probability at least $1 - p^{-c_1 n \theta} - n^{-c_2}$. Conditioned on the event $\mc E(\eps)$, let us denote
\begin{align*}
   \mb A \;=\; \mb C_{\mb a}^\top \mb C_{\mb a} \succ \mb 0,
\end{align*}
and let $\sigma_{\max}\paren{\mb A},\;\sigma_{\min}\paren{\mb A}$ be the largest and smallest singular values of $\mb A$, respectively. Then we observe,
\begin{align*}
   \frac{1}{\theta np} \sum_{ i=1}^p \mb C_{\mb y_i}^\top \mb C_{\mb y_i}	 \;&=\; \mb C_{\mb a}^\top \mb C_{\mb a} + \underbrace{\mb C_{\mb a}^\top \brac{ \frac{1}{\theta np} \sum_{i=1}^p \mb C_{\mb x_i}^\top \mb C_{\mb x_i} - \mb I }  \mb C_{\mb a} }_{\mb \Delta}, \\
   \;&=\; \mb A + \mb \Delta,\qquad \norm{\mb \Delta}{} \;\leq\; \eps \cdot \sigma_{\max}(\mb A) .
\end{align*}
Therefore, we have
\begin{align*}
   	\norm{ \frac{1}{\theta np} \sum_{ i=1}^p \mb C_{\mb y_i}^\top \mb C_{\mb y_i}}{} \;\leq \; \norm{ \mb A }{} + \norm{\mb \Delta}{} \;\leq \paren{ 1+ \eps } \norm{\mb C_{\mb a}}{}^2.
\end{align*}
By Lemma \ref{lem:matrix-perturbation}, whenever
\begin{align*}
   \norm{ \mb \Delta }{} \leq \frac{1}{2} \sigma_{\min}(\mb A) \quad \Longrightarrow \quad \eps \;\leq\; \frac{1}{2} \frac{\sigma_{\min}(\mb A)}{\sigma_{\max}(\mb A)} = \frac{1}{2\kappa^2},
\end{align*}
we know that
\begin{align*}
	\norm{  \paren{  \frac{1}{\theta n p } \sum_{i=1}^p \mb C_{\mb y_i}^\top \mb C_{\mb y_i} }^{-1/2} - \paren{\mb C_{\mb a}^\top \mb C_{\mb a}}^{-1/2} }{} \;&=\; \norm{ \paren{ \mb A + \mb \Delta }^{-1/2} - \mb A^{-1/2} }{} \\
	\;&\leq\; \frac{ 4 \norm{ \mb \Delta }{}  }{ \sigma_{\min}^2(\mb A) } \;\leq\; \frac{4 \eps \sigma_{\max}(\mb A) }{ \sigma_{\min}^2(\mb A) } \;=\; \frac{4 \kappa^2 \eps }{ \sigma_{\min}^2(\mb C_{\mb a})}. 
\end{align*}
\end{proof}

\begin{lemma}\label{lem:precond-bounds}
Let $\theta \in (1/n, 1/3)$, and given a $\delta \in (0,1)$. Whenever 
\begin{align*}
   p \;\geq\; C  \frac{ \kappa^8 }{\theta \delta^2 \sigma_{\min}^{2}(\mb C_{\mb a}) } \log^3 n, 	
\end{align*}
we have
\begin{align*}
    \norm{ \mb R\mb Q^{-1} - \mb I }{}  \;&\leq\; \delta,\;\;\qquad \norm{ \mb R\mb Q^{-1} }{} \;\leq \;  	1+ \delta, \\
    \norm{ \paren{\mb R\mb Q^{-1}}^{-1} - \mb I}{}  \;&\leq\; 2\delta,\qquad 
    \norm{ \paren{\mb R\mb Q^{-1}}^{-1} }{}  \;\leq\; 1+2\delta
\end{align*}
hold with probability at least $ 1- p^{-c_1 n \theta } - n^{-c_2}$.
\end{lemma}

\begin{proof} First, by Lemma \ref{lem:purterbation-bound}, for a given $\eps \in (0,1)$, when $p \geq C_1\theta^{-1} \eps^{-2} \log^3n $, we have
\begin{align*}
   \norm{ \mb R \mb Q^{-1} - \mb I  }{} \;&=\; \norm{ \mb I - \mb C_{\mb a} \paren{ \frac{1}{\sqrt{\theta np}} \sum_{i=1}^p \mb C_{\mb y_i}^\top \mb C_{\mb y_i}  }^{-1/2} \paren{ \mb C_{\mb a}^\top \mb C_{\mb a} }^{1/2} \mb C_{\mb a}^{-1} }{} 	\\
   \;&\leq \; \kappa \cdot \norm{\mb C_{\mb a}}{} \cdot \norm{  \paren{  \frac{1}{\theta n p } \sum_{i=1}^p \mb C_{\mb y_i}^\top \mb C_{\mb y_i} }^{-1/2} - \paren{\mb C_{\mb a}^\top \mb C_{\mb a}}^{-1/2} }{}  \\
   \;&\leq \; \kappa \norm{\mb C_{\mb a}}{} \frac{4 \kappa^2 \eps }{ \sigma_{\min}^2(\mb C_{\mb a})  } \;\leq\;   \frac{4  \kappa^4  \eps }{\sigma_{\min}(\mb C_{\mb a}) }, 
\end{align*}
and
\begin{align*}
    \norm{ \mb R\mb Q^{-1} }{} \;\leq \; 1 + \norm{ \mb I - \mb R\mb Q^{-1}  }{}	 \;\leq\; 1+ \frac{4  \kappa^4  \eps }{\sigma_{\min}(\mb C_{\mb a}) }
\end{align*}
hold with probability at least $ 1- p^{-c_1 n \theta } - n^{-c_2}$. Similarly, by Lemma \ref{lem:purterbation-bound}, 
\begin{align*}
  \norm{  \mb I -  \paren{\mb R\mb Q^{-1}}^{-1} }{} \;& =\; \norm{ \mb I - \mb C_{\mb a} \paren{ \mb C_{\mb a}^\top \mb C_{\mb a} }^{-1/2} \paren{ \frac{1}{\sqrt{\theta np}} \sum_{i=1}^p \mb C_{\mb y_i}^\top \mb C_{\mb y_i}  }^{1/2}  \mb C_{\mb a}^{-1} }{} \\
  \;&\leq \; \kappa \cdot \norm{   \frac{1}{\theta n p } \sum_{i=1}^p \mb C_{\mb y_i}^\top \mb C_{\mb y_i}  }{}^{1/2} \cdot \norm{  \paren{  \frac{1}{\theta n p } \sum_{i=1}^p \mb C_{\mb y_i}^\top \mb C_{\mb y_i} }^{-1/2} - \paren{\mb C_{\mb a}^\top \mb C_{\mb a}}^{-1/2} }{} \\
  \;&\leq \; \kappa \cdot  \frac{4 \kappa^2 \eps }{ \sigma_{\min}^2(\mb C_{\mb a})  } \cdot  (1+ \eps)^{1/2} \norm{ \mb C_{\mb a} }{} \;\leq\;  \frac{8  \kappa^4  \eps }{\sigma_{\min}(\mb C_{\mb a}) },
\end{align*}
and 
\begin{align*}
    \norm{ \paren{\mb R\mb Q^{-1}}^{-1} }{} \;\leq \; 1 + \norm{ \mb I -\paren{ \mb R\mb Q^{-1}}^{-1}  }{}	 \;\leq\; 1+ \frac{8  \kappa^4  \eps }{\sigma_{\min}(\mb C_{\mb a}) }
\end{align*}
Thus, replace $\delta =  \frac{4  \kappa^4  \eps }{\sigma_{\min}(\mb C_{\mb a}) } $, we obtain the desired result.
\end{proof}

%% file: sec/app_algorithm.tex
\begin{center}
\setlength{\tabcolsep}{11pt}
\renewcommand{\arraystretch}{2}
 \begin{table*}  
 \caption{Gradient for each different loss function}\label{tab:loss-grad}
  \centering
\resizebox{0.8\textwidth}{!}{
 \centering
 \begin{tabular}{c||c|c}
 \hline
Loss function &  $\nabla \varphi(\mb q)$ for 1D problem \eqref{eqn:problem-1d} & $\nabla \varphi(\mb Z)$ for 2D problem\tablefootnote{} \eqref{eqn:problem-2D} \\ 
 \hline
 $\ell^1$-loss & $\frac{1}{np}  \sum_{i=1}^p \wc{ \ol{\mb y}_i } \conv  \sign \paren{\ol{\mb y}_i \conv  \mb q } $  & $\frac{1}{n^2p} \sum_{i=1}^p \wc{ \ol{\mb Y} }_i  \cconv\sign\paren{ \ol{\mb Y}_i \cconv \mb Z }$ \\
   \hline 
  Huber-loss & $\frac{1}{np}  \sum_{i=1}^p \wc{ \ol{\mb y}_i } \conv  \nabla h_\mu  \paren{\ol{\mb y}_i \conv  \mb q } $ & $\frac{1}{n^2p} \sum_{i=1}^p \wc{ \ol{\mb Y} }_i  \cconv \nabla h_\mu \paren{ \ol{\mb Y}_i \cconv \mb Z }$  \\
 \hline
  $\ell^4$-loss  & $-\frac{1}{np}  \sum_{i=1}^p \wc{ \ol{\mb y}_i } \conv  \paren{\ol{\mb y}_i \conv  \mb q }^{\odot 3} $ & $-\frac{1}{n^2p} \sum_{i=1}^p \wc{ \ol{\mb Y} }_i  \cconv \paren{ \ol{\mb Y}_i \cconv \mb Z }^{\odot 3 }  $   \\
 \hline
\end{tabular}
}
\end{table*}
\end{center}
\footnotetext{Here, for 2D problem, $\wc{\mb Z}$ denotes a flip operator that flips a matrix $\mb Z \in \bb R^{n_1\times n_2}$ both vertically and horizontally, i.e., $\wc{\mb Z}_{i,j} = \mb Z_{n_1 - i+1,n_2-j+1}$.}
It should be noted that the rotated problem in \eqref{eqn:problem-rotate} and \eqref{eqn:problem-simple} are only for analysis purposes. In this section, we provide detailed descriptions of the actual implementation of our algorithms on optimizing the problem in the form of \eqref{eqn:problem}. First, we introduce the details Riemannian (sub)gradient descent method for 1D problem. Second, we discuss about subgradient methods for solving the LP rounding problem. Finally, we provide more details about how to solve problems in 2D. 

For the purpose of implementation efficiency, we describe the problem and algorithms based on circulant convolution, which is slightly different from the main sections. Because our gradient descent method works for any sparse promoting loss function (other than Huber loss), in the following we describe the problem and the algorithm in a more general form rather than \eqref{eqn:problem}. However, it should be noted that our analysis in this work is only specified for Huber loss.

\subsection{Riemannian (sub)gradient descent methods}
Here, we consider (sub)gradient descent for optimizing a more general problem 
\begin{align*}
    \min_{\mb q}\; \varphi(\mb q) := \frac{1}{np} \sum_{i=1}^p \psi ( \mb C_{\mb y_i} \mb P \mb q),\quad \text{s.t.}\; \norm{\mb q}{} = 1,
\end{align*}
where $\psi(\mb z)$ can be $\ell^1$-loss ($\psi(\mb z) = \norm{\mb z}{1}$), Huber-loss ($\psi(\mb z) = H_\mu(\mb z)$), and $\ell^4$-loss ($\psi(\mb z) = - \norm{ \mb z }{4}^4 $). The preconditioning matrix $\mb P$ can be written as 
\begin{align*}
   \mb P = \mb C_{\mb v},\quad \mb v = \mb F^{-1}\paren{ \paren{ \frac{1}{\theta np} \sum_{i = 1}^p \abs{ \wh{\mb y}_i}^{ \odot 2} }^{\odot -1/2} } ,
\end{align*}
where $\wh{\mb y}_i = \mb F \mb y_i$, so that
\begin{align*}
   \mb C_{\mb y_i} \mb P =  \mb C_{\mb y_i} \mb C_{\mb v} = \mb C_{ \mb y_i \conv \mb v } = \mb C_{\ol{\mb y}_i  },\quad \ol{\mb y}_i = \mb y_i \conv \mb v.
\end{align*}
Therefore, our problem can be rewritten as 
\begin{align}\label{eqn:problem-1d}
   \min_{\mb q}\; \varphi(\mb q) := \frac{1}{np} \sum_{i=1}^p \psi ( \ol{\mb y}_i \conv \mb q),\quad \text{s.t.}\; \norm{\mb q}{} = 1.
\end{align}
Starting from an initialization, we solve the problem via Riemannian (sub)gradient descent,
\begin{align*}
   \mb q^{(k+1)} \;=\; \mc P_{ \bb S^{n-1} } \paren{ \mb q^{(k)} - \tau^{(k)} \cdot \grad \varphi(\mb q^{(k)}) } , 
\end{align*}
where $\tau^{(k)}$ is the stepsize, and the Riemannian (sub)gradient is 
\begin{align*}
	\grad \varphi(\mb q) \;=\; \mc P_{\mb q^\perp} \nabla \varphi(\mb q),
\end{align*}
which is defined on the \emph{tangent space}\footnote{ We refer the readers to Chapter 3 of \cite{absil2009optimization} for more details.} $T_{\mb q} \bb S^{n-1}$ at a point $\mb q \in \bb S^{n-1}$. \Cref{tab:loss-grad} lists the calculation of (sub)gradients $\nabla \varphi(\mb q)$ for different loss functions. For each iteration, the projection operator $\mc P_{ \bb S^{n-1} }(\mb z) = \mb z / \norm{\mb z}{} $ retracts the iterate back to the sphere. Let $\odot$ denotes entry-wise power/multiplication, the overall algorithm is summarized in \Cref{alg:grad-descent}.

\begin{algorithm}[t]
\caption{Riemannian (sub)gradient descent algorithm}\label{alg:grad-descent}
\label{alg:gradient}
\begin{algorithmic}
\renewcommand{\algorithmicrequire}{\textbf{Input:}}
\renewcommand{\algorithmicensure}{\textbf{Output:}}
\Require~~\
observation $\Brac{\mb y_i}_{i=1}^m$
\Ensure~~\
the vector $\mb q_\star$, 
\State Precondition the data by $\ol{\mb y}_i = \mb y_i \conv \mb v $, with $\mb v = \paren{ \frac{1}{\theta np} \sum_{i = 1}^p \abs{\mb y_i}^{\odot 2} }^{\odot -1/2} $.
\State Initialize the iterate $\mb q^{(0)}$ and stepsize $\tau^{(0)}$.
\While{ not converged }
\State Update the iterate by
\begin{align*}
	\mb q^{(k+1)} \;=\; \mc P_{\bb S^{n-1}} \paren{ \mb q^{(k)} - \tau^{(k)}  \grad \varphi (\mb q^{(k)}) }.
\end{align*}
\State Choose a new stepsize $\tau^{(k+1)}$, and set $k \leftarrow k+1$.
\EndWhile
\end{algorithmic}
\end{algorithm}

\paragraph{Initialization.} In our theory, we showed that starting from a random initialization drawn uniformly over the sphere,
\begin{align*}
    \mb q^{(0)} \;=\; \mb d, \quad \mb d \sim \mc U(\bb S^{n-1}),
\end{align*}
for Huber-loss, Riemannian gradient descent method provably recovers the target solution. On the other hand, we could also cook up a data-driven initialization by choosing a row of $ \mb C_{\ol{\mb y}_i} $,
\begin{align*}
    \mb q^{(0)} \;=\; \mc P_{\bb S^{n-1}} \paren{ \mb C_{\ol{\mb y}_i}^\top  \mb e_j }
\end{align*}
 for some randomly chosen $1\leq i \leq p$ and $1\leq j\leq n$. By observing
 \begin{align*}
   	\mb C_{\ol{\mb y}_i}  \;\approx\; \mb C_{\mb x_i} \mb C_{\mb a} \paren{ \mb C_{\mb a}^\top \mb C_{\mb a}  }^{-1/2},\quad \mb q^{(0)} \;\approx\; \mc P_{\bb S^{n-1}}\paren{ 	\paren{ \mb C_{\mb a}^\top \mb C_{\mb a}  }^{-1/2} \mb C_{\mb a}^\top \shift{ \wc{\mb x}_i}{j} },
\end{align*}
we have 
\begin{align*}
     \mb C_{\ol{\mb y}_j}  \mb q^{(0)} \;\approx\; \alpha  \mb C_{\mb x_i} \mb C_{\mb a} (\mb C_{\mb a}^\top \mb C_{\mb a})^{-1} \mb C_{\mb a}^\top \shift{ \wc{\mb x}_i}{\ell} \;=\; \alpha  \mb C_{\mb x_j} \shift{ \wc{\mb x}_i}{\ell}.
\end{align*}
This suggests that our particular initialization $\mb q^{(0)}$ is acting like $\shift{ \wc{\mb x}_i}{\ell}$ in the rotated domain. It is sparse and possesses several large spiky entries more biased towards the target solutions. Empirically, we find this data-driven initialization often works better than random initializations.

\paragraph{Choice of stepsizes.} For Huber and $\ell^4$ losses, we can choose a fixed stepsize $\tau^{(k)}$ for all iterates to guarantee linear convergence. For subgradient descent of $\ell^1$-loss, it often achieves linear convergence when we choose a geometrically decreasing sequence of stepsize $\tau^{(k)}$ \cite{zhu2018dual}. Empirically, we find that the algorithm converges much faster when Riemannian linesearch is deployed (see \Cref{alg:linesearch-tau}).

\begin{algorithm}
\caption{Riemannian linesearch for stepsize $\tau$}
\label{alg:linesearch-tau}
\begin{algorithmic}
\renewcommand{\algorithmicrequire}{\textbf{Input:}}
\renewcommand{\algorithmicensure}{\textbf{Output:}}
\Require~~\
$\mb a$, $\mb x$, $\tau_0$, $\eta \in (0.5,1)$, $\beta \in (0,1)$,
\Ensure~~\
$\tau$, $\mc R_{\mb a}^{\mc M} \paren{ -\tau  \mb P_{T_{\mc M} } \nabla \psi_{\mb x}(\mb a)  }$
\State Initialize $\tau \leftarrow \tau_0$,
\State Set $\wt{\mb q} = \mc P_{\bb S^{n-1}} \paren{ \mb q - \tau  \grad \varphi (\mb q) } $
\While { $ \varphi( \wt{\mb q} ) \;\geq\;  \varphi(\mb q) - \tau \cdot \eta \cdot \norm{ \grad \varphi(\mb q) }{}^2	 $ }
\State $\tau \leftarrow \beta \tau $,
\State Update $\wt{\mb q} = \mc P_{\bb S^{n-1}} \paren{ \mb q - \tau  \grad \varphi (\mb q) } $.
\EndWhile
\end{algorithmic}
\end{algorithm}

\subsection{LP rounding}

Due to preconditioning or smoothing effects of our choice of loss functions, the Riemannian (sub)gradient descent methods can only produce an approximate solution. To obtain the exact solution, we use the solution $\mb r= \mb q_\star$ produced by gradient methods as a warm start, and solve another phase-two LP rounding problem,
\begin{align*}
   \min_{\mb q}\;  \zeta(\mb q):=\frac{1}{np} \sum_{i=1}^p \norm{ \ol{\mb y}_i \conv \mb q }{1} \quad\text{s.t.} \quad \innerprod{\mb r}{\mb q}= 1.	
\end{align*}
Since the feasible set $\innerprod{\mb r}{\mb q}= 1$ is essentially the tangent space of the sphere $\bb S^{n-1}$ at $\mb q_\star$, whenever $\mb q_\star$ is close enough to one of the target solutions, one should expect that the optimizer $\mb q_r$ of LP rounding exactly recovers the inverse of the kernel $\mb a$ up to a scaled-shift. To address this computational issue, we utilize a \emph{projected subgradient method} for solving the LP rounding problem. Namely, we take 
\begin{align*}
	\mb q^{(k+1)} \;&=\; \mb r + \paren{ \mb I - \mb r \mb r^\top } \paren{ \mb q^{(k)} - \tau^{(k)} \mb g^{(k)}  } \\
	\;&=\; \mb q^{(k)} - \tau^{(k)} \mc P_{\mb r^\perp} \mb g^{(k)},
\end{align*}
where $\mb g^{(k)}$ is the subgradient at $\mb q^{(k)}$ with
\begin{align*}
\quad \mb g^{(k)} \;=\; \frac{1}{np} \sum_{i=1}^p \wc{\ol{\mb y}}_i \conv \sign\paren{ \ol{\mb y}_i  \conv \mb q^{(k)} }.	
\end{align*}
By choosing a geometrically shrinking stepsizes
\begin{align*}
   \tau^{(k+1)} \;=\; \beta \tau^{(k)},\quad \beta \in (0,1).  	
\end{align*}
we show that the subgradient descent linearly converges to the target solution. The overall method is summarized in \Cref{alg:subgradient}.

\begin{algorithm}
\caption{Projected subgradient method for solving the LP rounding problem}
\label{alg:subgradient}
\begin{algorithmic}
\renewcommand{\algorithmicrequire}{\textbf{Input:}}
\renewcommand{\algorithmicensure}{\textbf{Output:}}
\Require~~\
observation $\Brac{\mb y_i}_{i=1}^m$, vector $\mb r$, stepsize $\tau_0$, and $\beta \in (0,1)$.
\Ensure~~\
the solution $\mb q_\star$, 
\State Precondition the data by $\ol{\mb y}_i = \mb y_i \conv \mb v $, with $\mb v = \paren{\frac{1}{\theta np} \sum_{i = 1}^p \abs{\mb y_i}^{\odot 2} }^{\odot -1/2} $.
\State Initialize $\mb q^{(0)} = \mb r$, $\tau^{(0)} = \tau_0 $
\While{ not converged }
\State Update the iterate
\begin{align*}
	\mb q^{(k+1)} \;=\; \mb q^{(k)} - \tau^{(k)} \mc P_{\mb r^\perp} \mb g^{(k)}.
\end{align*}
\State Set $\tau^{(k+1)} = \beta\tau^{(k)} $, and $k \leftarrow k+1$.
\EndWhile
\end{algorithmic}
\end{algorithm}

\subsection{Solving problems in 2D}

Finally, we briefly discuss about technical details about solving the MCS-BD problem in 2D, which appears broadly in imaging applications such as image deblurring \cite{levin2011understanding,zhang2013multi,sroubek2012robust} and microscopy imaging \cite{betzig2006imaging,hess2006ultra,rust2006sub}. 

\paragraph{Problem formulation.} Given the measurements
\begin{align*}
   \mb Y_i \;=\; \mb A \cconv \mb X_i, \quad 1\leq i\leq p,
\end{align*}
where $\cconv$ denotes 2D convolution, $\mb A\in \bb R^{n \times n}$ is a 2D kernel, and $\mb X_i \in \bb R^{n \times n}$ is a sparse activation map, we want to recover $\mb A$ and $\Brac{ \mb X_i }_{i=1}^p$ simultaneously. We first precondition the data via
\begin{align*}
   \ol{\mb Y}_i \; = \; \mb Y_i \cconv \mb V,\quad 	\mb V \;= \;	 \mc F^{-1} \paren{ \paren{\frac{1}{\theta n^2p } \sum_{i=1}^p \abs{\mc F(\mb Y_i)}^{\odot 2} }^{\odot -1/2} },
\end{align*}
where $\mc F(\cdot)$ denote the 2D DFT operator. By using the preconditioned data, we solve the following optimization problem
\begin{align}\label{eqn:problem-2D}
   \min_{ \mb Z }\; \varphi(\mb Z) \;:=\; \frac{1}{n^2p} \sum_{i=1}^p \psi (  \ol{\mb Y}_i \cconv \mb Z ),\quad \text{s.t.}\; \norm{\mb Z}{F} = 1,
\end{align}
where $\varphi(\cdot)$ is the loss function (e.g., $\ell^1$, Huber, $\ell^4$-loss), and $\norm{\cdot}{F}$ denotes the Frobenius norm. If the problem \eqref{eqn:problem-2D} can be solved to the target solution $\mb Z_\star$, then we can recover the kernel and the sparse activation map up to a signed-shift by
\begin{align*}
   \mb A_\star \;=\; \mc F^{-1} \paren{ \mc F\paren{ \mb V \cconv \mb Z_\star}^{\odot -1} },\qquad \mb X_i^\star \;=\;  \paren{\mb Y_i \cconv \mb V } \cconv \mb Z_\star,\;1\leq i\leq p.
\end{align*}

\paragraph{Riemannian (sub)gradient descent.} Similar to the 1D case, we can optimize the problem \eqref{eqn:problem-2D} via Riemannian (sub)gradient descent,
\begin{align*}
    \mb Z^{(k+1)} \;=\;  \mc P_F \paren{ \mb Z^{(k)} - \tau^{(k)} \cdot \grad \varphi(\mb Z^{(k)}) },
\end{align*}
where the Riemannian (sub)gradient
\begin{align*}
    \grad \varphi(\mb Z) \;=\; \mc P_{ \mb Z^\perp } \nabla \varphi(\mb Z).
\end{align*}
The gradient $\nabla \varphi(\mb Z)$ for different loss functions are recorded in \Cref{tab:loss-grad}. For any $\mb W\in \bb R^{n \times n}$, the normalization operator $\mc P_F(\cdot)$ and projection operator $\mc P_{\mb Z^\perp}(\cdot)$ are defined as
\begin{align*}
    \mc P_F(\mb W) \; :=\; \mb W / \norm{\mb W}{F},\quad \mc P_{\mb Z^\perp}(\mb W) \;:=\; \mb W - \norm{\mb Z}{F}^{-2} \innerprod{\mb Z}{\mb W} \mb Z.
\end{align*}
The initialization and stepsize $\tau^{(k)}$ can be chosen similarly as the 1D case.
\paragraph{LP rounding.} Similar to 1D case, we solve a phase-two linear program to obtain exact solution. By using the solution $\mb Z_\star$ produced by Riemannian gradient descent as a warm start $\mb U = \mb Z_\star$, we solve
\begin{align*}
    \min_{\mb Z}\; \frac{1}{n^2p} \sum_{i=1}^p \norm{ \ol{\mb Y}_i \cconv \mb Z }{1},\quad \text{s.t.}\; \innerprod{\mb U}{\mb Z} = 1.
\end{align*}
We optimize the LP rounding problem via subgradient descent,
\begin{align*}
   \mb Z^{(k+1)} = \mb Z^{(k)} - \tau^{(k)} \mc P_{ \mb U^\perp }\mb G^{(k)},
\end{align*}
where we choose a geometrically decreasing stepsize $\tau^{(k)}$ and set the subgradient
\begin{align*}
   	\mb G^{(k)}\;=\; \frac{1}{n^2p} \sum_{i=1}^p \wc{\ol{\mb Y}}_i \cconv \sign\paren{ \ol{\mb Y}_i \cconv \mb Z^{(k)} }.
\end{align*}